\newif\iflong
\newif\ifshort

\longtrue

\iflong
\else
\shorttrue
\fi

\def\year{2021}\relax
\documentclass{article} %
\usepackage{aaai21}  %
\usepackage{times}  %
\usepackage{helvet} %
\usepackage{courier}  %
\usepackage[hyphens]{url}  %
\usepackage{graphicx} %
 \usepackage[numbers]{natbib}
\usepackage{caption} %
\ifshort
\newcommand{\mytitle}{Fractional Matchings under Preferences: Stability and Optimality}
\else
\newcommand{\mytitle}{Fractional Matchings under Preferences: Stability and Optimality}
\fi
\title{\mytitle\thanks{Supported by the WWTF research grant~(VRG18-012).}}
\author{
  Jiehua Chen\textsuperscript{\rm 1}, Sanjukta Roy\textsuperscript{\rm 2}, Manuel Sorge\textsuperscript{\rm 1}\\
}

\affiliations {
    \textsuperscript{\rm 1} TU Wien, Austria \\
    \textsuperscript{\rm 2} Institute of Mathematical Sciences, India \\
    jiehua.chen@tuwien.ac.at, sanjuktaroy.mail@gmail.com, manuel.sorge@tuwien.ac.at
}

\copyrighttext{}
\usepackage[switch,mathlines,displaymath]{lineno}  %

\usepackage[pagebackref]{hyperref}
\hypersetup{%
  backref=true,
  pagebackref=true, 
  hypertexnames=true,
  colorlinks=true,citecolor=green!60!black,linkcolor=red!60!black%
}
\usepackage[textsize=tiny,textwidth=1.5cm,linecolor=green!70!black, backgroundcolor=green!10, bordercolor=black]{todonotes}

\usepackage{amsmath,amssymb,graphicx,amsthm,dsfont}
\usepackage{color,soul}
\usepackage{xspace}
\usepackage{enumerate}
\usepackage{expdlist}
\usepackage{mathtools}
\usepackage{bbm}

\usepackage{multicol}
\usepackage{booktabs}
\usepackage{paralist}

\usepackage{scrextend}

\usepackage[normalem]{ulem}

\usepackage[noend,ruled,linesnumbered]{algorithm2e}

\SetCommentSty{mycommfont}

\usepackage{xifthen}

\usepackage{mdframed}

\newcommand{\myemph}[1]{{\color{green!40!black}\emph{#1}}}
\newcommand{\myparagraph}[1]{
  \smallskip
  \noindent\textbf{#1}
}

\newcommand{\prob}[6]{%
  \smallskip

  \noindent
  \begin{tabular}{@{\;}l@{\,}p{0.4\textwidth}@{\;}}
    \multicolumn{2}{@{\;}l}{\textsc{#1}}\\
    \emph{#2} & #3\\
    \emph{#4} & #5\\
  \end{tabular}

  \smallskip

}

\newcommand{\probdef}[3]{\prob{#1}{Input:}{#2}{Question:}{#3}{as}}

\newcommand{\probdefopt}[3]{\prob{#1}{Input:}{#2}{Task:}{#3}{as}}

\newcommand{\appsymb}{$\star$}

\usepackage{tikz}
\usetikzlibrary{decorations,arrows,petri,topaths,backgrounds,shapes,positioning,fit,calc,decorations.pathreplacing,patterns,intersections,decorations.pathmorphing,matrix}

\tikzstyle{thickline} = [line width=1.8pt]
\tikzstyle{agent} = [circle, inner sep=1pt, draw, fill=black!60]

\tikzset{linemarkr/.style =   {line cap=round, opacity=.4, line width= 3pt, red}}
\tikzset{linemarkg/.style =   {line cap=round, opacity=.3, line width= 7pt, green}}
\tikzset{linemarkb/.style =   {line cap=round, opacity=.3, line width= 7pt, blue}}
\tikzset{linemarky/.style =   {line cap=round, opacity=.2, line width= 7pt, yellow}}

\makeatletter
\newcommand{\gettikzxy}[3]{%
  \tikz@scan@one@point\pgfutil@firstofone#1\relax
  \edef#2{\the\pgf@x}%
  \edef#3{\the\pgf@y}%
}
\makeatother
\newcommand{\npc}{\textsf{NP-c}}

\newcommand{\pp}{\textsf{P}}
\newcommand{\mycite}[1]{\small #1}
\newcommand{\newhardresult}[1]{\colorbox{red!60!black}{\parbox[b][1.5ex][c]{4.5ex}{\textcolor{white}{\small #1}}}} %
\newcommand{\newpresult}[1]{\colorbox{green!60!black}{\parbox[b][1.5ex][c]{4.5ex}{\textcolor{white}{%
        \small #1}}}} %

\newcommand{\knownresult}[1]{\parbox[b][1ex][c]{4.5ex}{\small #1}}

\newcommand{\ecresultsymbol}{$\clubsuit$}
\newcommand{\afresultsymbol}{$\heartsuit$}

\newcommand{\ecresultcite}[1]{{\knownresult{#1}}& [\ecresultsymbol]}

\newcommand{\IS}{\textsc{Independent Set}\xspace}
\newcommand{\ISs}{\textsc{IS}\xspace}
\newcommand{\MaxIS}{\textsc{Max Independent Set}\xspace}
\newcommand{\MaxISs}{\textsc{Max-IS}\xspace}

\newcommand{\EOSFM}{\textsc{Max-Welfare OSM}\xspace}
\newcommand{\POSFM}{\textsc{Max-Full OSM}\xspace}
\newcommand{\ECSFM}{\textsc{Max-Welfare CSM}\xspace}
\newcommand{\PCSFM}{\textsc{Max-Full CSM}\xspace}

\newcommand{\welfare}{\ensuremath{\gamma}}
\newcommand{\fullymatched}{\ensuremath{\#\mathsf{fully}}}
\newcommand{\fullnum}{\ensuremath{\tau}}
\newcommand{\lis}{\ensuremath{{\color{red!80!black}{\mathsf{IS}}}}}
\newcommand{\lvc}{\ensuremath{{\mathsf{VC}}}}

\newcommand{\bettere}[2]{\ensuremath{{\mathcal{BE}_{#1}\!(#2)}}}
\newcommand{\better}[2]{\ensuremath{{\mathcal{B}_{#1}\!(#2)}}}

\newcommand{\defqed}{$\diamond$}

\newcommand{\colU}[1]{\textcolor{blue!70!black}{#1}}
\newcommand{\colW}[1]{\textcolor{red!70!black}{#1}}

\usepackage{cleveref}

\newtheorem{theorem}{Theorem}[section]

\newtheorem{corollary}[theorem]{Corollary}
\newtheorem{lemma}[theorem]{Lemma}
\newtheorem{observation}[theorem]{Observation}
\newtheorem{proposition}[theorem]{Proposition}
\newtheorem{claim}{Claim}

\theoremstyle{definition}
\newtheorem{definition}[theorem]{Definition}
\newtheorem{example}[theorem]{Example}
\newtheorem{construction}{Construction}

\crefname{table}{Table}{Tables}
\crefname{figure}{Figure}{Figures}
\crefname{theorem}{Theorem}{Theorems}
\crefname{definition}{Definition}{Definitions}
\crefname{corollary}{Corollary}{Corollaries}
\crefname{observation}{Observation}{Observations}
\crefname{lemma}{Lemma}{Lemmas}
\crefname{example}{Example}{Examples}
\crefname{reduction}{Reduction}{Reductions}
\crefname{construction}{Construction}{Constructions}
\crefname{subsection}{Subsection}{Subsections}
\crefname{section}{Section}{Sections}
\crefname{proposition}{Proposition}{Propositions}
\crefname{algorithm}{Algorithm}{Algorithms}
\crefname{claim}{Claim}{Claims}
\crefname{rrule}{Reduction Rule}{Reduction Rules}

\newcommand{\perfect}{perfect\xspace}

\newcommand{\sat}{\ensuremath{\mathsf{sat}}}

\newcommand{\rank}{\ensuremath{\mathsf{rank}}}
\newcommand{\egal}{\ensuremath{\mathsf{egal}}}
\newcommand{\Pot}{\ensuremath{\mathcal{P}}\xspace}
\newcommand{\LLAST}{\ensuremath{\textcolor{red!60!black}{\mathsf{last}}}}
\newcommand{\ffirst}{\ensuremath{\textcolor{green!50!black}{\mathsf{1st}}}}

\newcommand{\wel}{\ensuremath{\mathsf{welfare}}}

\newcommand{\util}{\ensuremath{\mathcal{U}}}

\newcommand{\cutility}{utility\xspace}

\newcommand{\Utilities}{Utilities\xspace}

\newcommand{\R}{\ensuremath{\mathds{R}}}
\newcommand{\Q}{\ensuremath{\mathds{Q}}}

\newcommand{\pref}{\ensuremath{{\succ}}}
\newcommand{\wpref}{\ensuremath{{\succeq}}}
\newcommand{\indif}{\ensuremath{\sim}}

\newcommand{\meet}{\ensuremath{{\wedge}}}
\newcommand{\join}{\ensuremath{{\vee}}}
\newcommand{\med}{\ensuremath{\mathsf{med}}}

\newcommand{\cstable}{cardinally stable\xspace}
\newcommand{\ostable}{ordinally stable\xspace}

\newcommand{\lstable}{linearly stable\xspace}
\newcommand{\csm}{\textcolor{red!60!black}{CSM}\xspace}
\newcommand{\osm}{\textcolor{blue!80!black}{OSM}\xspace}
\newcommand{\lsm}{LSM\xspace}
\newcommand{\csms}{\ensuremath{\mathsf{C}}\xspace}
\newcommand{\osms}{\ensuremath{\mathsf{O}}\xspace}
\newcommand{\lsms}{\ensuremath{\mathsf{L}}\xspace}

\newcommand{\cblocking}{cardinally blocking\xspace}
\newcommand{\oblocking}{ordinally blocking\xspace}
\newcommand{\lblocking}{linearly blocking\xspace}

\newcommand{\cardinalstability}{cardinal stability\xspace}
\newcommand{\ordinalstability}{ordinal stability\xspace}
\newcommand{\linearstability}{linear stability\xspace}

\newcommand{\NPH}{\text{{\normalfont NP}\nobreakdash-hard}\xspace}
\newcommand{\APXH}{\text{{\normalfont APX}\nobreakdash-hard}\xspace}
\newcommand{\NPC}{\text{{\normalfont NP}\nobreakdash-complete}\xspace}
\newcommand{\NP}{\text{\normalfont NP}\xspace}
\newcommand{\WOH}{\text{{\normalfont W[1]}\nobreakdash-hard}\xspace}

\newcommand{\agent}{\xspace}

\begin{document}
\maketitle

\begin{abstract}
  \looseness=-1
  We thoroughly study a generalized version of the classic Stable Marriage and Stable Roommates problems where agents may share partners. %
  We consider two prominent stability concepts: \ordinalstability~\cite{AF03Scarfs} and \cardinalstability~\cite{CFKV19}, and two optimality criteria:
  maximizing social welfare (i.e., the overall satisfaction of the agents) and
  maximizing the number of fully matched agents (i.e., agents whose shares sum up to one).
  After having observed that \ordinalstability always exists and implies \cardinalstability,
  and that the set of \ostable matchings in a restricted case admits a lattice structure,
  we obtain a complete picture regarding the computational complexity of finding an optimal \ostable or \cstable matching.
  In the process we answer an open question raised by \citet{CFKV20j}.
\end{abstract}

\section{Introduction}

{\centering ``A joy shared is a joy doubled!''\par}

This is particularly prevalent in matching markets,
where the market participants, jointly referred to as \myemph{agents},
have preferences over whom they want to have as partner. %
The goal is to match agents with partners so as to achieve some desirable properties, such as \myemph{stability}, i.e., no two agents would like to deviate from their current assignments under the matching. %
In its most simple form, a \myemph{matching} consists of disjoint pairs of agents, meaning that each agent is assigned to at most one other agent; we call such matchings \myemph{integral matchings}.
A \myemph{stable integral matching} is an integral matching where no two agents would prefer to be matched to each other rather than with their assigned partners, if any.

\looseness=-1
Unfortunately, a stable integral matching does not always exist.
If however the agents are allowed to \emph{share} partners, i.e., to have a \myemph{fractional} matching,
then the social welfare may increase and stability is guaranteed! %
Here, a \myemph{fractional} matching is a function %
which assigns each pair of agents %
a value %
between zero and one such that for each agent, the sum of the values of all pairs containing this agent is at most one. %
An integral matching is hence a restricted variant of fractional matchings where each fraction is either zero or one.

Fractional matchings have applications in time-sharing. %
{For example, in a job market the agents may be partitioned into two sets, freelancers and companies.
  A fractional matching models the amount of time a freelancer spends working for a company.
  The preferences can model intensity of interest in working with the agents of the other set,
  and then \emph{stability} models an equilibrium in such a job market.
  Similar scenarios are time-sharing assignments between advisors and apprentices or between workers and projects.
  An instance of the non-bipartite case occurs when agents (e.g., nurses) work in multiple shifts, and each shift is carried out by two workers.
  A fractional matching determines the fraction of shifts that each worker carries out with another worker.
  The preferences can model the intensity of willingness to work with each other,
  and then \emph{stability} models the situation where no workers want to swap shifts. 
  Fractional matchings also find application in random matching~\cite{RothRothblumVate1993lsm-lattice,aziz_random_2019}: By the Birkhoff-von Neumann theorem a fractional matching can be interpreted as a probability distribution over integral matchings in the bipartite case.
  Choosing an integral matching at random instead of deterministically enables many desirable properties such as fairness and increased expected welfare.}

There are multiple natural ways to extend the notion of stability for integral matchings to fractional matchings.
For an illustration, let us consider the following example with six agents, called~$a,b,c,d,e,f$ as shown in Figure~\ref{fig:example_intro}.

 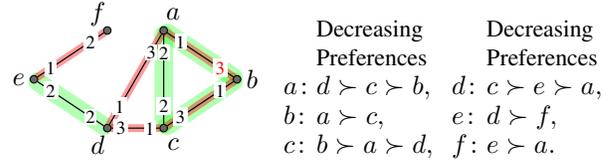
\begin{figure}[t]
\begin{minipage}{.4\linewidth}
  \begin{tikzpicture}
    \foreach \n / \i / \r / \s in {a/1/.75/1.0,b/2/1.35/1.55,c/3/.75/1.0,d/4/.75/1.0,e/5/1.35/1.55,f/6/.75/1.0} {
      \node[agent] at (\i*-60+120:\r) (\i) {};
      \node[] at (\i*-60+120:\s) {$\n$};
    }

    \foreach \s / \t / \w / \v in {1/2/1/{\textcolor{red}{3}},1/3/2/2,1/4/{{3}}/1,%
      2/3/1/3,3/4/1/{{3}},4/5/2/2,5/6/1/2} {
      \path[draw] (\s) -- node[pos=0.2, fill=white, inner sep=.5pt] {\scriptsize \w} node[pos=0.8, fill=white, inner sep=1pt] {\scriptsize \v} (\t);
    }
    \begin{pgfonlayer}{background}
      \foreach \i / \j in {1/2,2/3,3/1,4/5} {
        \draw[linemarkg] (\i) edge (\j);
        \draw[gray] (\i) edge (\j);
      }

      \foreach \i / \j in {1/2,2/3,3/4,4/1,5/6} {
        \draw[linemarkr] (\i) edge (\j);
        \draw[gray] (\i) edge (\j);
      }
    \end{pgfonlayer}
  \end{tikzpicture}
\end{minipage}
\begin{minipage}{.55\linewidth}
  \[\begin{array}{l@{}l@{\,\;}l@{}l}
      & \mbox{\small Decreasing}  & & \mbox{\small Decreasing}\\
      & \mbox{\small Preferences}  & & \mbox{\small Preferences}\\
      \agent~a\colon & d\succ c\succ b,  &   \agent~d\colon & c\succ e\succ a,\\
      \agent~b\colon & a\succ c, &   \agent~e\colon & d\succ f,\\
      \agent~c\colon & b\succ a \succ d, &        \agent~f\colon & e \succ a.
    \end{array} \]
\end{minipage}
\caption{{Left: An acceptability graph with six agents $a$, $b$, $c$, $d$, $e$, and $f$. The values on an edge denote the satisfactions of the endpoints of the edge towards each other, e.g., $a$'s satisfaction towards $b$ is $1$ and $b$'s satisfaction towards $a$ is $3$.
Right: The ordinal preferences derived from the satisfaction values.}}\label{fig:example_intro}
\end{figure}

The preference of an agent~$i$ towards another agent~$j$ is specified through a non-negative cardinal value (the higher the better), called \myemph{satisfaction},
and is depicted at the end of edge~$\{i,j\}$ closer to~$i$ in the graph on the left.
For instance,~$b$'s satisfactions towards~$a$ and $c$ are~$3$ and $1$, respectively.
This means that $b$ \myemph{prefers}~$a$ to~$c$, expressed as $b\colon a\pref c$.

In this example, no integral matching is stable due to the \emph{cyclic} preferences of the three agents~$a,b,c$:
No matter how an integral matching looks like, at least two of the three would prefer to be with each other rather than with the assignment by the matching. 
Indeed, odd cycles with such \iflong kinds of \fi cyclic preferences are the main obstruction to having a stable integral matching.
However, in practice, odd cycles are rather the norm as social networks often have large clustering coefficients which essentially means that it is likely for three agents to form a triangle~%
\mbox{\cite{newman_networks_2010}}.
Thus it is not far-fetched to suppose that odd cycles with cyclic preferences are likely in matching markets and hence, no stable integral matchings exist.

For fractional matchings the situation is different.
{Consider the green fractional matching~\myemph{$M$} in Figure~\ref{fig:example_intro} (indicated by the green edges):} %
Agents~$a,b,c$ in the triangle are half-integrally matched with each other~(i.e., each of the three pairs receives a half-integral value:~$0.5$), and $d$ is integrally matched with $e$. %
 This green matching~\myemph{$M$} is \myemph{\cstable}~\cite{CFKV19}, i.e., no two agents could increase their utilities by being integrally matched with each other.
Herein, the \myemph{utility} of an agent towards a fractional matching is the sum of her satisfactions towards her potential partners weighted by the respective fractional matching value.
The utilities of agents~$a$, \(b\), and~\(c\) are $1 \cdot 0.5 + 2 \cdot 0.5= 1.5$, %
$(1+3)\cdot 0.5=2$, and $(3+2)\cdot 0.5=2.5$, respectively. 

The green matching~\myemph{$M$}  (shown by the green lines in Figure~\ref{fig:example_intro}) satisfies two more fractional stability concepts which were originally defined for ordinal preferences~(see the right hand side of \cref{fig:example_intro}).
More precisely, the green matching~\myemph{$M$} is \myemph{\ostable}~\cite{AF03Scarfs}, i.e.,
for each pair of agents at least one agent in the pair is \emph{satisfied with~$M$ regarding the pair}.
Herein, \emph{an agent~$i$ is satisfied with a matching regarding a pair~$\{i,j\}$} 
if the fractional values assigned by the matching between~$i$ and someone she finds better or equal to~$j$ sum up to one. 
{For instance, agent~$a$ is satisfied with the green matching~\myemph{$M$} regarding~$\{a,b\}$ since she prefers $c$ to $b$ and the values of matching her to $c$ and $b$ sum up to one.
 However, she is not satisfied with~\myemph{$M$} when regarding~$\{a,c\}$ since the values of matching her to someone better or equal to $c$ sum up to $0.5$ which is less than one.}
Ordinal stability models the desired property that no two agents exist who both can increase the fractional value of matching them,
by possibly decreasing the fractional values of matching either of them to someone less preferred.

Lastly, the green matching~\myemph{$M$} is also \myemph{\lstable}~\cite{RothRothblumVate1993lsm-lattice,abeledo_stable_1994}, meaning that each pair of agents is \myemph{jointly satisfied with~$M$}.
A pair~$\{i,j\}$ is \myemph{jointly satisfied with a matching}
if the fractional values of matching~$i$ to someone better or equal to~$j$ plus the fractional values of matching~$j$ to someone better or equal than $i$ sum up to at least one. %
 For instance, under the green matching, for pair~$\{a,c\}$ the sum is $0.5+0.5=1$, and hence $\{a,c\}$ is jointly satisfied with the green matching.

For inclusivity we may strive to maximize the total matching values. %
Indeed, there is another \cstable\ (fractional) matching~\textcolor{red!85!black}{$M'$}, indicated by the red lines in the graph in Figure~\ref{fig:example_intro}, where everyone is \myemph{fully matched}, i.e., the matching values for each agent sum up to one.
In matching~\textcolor{red!85!black}{$M'$}, we match both $a$ and $c$ each half-integrally with both~$d$ and $b$,
and match $e$ with~$f$ integrally.
Matching~\textcolor{red!85!black}{$M'$} is, however, neither \ostable nor \lstable since $\{d,e\}$ is not jointly satisfied with~\textcolor{red!85!black}{$M'$}, implying also that neither $d$ nor $e$ is satisfied with~\textcolor{red!85!black}{$M'$} regarding the pair~$\{d,e\}$.

\looseness=-1 
In terms of social welfare, defined as the sum of the utilities of all agents,
the red matching~\textcolor{red!85!black}{$M'$} has a welfare of~$11$, making it superior to the green matching~\myemph{$M$}, with a welfare of~$10$.
Indeed, matching~\textcolor{red!85!black}{$M'$} has achieved maximum-welfare since this value is the maximum that any matching of the corresponding edge-weighted graph can achieve; here the weight of an edge representing a pair~$\{i,j\}$ is equal to the sum of the satisfactions of $i$ and $j$ towards each other. 

\myparagraph{Our contribution.}
It is fairly straightforward to see that when restricted to integral matchings all three stability concepts coincide with the classical (weak) stability concept. 
Aiming for a better understanding of fractional matchings under preferences,
in the first part of the paper
we take a structural approach to study how the three stability concepts (\cardinalstability, \ordinalstability, and \linearstability) relate to each other.
In the second part, we focus on computing stable fractional matchings that maximize the number of fully matched agents or the social welfare.
Since \linearstability can be formulated via linear programs, finding an optimal \lstable matching can be solved in polynomial time whenever the objective can be formulated as a linear function of the matching values. %
Hence, we focus on the other two stability concepts.
We investigate how the complexity of finding an optimal stable fractional matching is influenced by
\ifshort\begin{inparaitem}\else \begin{compactitem}
  \item the presence of \myemph{ties} (i.e., an agent may have the same satisfaction towards different agents) and
  \item the type of matching market (i.e., in a \myemph{marriage market} the agents are divided into two disjoint parts such that all agents in one part have preferences over a subset of agents in the other part whereas in a \myemph{roommates market} there is no such division).
\ifshort\end{inparaitem}\else \end{compactitem} 
\ifshort
We highlight our findings below; also see  \cref{tab:results}.
\else
We highlight our findings below.
\fi

\begin{table*}[t]
  \centering
  \resizebox{\textwidth}{!}
  {
    \begin{tabular}{@{}l@{\,}|
    @{\,}p{.85cm}@{\,}p{.9cm}@{\;\;}p{.85cm}@{\,}p{.9cm}@{\,}|@{}c@{\,}p{.9cm}@{\;}p{.85cm}@{\;\;}p{.9cm}@{\,}p{.85cm}@{\,}c@{\,}|
    @{\,}p{.9cm}@{}p{.85cm}@{\;\;}p{.9cm}@{}p{.85cm}@{\,}|@{}c@{\,}p{.9cm}@{\;}p{.85cm}@{\;\;}p{.9cm}@{\,}p{.85cm}@{}}
  \toprule
  & \multicolumn{9}{c}{\cardinalstability} &  & \multicolumn{9}{c}{\ordinalstability}\\\cline{2-10}\cline{12-20}
  & \multicolumn{4}{c}{Marriage} &\multicolumn{1}{@{}c@{}}{} & \multicolumn{4}{c}{Roommates} & & \multicolumn{4}{c}{Marriage} &\multicolumn{1}{@{}c@{}}{} & \multicolumn{4}{c}{Roommates}  \\\cline{2-5}\cline{7-10}\cline{12-15}\cline{17-20}
    & \multicolumn{2}{c}{no ties} & \multicolumn{2}{c}{ties} & & \multicolumn{2}{c}{no ties}& \multicolumn{2}{c}{ties}&&
                                                                                                          \multicolumn{2}{c}{no ties} & \multicolumn{2}{c}{ties} & & \multicolumn{2}{c}{no ties}& \multicolumn{2}{c}{ties}\\\midrule
    always exists? & \knownresult{yes} & \mycite{[\ecresultsymbol]} & \knownresult{yes} & \mycite{[\ecresultsymbol]} && \newpresult{yes} & \mycite{[L~\ref{osroommates:always}]}  &\newpresult{yes} & \mycite{[L~\ref{osroommates:always}]} & & 
       \knownresult{yes} & \mycite{[\afresultsymbol]} & \newpresult{yes} & \mycite{[L~\ref{osroommates:always}]} && \knownresult{yes} & \mycite{[\afresultsymbol]} & \newpresult{yes}& \mycite{[L~\ref{osroommates:always}]}\\\hline
   max-\#-fully-matched %
  & \newhardresult{\npc} & \mycite{[T~\ref{th:PCSFM_hard}]} & \newhardresult{\npc} & \mycite{[T~\ref{th:ECSFM_hard}]} & & \newhardresult{\npc} & \mycite{[T~\ref{th:PCSFM_hard}]} & \newhardresult{\npc} & \mycite{[T~\ref{th:ECSFM_hard}]} &&  \newpresult{\pp} & \mycite{[L~\ref{osm:noties+perfect:p}]}
  & \newhardresult{\npc}  & \mycite{[T~\ref{th:Perfect-Welfare-ties-OSM_hard}]} & & \newpresult{\pp} & \mycite{[L~\ref{osm:noties+perfect:p}]}  & \newhardresult{\npc}  & \mycite{[T~\ref{th:Perfect-Welfare-ties-OSM_hard}]}\\
   max-welfare 
  &  \newhardresult{\npc} & \mycite{[T~\ref{th:ECSFM_hard}]} & \ecresultcite{\npc}  &&  \newhardresult{\npc} & \mycite{[T~\ref{th:ECSFM_hard}]} & \ecresultcite{\npc} && \newpresult{\pp} & \mycite{[T~\ref{thm:egal-bipartite-poly}]} & \newhardresult{\npc} & \mycite{[T~\ref{th:Perfect-Welfare-ties-OSM_hard}]} & & \newhardresult{\npc} & \mycite{[T~\ref{th:Welfare-OSR_hard}]} & \newhardresult{\npc} & \mycite{[T~\ref{th:Welfare-OSR_hard}]}  \\
  \bottomrule
\end{tabular}}
\iflong \caption{Complexity Results of known and new results for deciding perfect (resp.\ max-welfare) \cstable (\ostable) fractional matchings. Results marked with \ecresultsymbol\ are from \citet{CFKV19}. Results marked with \afresultsymbol\ are from \citet{AF03Scarfs}. Results marked in red and green are derived from the current paper. For all results regarding the existence of stable fractional matchings, finding such one can be done in polynomial time.}
\else
\caption{New and known results for deciding a
  \cstable (resp.\ \ostable) matching with max.\ \# of fully matched agents (resp.\ max-welfare).
  Results marked with \ecresultsymbol\ and \afresultsymbol\ are from  \citet{CFKV19} and \citet{AF03Scarfs}, respectively. Results marked in red and green are new; green means polynomial-time algorithms and red NP-completeness.}\fi
\label{tab:results}
\end{table*}%

\begin{compactenum}[(1)]
  \item Among the three stability concepts, \ordinalstability is the most stringent one since it implies both \cardinalstability and \linearstability, even in roommates markets, whereas the latter two do not necessarily imply each other.
  Similar to the \linearstability for strict preferences,
  in the marriage case the set of \ostable matchings admits a distributed lattice, %
  and in the roommates case this set is closed under a median operation~(see \cref{sec:structural}). %
  \item \looseness=-1 We introduce the problem of finding an \ostable or \cstable matching maximizing the number of fully matched agents.
  We show that for \ordinalstability, ties make a difference: It is polynomial-time solvable when ties are not present, and NP-hard otherwise.
\iflong {For the marriage case, the tractability result in this dichotomy comes from the fact that each \ostable\ matching is a convex combination of integral stable matchings and hence techniques for integral stable matchings can be applied.}
  
  For \cardinalstability, it is NP-hard even for preferences without ties and for the marriage case. 
  \else
  For \cardinalstability, it is NP-hard even in the marriage case without ties.
  \fi
  \item For maximizing the social welfare, the problem is mostly NP-hard, with only one exception: Finding a maximum-welfare \ostable matching for preferences without ties and the marriage case is polynomial-time solvable; the other cases remain NP-hard.  

  {Note that the hardness result for \cardinalstability behind \cref{th:ECSFM_hard} (also see the remark afterwards) is in stark contrast to the usual understanding of marriage problems without ties, for which most problems are solvable in polynomial time.
    Moreover, the result resolves an open question asked by Caragiannis et al.~\cite{CFKV20j}.}
\end{compactenum}
\ifshort Due to paucity of space we defer the proofs of the statements marked with (\appsymb) and statements of some known results to an appendix. %
\fi
\iflong
Our results are summarized in \cref{tab:results}.
\fi

\myparagraph{Related work.}
\looseness=-1
\citet{RothRothblumVate1993lsm-lattice} studied \linearstability~(they called it \emph{fractional stability}) in marriage markets without ties, and showed that the set of \lstable matchings enjoys a lattice structure.
\citet{abeledo_stable_1994} also studied \linearstability, but in roommates markets.
They observed that \linearstability in roommates markets does not have a lattice structure in general, but showed that \lstable matchings are closed under the so-called median operation.
Following \citeauthor{RothRothblumVate1993lsm-lattice,abeledo_stable_1994}, we show that the same results hold also for \ordinalstability.

\citet{aziz_random_2019} considered multiple fractional stability concepts in marriage markets, including \linearstability and \ordinalstability (which they called \emph{fractional stability} and \emph{ex-ante stability}, respectively), but not \cardinalstability.
They showed that \ordinalstability implies \linearstability.
We strengthen their result by showing the same for the roommates case.

\citet{CFKV19} introduced the problem of finding maximum-welfare \cstable matchings in marriage markets.
They showed that the problem is NP-hard and hard to approximate even if each agent has at most three different satisfaction values but may contain ties in her preferences.
We improve on this result by showing NP-hardness even when no ties are present and each agent finds at most five agents acceptable.
A subset of the structural results, namely the ones about \cardinalstability{} in the marriage setting and for perfect matching (see \cref{obs:ostable->cstable}) has been observed independently in parallel in a recent journal version~\cite[Appendix A]{CFKV20j} of this paper~\cite{CFKV19}.

Finally, \citet{AF03Scarfs} studied \ordinalstability in the hypergraphic setting where each agent~$i$ has \emph{strict} preferences over subsets (hyperedges) of agents which contain~$i$, and a fractional matching is a function that gives each hyperedge a non-negative fractional value such that the sum of values of the hyperedges incident to each agent is at most one.
They elaborated that the powerful Scarf lemma from game theory guarantees the existence of \ostable matchings.
However, \citeauthor{KiPRST2013osmhypegraphic}~\cite{KiPRST2013osmhypegraphic} and \citet{IK2018osm-hypergraphic} showed that finding an \ostable matching in the hypergraphic setting is as hard as finding a Nash equilibrium (PPAD-hard), even when each agent finds only a constant number of hyperedges acceptable.

\looseness=-1
For an overview on integral stable matchings, we refer to the books of \citet{GusfieldIrving1989} and \citet{Manlove2013}.%

\section{Preliminaries}\label{sec:prelim}
Given an integer~$z$, we use~\myemph{$[z]$} to denote the set $\{1,2,\ldots, z\}$.%

\myparagraph{Graphs with cardinal preferences, and matchings.}
Let $G$$=$$(V, E)$ be a graph and $\sat\colon V\times V \to \Q_{\geq 0}$ be a function, 
where
\iflong
\begin{compactitem}[--]
\else \begin{inparaitem} \fi
  \item \myemph{$V$} denotes a set of vertices (also called \myemph{agents}),
  \item \myemph{$E$} denotes a set of edges such that an edge between two vertices means that the corresponding agents find each other \myemph{acceptable},
  and
  \item \myemph{$\sat$} specifies the \myemph{cardinal preferences} \iflong (also called \myemph{satisfaction}) \fi of an agent towards another agent, i.e., for all~$u,v \in V$, \iflong the value~\fi$\sat(u,v)$ specifies the \myemph{satisfaction} of~$u$ towards~$v$.
\iflong \end{compactitem}
\else \end{inparaitem}\fi

\smallskip

\noindent \textbf{Remarks.} We assume throughout that
\begin{inparaenum}[(1)]
  \item $G$ contains no isolated vertices,
  \item \iflong for all~$u\in V$ it holds that $\sat(u,u)=0$, and
  \else $\forall u\in V\colon \sat(u,u)=0$, and \fi
  \item \iflong for all~$u,v\in V$ it holds that $\{u,v\} \in E$ if and only if ``$\sat(u,v) > 0$ or $\sat(v,u) > 0$''.\else
  $\forall u,v\in V\colon \{u,v\} \in E \Leftrightarrow (\sat(u,v) > 0 \vee \sat(v,u) > 0)$.\fi
\end{inparaenum}
\iflong

\medskip
\fi
\looseness=-1
From the satisfaction function~$\sat$ of~$G$ we derive a \myemph{preference list}~$\wpref_v$ over the \myemph{neighborhood} $N_G(v)=\{u \mid \{v,u\}\in E\}$ of each agent~$v\in V$ as follows:
Let \myemph{$\wpref_v$} denote a complete and transitive binary relation of~$N_G(v)$ %
such that for each two agents~$x,y$ with $x,y\in N_G(v)$ it holds that
\ifshort $x\wpref_v y$ if and only if $\sat(v,x)\ge \sat(v,y)$; we say that $v$ \myemph{weakly prefers} $x$ to~$y$.\fi
\iflong \begin{align}
\nonumber x\wpref_v y&\text{ if and only if }\sat(v,x)\ge \sat(v,y);\text{ we say that}\\ &v ~\myemph{weakly prefers} ~x~to~y.\hfill\tag{PREF} \label{eq:Pot}
\end{align}
\fi
We use \myemph{$\pref_v$} to denote the asymmetric part of~$\wpref_i$ (i.e., $\sat(v,x) > \sat(v,y)$), meaning that $v$ \myemph{(strictly) prefers} $x$ to $y$,
and \myemph{$\indif_i$} to denote the symmetric part of $\wpref_i$ (i.e., $\sat(v,x) = \sat(v,y)$), meaning that
$x$ and $y$ are \myemph{tied} by~$v$.
We use \myemph{$\Pot=(\wpref_v)_{v\in V}$} to denote the collection of the preference lists derived from~$\sat$.
We say that $x$ is a \myemph{most preferred} agent of~$v$ if for each agent~$y\in N_G(v)$ we have $x \wpref_v y$.

For each two agents~$u,v\in V$, we use \myemph{$\bettere{u}{v}$} (resp.\ \myemph{$\better{u}{v}$}) to denote the set of agents that $u$ weakly prefers (resp.\ strictly prefers) over~$v$,
i.e., $\bettere{u}{v} \coloneqq \{w\in V \setminus \{u\}\mid w\wpref_uv\}$, and 
$\better{u}{v} \coloneqq \{w\in V \setminus \{u\}\mid w\pref_uv\}$. %

\looseness=-1
An instance~$I=(G,\sat)$ contains \myemph{(preferences with) ties} if there exists $v\in V$ and two neighbors~$x,y\in N_G(v)$ with $\sat(v,x)=\sat(v,y)$; otherwise it has \myemph{strict preferences}.
We extend the standard integral matching concept to fractional ones.
\iflong \begin{definition}[Fractional Matching] \fi
  A \myemph{fractional matching}~$M \colon E \rightarrow \R_{\geq 0}$ is an assignment of non-negative weights to each edge~$e\in E$ such that
  $\sum_{\{v,u\} \in E}  M(\{u,v\}) \leq 1$
  for each agent~$v \in V$.
\iflong \hfill \defqed\end{definition}\fi
For the sake of readability and when there are no ambiguities, we abbreviate ``fractional matchings'' to ``matchings''.
To ease notation, for each edge~$\{u,v\}$ we use \myemph{$M(x,y)$} and \myemph{$M(y,x)$} to refer to the matching value~$M(\{x,y\})$.
Moreover, for each two agents~$x,y$, we use~\myemph{$M(x,\textcolor{red!60!black}{\wpref}y)$} and \myemph{${M}(x,\textcolor{red!60!black}{\pref}y)$} to denote the following sums:
\begin{align*}
  {M}(x,\textcolor{red!60!black}{\wpref} y) \coloneqq  \sum_{\mathclap{y'\in \textcolor{red!60!black}{\bettere{x}{y}}}}{M(x,y')}\text{, and } {M}(x,\textcolor{red!60!black}{\pref} y) \coloneqq  \sum_{\mathclap{y'\in \textcolor{red!60!black}{\better{x}{y}}}}{M(x,y')}.
\end{align*}

\noindent
\iflong A fractional matching~$M$ may satisfy one of the following properties: \fi
  An agent~$v$ is called \myemph{fully matched} (resp.\ \myemph{matched}) under~$M$
  if $\sum_{u \in N_G(v)} M(v,u)=1$ (resp.\ $\sum_{u \in N_G(v)} M(v,u)> 0$).
  $M$ is called \myemph{\perfect} %
  if each agent is fully matched. %
  $M$ is called \myemph{integral} (resp.\ \myemph{half-integral}) if $M(e)\in \{0,1\}$
   (resp.\ $M(e)\in \{0,0.5,1\}$) for each edge~$e$.

   As noted by \citet[p.\ 226]{aziz_random_2019}, by the Birkhoff-von Neumann theorem a fractional matching~$M$ in a bipartite graph can be decomposed into a convex combination of integral matchings~\cite[Theorem~3.2.6]{horn_topics_1991}. %
   \iflong
   (The bound \(k \in O(n^2)\) below follows from the fact that Theorem~3.2.6 in \cite{horn_topics_1991} indeed shows that the each fractional matching is contained in an \(O(n^2)\)-dimensional polyhedron together with Carath\'eodory's Theorem about convex hulls.)
   \fi
\iflong \begin{proposition}\label{prop:marriage:support}
For each fractional matching~$M$ of a bipartite graph~$G$ over $n$~vertices, there exists an integer~$k\in O(n^2)$, positive coefficients~$x_1,x_2,\ldots, x_k \in \mathds{R}_{>0}$, and $k$~integral matchings~$M_1,M_2,\ldots, M_k$ of $G$ such that $\sum_{j\in[k]} x_j = 1$ and for each edge~$e \in E$ it holds that
\begin{align*}
  M(e) = \sum_{j\in [k]} x_j \cdot M_j(e). 
\end{align*}
\end{proposition}
\else
\begin{proposition}\label{prop:marriage:support}
For each fractional matching~$M$ of a bipartite graph~$G$ over $n$~vertices, there exists an integer~$k\in O(n^2)$, positive coefficients~$x_1,x_2,\ldots, x_k \in \mathds{R}_{>0}$, and integral matchings~$M_1,M_2,\ldots, M_k$ of $G$ such that $\sum_{j\in[k]} x_j = 1$ and for each edge~$e \in E$ it holds that $M(e) = \sum_{j\in [k]} x_j \cdot M_j(e)$. 
\end{proposition}
\fi
\noindent\looseness=-1
The integral matchings $(M_j)_{j \in [k]}$ constitute a \myemph{support} of the matching~$M$.
There may be multiple supports of~\(M\).

\myparagraph{Three stability concepts wrt.\ fractional matchings.}
\ifshort \begin{definition}\label{def:util-bp-stability}
  \else \begin{definition}[\Utilities, blocking pairs, and stability]\label{def:util-bp-stability}
    \fi
    Let $G$ be a graph with cardinal preferences~$\sat$. %
    The \myemph{\cutility} of each agent~$v\in V$ under a matching~$M$ of $(G,\sat)$
    is defined as \myemph{$\util_{\sat, M}(v) \coloneqq \sum_{\{v,u\}\in E(G)} \sat(v,u)\cdot M(v,u)$}.
    If $\sat$ is clear from the context, we omit it from $\util_{\sat, M}$.
    
    \noindent  Given a matching~$M$ of~$(G,\sat)$, an edge~$\{u,v\} \in E(G)$ is %
  \begin{compactitem}[--]
    \item a \myemph{\cblocking pair} \iflong(or \myemph{\cblocking edge}) \fi if 
    {$\util_M(u) <\sat(u,v)$ and  $\util_M(v) < \sat(v,u)$;}
  \item an \myemph{\oblocking pair} \iflong (or \myemph{\oblocking edge}) \fi if
  {$M(u,\wpref v) < 1$ and 
  $M(v,\wpref u) < 1$; }
  \item a \myemph{\lblocking pair} \iflong (or \myemph{\lblocking} edge) \fi
  if
    $M(u,\wpref v) \!+\!  M(v, \wpref u) \!-\! M(u, v) < 1$.
\end{compactitem}
A matching~$M$ of $(G,\sat)$ is \myemph{\cstable}, \myemph{\ostable}, and \myemph{\lstable} if it contains no \cblocking pairs, no \oblocking pairs, and no  \lblocking pairs, respectively.
The acronyms \myemph{\csm}, \myemph{\osm}, and \myemph{\lsm} stand for {\cstable}, \ostable, and \lstable fractional matching, respectively.\hfill \defqed
\end{definition}

\iflong
\noindent \textbf{Remark.}
\fi Note that for integral matchings, all three stability concepts are equivalent to the classical stability concept.
\iflong
\smallskip

\fi
To illustrate the three stability concepts, consider the following.

\begin{example}\label{ex:cstable-not-always-ostable}
  Take the following bipartite graph on vertices $U\cup W$, $U=[5]$, $W=\{a,b,\ldots,e\}$, with strict preferences.
  
  {
  \centering
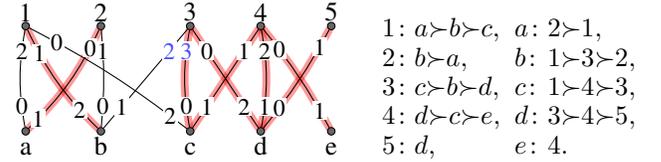
  \begin{minipage}{.55\linewidth}
    \begin{tikzpicture}
      \def \xs {1.25}
      \def \os {.7}
      \def\ragx{2}
      \foreach \agent / \x / \y / \o in
      {1/0/2/above, 2/.8/2/above,
        a/0/0/below, b/.8/0/below, 3/1.75/2/above, c/1.75/0/below, 4/2.5/2/above, 5/3.25/2/above,
        d/2.5/0/below, e/3.25/0/below%
      }
      {
        \node[agent] at (\x*\xs, \y*\os) (\agent) {};
        \node[\o = 0pt of \agent,text height=1.6ex, inner sep=0pt] {\agent};
      }
      \begin{pgfonlayer}{background}
        \foreach \s / \t / \col / \ang in
        {%
          1/b/linemarkr/10, 2/a/linemarkr/-10,
          3/c/linemarkr/10, 3/d/linemarkr/0,
          4/c/linemarkr/0, 4/e/linemarkr/-5,
          4/d/linemarkr/-10, 5/d/linemarkr/-5%
        } {
          \draw[\col] (\s) edge[bend right = \ang] (\t);
        }
      \end{pgfonlayer}
      \foreach \s / \t / \w / \v / \a / \b / \angle in
      {1/a/2/0/0.2/0.8/10,
        1/b/1/2/0.2/0.8/10,
        2/a/0/1/0.15/0.9/-10,
        2/b/1/0/0.2/0.8/-5,
        3/b/{\textcolor{blue!70}{2}}/1/0.2/0.8/0,
        3/c/{\textcolor{blue!70}{3}}/0/0.2/0.8/10,
        3/d/0/2/0.2/0.8/0,
        4/c/1/1/0.2/0.8/0,
        4/d/2/1/0.2/0.8/-10,
        1/c/0/2/0.15/0.9/-5, 4/e/0/1/0.2/0.85/-5, 5/d/1/0/0.15/0.8/-5%
      } {
        \path[draw] (\s) edge[bend right = \angle] node[pos=\a, fill=white, inner sep=0.4pt] {\small \w} node[pos=\b, fill=white, inner sep=0.4pt] {\small \v} (\t);
      }
    \end{tikzpicture}
  \end{minipage}
  \begin{minipage}{.43\linewidth}
    \[\begin{array}{@{}l@{}l@{\;}l@{}l@{}}
      \agent~1\colon & a\pref b \pref c, &      \agent~a\colon & 2\pref 1,\\
      \agent~2\colon & b\pref a,   &   \agent~b\colon & 1\pref 3 \pref 2,\\
      \agent~3\colon & c\pref b \pref d,   &    \agent~c\colon & 1 \pref 4\pref 3,\\
      \agent~4\colon & d\pref c \pref e, &  \agent~d\colon & 3\pref 4 \pref 5, \\
      \agent~5\colon & d, &  \agent~e\colon &  4.
      \end{array} \]
  \end{minipage}
  \par}

\noindent It admits three \emph{stable integral} matchings~$N_1, N_2, N_3$, where
\begin{inparaenum}[(1)]
  \item $N_1(1,a)=N_1(2,b)=N_1(3,c)=N_1(4,d)=1$,
  \item $N_2(1,b)=N_2(2,a)=N_2(3,c)=N_2(4,d)=1$, and
  \item $N_3(1,b)=N_3(2,a)=N_3(3,d)=N_3(4,c)=1$,
\end{inparaenum}
all remaining edges are set to zero.
Matching~$M_1$ with $M_1=0.5\cdot N_1+0.5\cdot N_2$ is \ostable, \cstable, and \lstable.
In terms of \linearstability, matching~$M_2$ with $M_2=(1/2+\varepsilon)\cdot N_1 + (1/2-\varepsilon)\cdot N_3$ and $0<\varepsilon<1/6$ is \lstable, but it is neither \ostable nor \cstable: Edge $\{3,b\}$ is both \oblocking and \cblocking~$M_2$.
In terms of \cardinalstability, matching~$M_3$~(marked in red), where $M_3(1,b)=M_2(2,a)=1$, $M_3(3,c)=M_2(3,d)=M_3(4,c)=M_3(4,d)=M_3(4,e)=M_3(5,d)=1/3$, and all remaining edges are set to zero, is \cstable.
\iflong For instance, $\util_{M_3}(3)=1$ and $\util_{M_3}(c)=1/3$. \fi
$M_3$ is, however, neither \lstable nor \ostable: Edge~$\{3,d\}$ is both \lblocking and \oblocking~$M_3$.
Observe that in $M_3$ every agent is matched although no stable integral matching can match agent~$5$ or agent~$e$.
\end{example}

\subsection{Computational problems}\label{sub:comp-problems}
\looseness=-1
We focus on two types of decision problems, 
one aiming for maximizing the number of fully matched agents,
and the other aiming for maximizing social welfare.
For this, given a graph~$G$$=$$(V, E)$ with satisfaction function~$\sat$$\colon V\times V \to \Q_{\ge 0}$,
and given a fractional matching~$M$ in~$G$, 
let $\fullymatched(M)$ and $\wel_{\sat}(M)$ denote
\ifshort
the following:
\myemph{$\fullymatched(M)$$\coloneqq$$|\{x$$\in$$V$$\mid$$\sum_{y\in N_G(x)}$$M(x,y)$$=$$1\}|$} and \myemph{$\wel_{\sat}(M)\coloneqq\sum_{v\in V}\util_{\sat,M}(v).$}
\else
 the number of fully matched agents and the sum of utilities of the agents under~$M$: \begin{align*}
   \fullymatched(M)\coloneqq & |\{x\in V \mid \sum_{y\in N_G(x)}M(x,y)=1\}|, \text{ and }\\
  \wel_{\sat}(M)\coloneqq & \sum_{v\in V}\util_{\sat,M}(v).
 \end{align*}
\fi%
If $\sat$ is clear from the context then we drop it in~$\wel_\sat$.
\iflong

\else %
\fi The problems are defined as follows, where~$\Pi \in \{\text{\osm, \csm}\}$\footnote{We omit \linearstability since both problems for \linearstability can be formulated as linear programs and are hence polynomial.}:
\probdef{Max-Full $\Pi$ Matching~(Max-Full $\Pi$)}{%
  A graph $G = (V,E)$, a satisfaction function $\sat\colon V\times V \to \mathds{R}_{\ge 0}$, and a non-negative integer~$\fullnum$.%
}{%
  \ifshort
  Does $(G, \sat)$ admit a $\Pi$~$M$ with $\fullymatched(M)$$\ge$$\fullnum$?%
  \else Does $(G, \sat)$ admit a $\Pi$ matching~$M$ under which at least $\fullnum$~agents are fully matched, i.e., $\fullymatched(M)\ge \fullnum$?%
  \fi
}

\probdef{Max-Welfare $\Pi$ Matching~(Max-Welfare $\Pi$)}{%
  A graph $G$$=$$(V,E)$, a satisfaction function $\sat\colon V$ $\times$ $V$ $\to$ $\mathds{R}_{\ge 0}$, and a non-negative real~$\welfare$ $\in$ $\R_{\ge 0}$.%
}{%
  \ifshort Does $(G, \sat)$ admit a $\Pi$~$M$ with $\wel(M)\ge \welfare$?%
  \else
   Does $(G, \sat)$ admit a $\Pi$ matching~$M$ with welfare at least $\welfare$, i.e., $\wel(M)\ge \welfare$?%
  \fi
}

\begin{proposition}[\appsymb%
  ]
  \POSFM, \EOSFM, \PCSFM, and \ECSFM are contained in \NP.
\end{proposition}
\iflong
\begin{proof}
  To show NP-containment, we observe that our problems can be formulated via mixed integer linear programs~(MILP), which are contained in \NP~\cite{MIQP}. A similar MILP approach has already been used for \cardinalstability by \citet{CFKV19}, but they did not address the issue regarding NP-containment as they only considered the maximization variant of \ECSFM.

  In fact, our problems reduce in polynomial time to a very restricted variant of MILP for which all integer variables have \emph{binary} values.
  Due to this, we can directly provide a polynomial-time \emph{non-deterministic} algorithm to solve our problem:
  We guess non-deterministically in polynomial-time the values of the integer variables and
  solve the resulting linear program (LP) in polynomial time.
  For the sake of completeness, we describe this approach here.
  To this end, let us first describe the MILP for our problems.
  Let $I=(G,\sat)$ be an instance with graph~$G=(V,E)$ and cardinal preferences~$\sat$.
  A fractional matching~$M$ of $G$ can be encoded via an LP as follows.
  For each edge $\{u,v\}$, we introduce a \myemph{fractional variable $x_{\{u,v\}}$} to denote the matching value assigned to edge~$\{u,v\}$ by a solution matching~$M$.
  \begin{alignat}{3}
    \sum_{v \in N_G(u)} x_{\{u,v\}} &\leq 1,&~~~~~~& \forall u \in V\tag{LP1}\label{eq:matchingLP1}\\
    x_{\{u,v\}} & \in \mathds{R}_{\ge 0},&&  \forall \{u,v\} \in E\tag{LP2}\label{eq:matchingLP2}
\end{alignat}
To encode the \cardinalstability (resp.\ \ordinalstability) of $M$,
we need to make sure that no edge~$\{u,v\}\in E$ is \cblocking (resp.\ \oblocking) $M$. %
To formulate these constraints, for each edge~$\{u,v\}$
we introduce a binary variable $y_{\{u,v\}}$ and add the following three MILP constraints for \cardinalstability: %
\begin{align}
       \forall \{u,v\}\in E\colon \nonumber\\
         y_{\{u,v\}} &\in \{0,1\},  \tag{\csm{1}} \label{eq:cblockingconstraint1}\\
         \sum_{\mathclap{w\in N_G(u)}}\sat(u,w)\cdot x_{\{u,w\}} & \ge \sat(u,v)\cdot y_{\{u,v\}},  \tag{\csm{2}}%
        \label{eq:cblockingconstraint2} \\
     \sum_{\mathclap{w\in N_G(v)}}\sat(v,w)\cdot x_{\{v,w\}} & \ge \sat(v,u)\!\cdot\! (1\!-\!y_{\{u,v\}}).\tag{\csm{3}}\label{eq:cblockingconstraint3}
\end{align}
Note that the intended meaning of~$y_{\{u,v\}}=1$ is that the utility of agent~$u$ under $M$ should be at least~$\sat(u,v)$, while $y_{\{u,v\}}=0$ means that the utility of agent~$v$ under $M$ should be at least~$\sat(v,u)$.

  For \ordinalstability, we instead add the following three MILP constrains: %
  \begin{align}
      \forall \{u,v\}\in E\colon \nonumber\\
        y_{\{u,v\}} &\in \{0,1\}, \tag{\osm{1}} \label{eq:oblockingconstraint1}\\
       \sum_{w\in \bettere{u}{v}}x_{\{u,w\}} & \ge y_{\{u,v\}}, \tag{\osm{2}} \label{eq:oblockingconstraint2}\\
        \sum_{w\in \bettere{v}{u}}x_{\{v,w\}} & \ge 1-y_{\{u,v\}}.\tag{\osm{3}} \label{eq:oblockingconstraint3}
    \end{align}
    Note that the intended meaning of~$y_{\{u,v\}}=1$ is that the sum of values of matching agent~$u$ to someone better or equal to~$v$ should be at least one,
    while $y_{\{u,v\}}=0$ means that the sum of values of matching agent~$v$ to someone better or equal to~$u$ should be at least one,
    
    To solve \PCSFM (resp.\ \POSFM) with objective value~$\fullnum$, we introduce one more \myemph{binary variable~$z_u$} for each agent~$u\in V$
    to specify whether it will be fully matched and add the following MILP constraints: %
    \begin{align}
      \sum_{u\in V}z_{u} &\ge \fullnum,&\tag{FULL1}\label{eq:full1}\\
      \sum_{v\in N_G(u)} x_{\{u,v\}} &\geq z_{u}, ~~~~~~~~~~~~~& \forall u\in V, \tag{FULL2}\label{eq:full2}\\
      z_{u} &\in \{0,1\}, ~~~~~~~~~~~~~ &\forall u\in V. \tag{FULL3}\label{eq:full3}
    \end{align}
    to the constraints~\eqref{eq:matchingLP1}--\eqref{eq:matchingLP2} and \eqref{eq:cblockingconstraint1}--\eqref{eq:cblockingconstraint3}
    (resp.\ to the constraints~\eqref{eq:matchingLP1}--\eqref{eq:matchingLP2} and \eqref{eq:oblockingconstraint1}--\eqref{eq:oblockingconstraint3}).
    
    To solve \ECSFM (resp.\ \EOSFM) with objective value~$\welfare$, we only add the following constraint:
    \[
      \sum_{\{u,v\}\in E}(\sat(u,v)+\sat(v,u))\cdot x_{\{u,v\}} \ge \welfare.\tag{WELFARE}\label{eq:welfare}
    \]
    to the constraints~\eqref{eq:matchingLP1}--\eqref{eq:matchingLP2} and \eqref{eq:cblockingconstraint1}--\eqref{eq:cblockingconstraint3}
    (resp.\ to the constraints~\eqref{eq:matchingLP1}--\eqref{eq:matchingLP2} and \eqref{eq:oblockingconstraint1}--\eqref{eq:oblockingconstraint3}).
    
    This completes the description of the MILPs for our problems. 
    As already discussed at the beginning of the proof,
    since our MILPs have $O(|E|+|V|)$ binary variables~$y_{\{u,v\}}$, $\{u,v\}\in E$  and $z_{u}$, $u\in V$,
    we guess their values and check in polynomial time whether the guessed values combined with the resulting LP constraints %
    are feasible.
    This shows that our decision problems belong to NP.
\end{proof}
\fi
\noindent\looseness=-1
By the above containment results, when we show \NPC{ness} later it suffices to prove NP-hardness.

\section{Structural properties}%
\label{sec:structural}
We now discuss relations among\ifshort~(see \cref{fig:conceptsrelation}) \fi~ and existence of fractional matchings regarding the three stability concepts, and then show that \osm{s} behave similarly to \lsm{s} in terms of lattice property. %
\iflong First, we observe that \ordinalstability is a notion stronger than \linearstability and \cardinalstability, while \cardinalstability and \linearstability are not comparable to each other~(see \cref{fig:conceptsrelation}).
\fi
\begin{observation}[\appsymb]\label{obs:ostable->cstable}
\iflong  \begin{compactenum}[(i)]
\else  \begin{inparaenum}[(i)]\fi
    \item\label{ordinal->linear} Every \osm of a graph with cardinal preferences is a \lsm and a \csm.
    \item\label{linear-nocardinal} There exists a graph~$G$ with strict preferences such that
    $G$ admits a \lsm which is neither an \osm nor a \csm and
    admits a \csm which is neither an \lsm nor an \osm.
\iflong    \end{compactenum}
\else    \end{inparaenum}\fi
\end{observation}
\iflong
\noindent
The implication from \ordinalstability\ to \linearstability\ in Statement~\eqref{ordinal->linear} has been proved by \citet{aziz_random_2019} in the marriage setting (see the statement that ex-ante weak stability implies robust ex-post weak stability in their Theorem~3.)
A counterexample for the statement that \linearstability\ implies \ordinalstability\ (which is part of Statement~\eqref{linear-nocardinal} above) has also been given by \citet{aziz_random_2019}, see their Theorem~3 as well.
\begin{proof}[Proof of \cref{obs:ostable->cstable}]
  The first of Statement~\eqref{ordinal->linear} regarding \osm{s} and \lsm{s} follows directly from the definition.
  Now, to show the second part of Statement~\eqref{ordinal->linear}, let $M$ be an \osm of an instance~$I=(G,\sat)$ with graph~$G$ and cardinal preferences~$\sat$. %
  Consider an arbitrary edge~$\{u,v\}\in E(G)$.
  Since $M$ is \ostable it follows that $\{u,v\}$ is not an \oblocking edge.
  That is,

 {\centering $M(u,\wpref v) \ge 1,~~~~~\refstepcounter{equation}(\theequation){\label{eq:oblocking1}} \text{~~~ or }
    M(v,\wpref u) \ge 1. ~~~~~\refstepcounter{equation}(\theequation){\label{eq:oblocking2}}$ \par}

 \noindent   If \eqref{eq:oblocking1} holds, then it follows that
  \begin{align*}
    \util_M(u) %
    & \ge \sum_{w \in \bettere{u}{v}}\sat(u,w) \cdot M(u,w) \stackrel{\eqref{eq:oblocking1}}{\ge}  \sat(u,v).
  \end{align*}
  If \eqref{eq:oblocking2} holds, then it follows that 
  \begin{align*}
    \util_M(v) %
               & \ge \sum_{w \in \bettere{v}{u}}\sat(v,w) \cdot M(v,w) \stackrel{\eqref{eq:oblocking2}}{\ge}  \sat(v,u).
  \end{align*}
  Hence, $\{u,v\}$ is not \cblocking~$M$, implying that $M$ is \cstable.

  In the instance given in \cref{ex:cstable-not-always-ostable} matchings $M_2$ and $M_3$ show Statement~\eqref{linear-nocardinal}.
\end{proof}
\fi
\begin{figure}[t!]
  \centering
\begin{tikzpicture}[>=stealth', shorten <= 1pt, shorten >= 1pt]
  \tikzstyle{concept} = [rounded corners, draw, rectangle, text width=12ex, inner sep=1pt, minimum height=5ex, align=center]
  \def \oy {6ex}
  \node[concept] (integralsm) {stable and integral};
  \node[concept, right = \oy of integralsm] (osm) {\ostable};
  \node[concept, above right = -10pt and \oy of osm] (lsm) {\lstable};
  \node[concept, below right = -10pt and \oy of osm] (csm) {\cstable};

  \foreach \s / \t in {integralsm/osm,osm/lsm,osm/csm} {
    \draw[->] (\s) -- (\t);
  }
\end{tikzpicture}
\caption{Relation between the three stability concepts, where ``$\alpha \rightarrow \beta$'' means that ``an $\alpha$ matching is also $\beta$''.}\label{fig:conceptsrelation}
\end{figure}
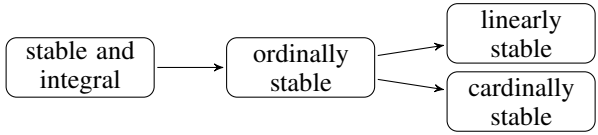

The following concept of \myemph{stable partitions}, introduced by \citet{Tan1991}, turns out to be very useful for showing the existence of \ostable matchings.

\iflong \begin{definition}[Stable partitions and cycles, their corresponding matchings] \label{def:stable-partitions}
  \else\begin{definition} \label{def:stable-partitions}
    \fi
  A \myemph{stable partition} of $(G=(V,E),\sat)$ with $\sat$ being strict is a permutation~$\pi\colon V$$\to$$V$ on the vertices,
  which satisfies the following two conditions for each vertex~$v_i\in V\colon$
\iflong  \begin{compactenum}[(1)]
    \item\label{stable-partition:1} if $\pi(v_i) \neq \pi^{-1}(v_i)$, then $\{v_i,\pi(v_i)\}, \{v_i,\pi^{-1}(v_i)\}\in E$ and $\sat(v_i,\pi(v_i)) > \sat(v_i,\pi^{-1}(v_i))$; 
    \item\label{stable-partition:2} for each vertex~$v_j$ adjacent to~$v_i$, if $\pi(v_i)=v_i$ or $\sat(v_i,v_j) > \sat(v_i,\pi^{-1}(v_i))$, then $\sat(v_j, \pi^{-1}(v_j)) > \sat(v_j, v_i)$.
  \end{compactenum}
  \else
  \begin{inparaenum}[(1)]
    \item\label{stable-partition:1} if $\pi(v_i) \neq \pi^{-1}(v_i)$, then $\{v_i,\pi(v_i)\}, \{v_i,\pi^{-1}(v_i)\}\in E$ and $\sat(v_i,\pi(v_i)) > \sat(v_i,\pi^{-1}(v_i))$; 
    \item\label{stable-partition:2} for each~$v_j$ adjacent to~$v_i$, if $\pi(v_i)=v_i$ or $\sat(v_i,v_j) > \sat(v_i,\pi^{-1}(v_i))$, then $\sat(v_j, \pi^{-1}(v_j)) > \sat(v_j, v_i)$.
  \end{inparaenum}\fi
  We call $v_i$ a \myemph{singleton} if $\pi(v_i)=v_i$.
  A stable partition~$\pi$ can be decomposed into \myemph{cycles}, \myemph{singletons}, and \myemph{transpositions} (i.e., disjoint edges).
  Here, a subpermutation~$\sigma$ on a subset~$V'\subseteq V$ \iflong of vertices \fi is called a \myemph{cycle} if the edge set~$\{\{v,\sigma(v)\} \mid v \in V'\}$ forms a cycle in~$G$; we define the \myemph{length} of a cycle~$\sigma$ to be \iflong the size of $V'$. \else $|V'|$.\fi
  
  \looseness=-1
  Let $\pi$ be a stable partition.
  Define a matching \myemph{$M^{\pi}$} for $G$ \myemph{corresponding to}~$\pi$ as follows.
\iflong   \begin{compactenum}[(a)]
    \item For each non-singleton~$v_i\in V$ (i.e., $\pi(v_i)\neq v_i$),
    if $\pi(v_i)=\pi^{-1}(v_i)$, meaning that $(v_i,\pi(v_i))$ forms a transposition in~$\pi$, then $M^{\pi}(v_i,\pi(v_i))\coloneqq 1$; otherwise
    $M^{\pi}(v_i,\pi(v_i)) = M(v_i,\pi^{-1}(v_i))\coloneqq 0.5$.
    \item For each remaining edge~$e$, let $M^{\pi}(e)\coloneqq 0$.\hfill \defqed
  \end{compactenum}
  \else
  \begin{inparaenum}[(a)]
    \item For each non-singleton~$v_i\in V$ (i.e., $\pi(v_i)\neq v_i$),
    if $\pi(v_i)$$=$$\pi^{-1}(v_i)$, then let $M^{\pi}(v_i,\pi(v_i))\coloneqq 1$; otherwise
    let $M^{\pi}(v_i,\pi(v_i))=M(v_i,\pi^{-1}(v_i))\coloneqq 0.5$.
    \item For each remaining edge~$e$, let $M^{\pi}(e)\coloneqq 0$.\hfill \defqed
  \end{inparaenum}\fi
\end{definition}

\begin{example}
 \iflong Consider the instance from the introduction. %
  There is only one stable partition~$\pi=(a,b,c)(d,e)(f)$.
  \else
  There is only one stable partition~$\pi=(a,b,c)(d,e)(f)$ in the instance from the introduction. \fi
  Since it consists of an odd cycle of length three, namely $(a,b,c)$, the instance does \emph{not} admit a stable and \emph{integral} matching.
  The matching marked in green is \ostable, and hence \cstable and \lstable. %
  Note that this matching is exactly $M^{\pi}$ that we define for $\pi$ in \cref{def:stable-partitions}.
\end{example}

\iflong %
Singleton agents of a graph with strict preferences are unique in the following sense. %

 \begin{proposition}[\cite{Tan1991}]\label{prop:singletons}
  Let $G$ be a graph with $n$ vertices and strict preferences~$\sat$.
  Then, $(G,\sat)$ admits a stable partition, which can be found in $O(n^2)$~time.
  Moreover, every stable partition of $(G,\sat)$ has the same set of singleton agents. %
\end{proposition}

 \else \refstepcounter{theorem} \fi 
We will see that the matchings corresponding to stable partitions are \ostable\ifshort. \else, even when ties are present. \fi
Before we show this, we note that the case without ties is already observed by \citet{AF03Scarfs,BiroCechFleiner08stablehalfmatchings}. %
\begin{algorithm}[t!]
  \DontPrintSemicolon
  \caption{Compute an \osm for a graph~$G$ with cardinal preferences~$\sat$.}\label[algorithm]{alg:arbitrary-osm}
  \small
  Compute the preference lists~$\Pot$ from $\sat$\label{alg:osm-ties}

  $\hat{\sat} \leftarrow$ Break ties in $\sat$ arbitrarily\label{alg:osm-noties}
  
  Compute a stable partition $\pi$  for $(G,\hat{\sat})$ and the corresponding matching~$M^{\pi}$ according to  \cref{def:stable-partitions} \label{alg:osm-partition}
  
  \Return $M^{\pi}$
\end{algorithm}

\iflong
\begin{proposition}[\cite{AF03Scarfs,BiroCechFleiner08stablehalfmatchings}]\label{lem:sp-matching->osm}
  Let $G=(V,E)$ be a graph with strict preferences~$\sat$.
  The matching~$M^{\pi}$ as defined in \cref{def:stable-partitions} is an \osm and a \csm for $(G,\sat)$. %
\end{proposition}

In the following, we derive the same result for cardinal preferences using the notion of stable partitions.
\else \refstepcounter{theorem}  \fi
\begin{lemma}[\appsymb]\label{osroommates:always}
  Each graph on $n$~vertices and with cardinal preferences (and possibly ties) admits an \osm, and hence a \csm, that is half-integral and matches each matched agent fully.
  \cref{alg:arbitrary-osm} finds such a matching in $O(n^2)$ time.
\end{lemma}
\iflong
\begin{proof}
  To show the statement we first show that $(G,\sat)$ admits an \osm.
  Let $G=(V,E)$ be a graph with cardinal preferences~$\sat$, and let $\Pot$ be the preferences lists derived from~$\sat$.
  We aim to show that $M^{\pi}$ as returned by \cref{alg:arbitrary-osm} on input~$(G,\sat)$ is a half-integral \osm where each \emph{matched} agent is fully matched.
  Let \myemph{$\hat{\sat}$} be a satisfaction function with strict preferences, which is derived from $\sat$ by breaking ties arbitrarily (except the values~$\sat(x,y)$
  with $x=y$ or with $\{x,y\}\notin E$).
  That is, $\hat{\sat}$ is a satisfaction function without ties such that
  for each agent~$x$ and each two neighbors~$y,z\in N_G(x)$ it holds that if $\sat(x,y) > \sat(x,z)$ then $\hat{\sat}(x,y) > \hat{\sat}(x,z)$.
  By \cref{prop:singletons}, let $\pi$ be the stable partition of $(G,\hat{\sat})$ and $M^{\pi}$ be the corresponding matching computed in line~\ref{alg:osm-partition} of \cref{alg:arbitrary-osm}. %

  Clearly,~$M^{\pi}$ is half-integral such that every matched agent is fully matched. 
  It remains to show that $M^{\pi}$ is an \osm of $(G,\sat)$.
  First of all, we show the following claim.
  \begin{claim}\label{claim:all1}
    For each edge~$\{x,y\}\in E$ it holds that
    if  $x$ is not a singleton and ${\sat}(x,\pi^{-1}(x)) \ge {\sat}(x,y)$,
    then  $M(x,\wpref y)=1$. %
  \end{claim}
  \begin{proof}[Proof of \cref{claim:all1}]
    \renewcommand{\qedsymbol}{$\diamond$} %
    We distinguish between two cases.
  \begin{itemize}
    \item If $\pi(x)=\pi^{-1}(x)$, then since $x$ is not a singleton, by the definition of $M^{\pi}$, it follows that $M^{\pi}(x,\pi^{-1}(x))=M(x,\pi(x))=1$, implying our claim.
    \item If $\pi(x)\neq\pi^{-1}(x)$, then by the definition of $M^{\pi}$, it follows that $M^{\pi}(x,\pi^{-1}(x))+M(x,\pi(x))=1$.
     Since $x$ is not a singleton, by \cref{def:stable-partitions}\eqref{stable-partition:1}, it follows that $\hat{\sat}(x,\pi(x)) > \hat{\sat}(x,\pi^{-1}(x))$.
    By our definition of $\hat{\sat}$ it follows that $\sat(x,\pi(x))\ge \sat(x,\pi^{-1}(x))$.
    The claim follows immediately.\qedhere%
\end{itemize}%
\end{proof}

  \noindent Now, we are ready to show that $M^{\pi}$ is an \osm of $(G,\sat)$.
  Suppose, for the sake of contradiction, that $M$ admits an \oblocking pair, say~$e=\{u,v\}$. We distinguish between three cases, in each case obtaining a contradiction.
  \begin{itemize}
    \item If one of $u$ and $v$ is a singleton, say $u$, then by \cref{def:stable-partitions}\eqref{stable-partition:2},~$v$ cannot be a singleton as otherwise, $\hat{\sat}(v,u) < \hat{\sat}(v,v)=0$, which is not possible by the definition.
    Moreover, again by \cref{def:stable-partitions}\eqref{stable-partition:2}, we have $\hat{\sat}(v, \pi^{-1}(v)) > \hat{\sat}(v,u)$ since $u$ is a singleton.
    This implies that $\sat(v, \pi^{-1}(v)) \ge \sat(v,u)$.
    By \cref{claim:all1}, we have that $M(v,\wpref u)=1$, a
    contradiction to $\{u,v\}$ be an \oblocking pair of $M$.
    \item If neither $u$ nor $v$ is a singleton, but ${\sat}(u,v) > {\sat}(u, \pi^{-1}(u))$, %
    then by the definition of $\hat{\sat}$, %
    it follows that $\hat{\sat}(u,v) > \hat{\sat}(u, \pi^{-1}(u))$.
    By \cref{def:stable-partitions}\eqref{stable-partition:2} of stable partitions it follows that $\hat{\sat}(v,\pi^{-1}(v)) \ge \hat{\sat}(v,u)$.
    By \cref{claim:all1}, we again obtain that $\{u,v\}$ is not an \oblocking pair, a contradiction.
    \item If neither $u$ nor $v$ is a singleton, but ${\sat}(u,v) \le {\sat}(u, \pi^{-1}(u))$, then
   \cref{claim:all1} immediately implies that $\{u,v\}$ is not an \oblocking pair, a contradiction.
  \end{itemize}

Together with \cref{obs:ostable->cstable}, we know that $M^{\pi}$ is also \cstable. Hence, every graph with cardinal preferences admits an~\csm.

It remains to consider the running time of \cref{alg:arbitrary-osm}.
Recall that $n$ denotes the number of vertices in $G$.
Computing $\Pot$ and $\hat{\sat}$ in lines~\ref{alg:osm-ties}--\ref{alg:osm-noties} can clearly be done in $O(n^2)$~time.
By \cref{prop:singletons}, $\pi$ can be computed from $(G,\Pot')$ in $O(n^2)$. %
By \cref{def:stable-partitions}, the matching~$M^{\pi}$ corresponding to~$\pi$ can be computed $O(n)$~time.
In total, the running time of \cref{alg:arbitrary-osm} is $O(n^2)$.
\end{proof}
\fi
We close this section by considering the lattice property of \osm{s}.
It is well-known by \citet{RothRothblumVate1993lsm-lattice} that for bipartite graphs with strict cardinal preferences,
the set of \lsm{s} displays a certain lattice structure.
Following their result, we show that the same holds for \osm{s}.
To this end, given two fractional matchings~$M_1$ and $M_2$ of a bipartite graph~$G=(U\cup W,E)$ with preference lists~$(\succeq_x)_{x\in U\cup W}$, the \myemph{join~$\join$} and \myemph{meet~$\meet$} of $M_1$ and $M_2$ by \citet{RothRothblumVate1993lsm-lattice} are defined as follows:
\iflong
\begin{align*}
  \forall x \in U, \forall y \in W \colon &\\
  M_1\join M_2(\{x,y\})\coloneqq & \max\big(M_1(x,\wpref y), M_2(x,\wpref y)\big) - \\
                             &   \max\big(M_1(x,\pref y), M_2(x,\pref y)\big)\\
  M_1\meet M_2(\{x,y\})\coloneqq & \min\big(M_1(x,\wpref y), M_2(x,\wpref y)\big) - \\
                             &   \min\big(M_1(x,\pref y), M_2(x,\pref y)\big)
\end{align*}
\else
$  \forall x \in U, \forall y \in W \colon$ $M_1\join M_2(\{x,y\})$ $\coloneqq $ $\max\big(M_1(x,\wpref y), M_2(x,\wpref y)\big) - \max\big(M_1(x,\pref y), M_2(x,\pref y)\big)$; and

  $M_1\meet M_2(\{x,y\})$ $\coloneqq$ $\min\big(M_1(x,\wpref y), M_2(x,\wpref y)\big) - \min\big(M_1(x,\pref y), M_2(x,\pref y)\big)$.
\fi

\ifshort \noindent Matching $M_1$
\else \noindent  We say that matching~$M_1$
\fi \myemph{weakly $U$-dominates} (resp.\ \myemph{$U$-dominates})
\ifshort $M_2$,
\else
 matching~$M_2$,
\fi written \myemph{$M_1\wpref_U M_2$} (resp.\ \myemph{$M_1 \pref_U M_2$}) if 
\iflong \begin{align*}
  \forall (x,y)\in U\times W \colon M_1(x,\wpref y) \ge M_2(x,\wpref y)\\
  (\text{resp. }   M_1(x,\wpref y) > M_2(x,\wpref y)).
\end{align*}
\else
$\forall (x,y)\in U\times W \colon$\\ $M_1(x,\wpref y) \ge M_2(x,\wpref y)
  (\text{resp. }   M_1(x,\wpref y) > M_2(x,\wpref y))$.
\fi

For an illustration, consider the following example.

\begin{example}
\iflong  Consider the following instance with strict preference lists; the underlying acceptability graph is a complete bipartite graph on~$U\times W$ with $U=[3]$ and $W=\{a,b,c\}$.
\else
  Consider a bipartite graph on~$U\cup W$, $U=[3]$, and $W=\{a,b,c\}$ with strict preferences.\fi 
  \iflong
  \begin{alignat*}{4}%
      \agent~1\colon & a\succ b \succ c, ~~~
    & \agent~2\colon & b\succ c \succ a, ~~~
    & \agent~3\colon & c\succ a \succ b, \\
    \agent~a\colon & 2\succ 3 \succ 1, &
    \agent~b\colon & 3\succ 1 \succ 2,   & \agent~c\colon & 1\succ 2 \succ 3.
  \end{alignat*}%
  \else
  $\agent~1\colon a\pref b \pref c$, $\agent~2\colon  b\pref c \pref a$, 
  $\agent~3\colon c\pref a \pref b$,  $\agent~a\colon 2\pref 3 \pref 1$, 
  $\agent~b\colon 3\pref 1 \pref 2$,   $\agent~c\colon  1\pref 2 \pref 3$.
  \fi
  \iflong
  Note that for \osm{s}, the exact values of the cardinal preferences in~$\sat$ are not important.
  For the sake of completeness, we define~$\sat$ as follows:
  $\sat(1,a)=\sat(2,b)=\sat(3,c)=\sat(a,2)=\sat(b,3)=\sat(c,1)=2$, 
  $\sat(1,b)=\sat(2,c)=\sat(3,a)=\sat(a,3)=\sat(b,1)=\sat(c,2)=1$, and
  $\sat(1,c)=\sat(2,a)=\sat(3,b)=\sat(a,1)=\sat(b,2)=\sat(c,3)=0$.
  
  \noindent This instance admits $3$ stable integral matchings~$N_1, N_2$, and $N_3$ with $N_1(1,a)=N_1(2,b)=N_1(3,c)=1$, $N_2(1,b)=N_2(2,c)=N_2(3,a)=1$, $N_3(1,c)=N_3(2,a)=N_3(3,b)=1$ (unmentioned edges are set to zero).
  It also has two \osm{s}: $M_1$ and $M_2$ where each agent is fully matched:
  $M_1=\frac{3}{4}\cdot N_2 + \frac{1}{4}\cdot N_3$ and $M_2=\frac{1}{3}\cdot N_1 + \frac{1}{3}\cdot N_2 + \frac{1}{3}\cdot N_3$.
\else \noindent It admits $3$ stable integral matchings~$N_1, N_2$, and $N_3$ with $N_1(1,a)$$=$$N_1(2,b)$$=$$N_1(3,c)$$=$$1$; $N_2(1,b)$$=$$N_2(2,c)$$=$
  $N_2(3,a)$${}={}$$1$; $N_3(1,c)$${}={}$$N_3(2,a)$${}={}$$N_3(3,b)$${}={}$$1$ (unmentioned edges are set to zero).
  It also has two perfect \osm{s}: $M_1$ and $M_2$:
  $M_1=\frac{3}{4}\cdot N_2 + \frac{1}{4}\cdot N_3$ and $M_2=\frac{1}{3}\cdot N_1 + \frac{1}{3}\cdot N_2 + \frac{1}{3}\cdot N_3$.  
  \fi
  We obtain the join and the meet of $M_1$ and $M_2$ as follows:
  $M_1 \join M_2 = \frac{1}{3}\cdot N_1 + \frac{1}{12}\cdot N_2 + \frac{1}{4}\cdot N_3$, and
  $M_1 \meet M_2 = \frac{1}{3}\cdot N_2 + \frac{2}{3}\cdot N_3$.
  Note that any convex combination of $N_1, N_2$, and $N_3$ forms an \osm, which is not always the case \iflong for all bipartite graphs with strict preferences.
  \else for bipartite graphs with strict preferences. \fi
\end{example}

\iflong \begin{proposition}[\cite{RothRothblumVate1993lsm-lattice}]\label{prop:lsm-lattice}
  Let $G=(U\cup W, E)$ be a bipartite graph with strict preferences and let $\Pot=(\succ_x)_{x\in U\cup W}$ be the associated preference lists.
  Let~$M_1$ and $M_2$ be two \lsm{s} for \((G, \Pot)\).
  Then $M_1\join M_2$ and $M_1\meet M_2$ are two \lsm{s} for~$(G, \Pot)$.
  Moreover, for each $(x,y)\in U\times W$ it holds that
  \begin{compactenum}[(1)]
    \item\label{lsm-lattice:joinU} $M_1\join M_2(x, \wpref y) = \max(M_1(x,\wpref y), M_2(x,\wpref y))$,
    \item\label{lsm-lattice:meetU} $M_1\meet M_2(x, \wpref y) = \min(M_1(x,\wpref y), M_2(x,\wpref y))$,
    \item\label{lsm-lattice:joinW} $M_1\join M_2(y, \wpref x) = \min(M_1(y,\wpref x), M_2(y,\wpref x))$, and
    \item\label{lsm-lattice:meetW} $M_1\meet M_2(y, \wpref x) = \max(M_1(y,\wpref x), M_2(y,\wpref x))$.
  \end{compactenum}
  The set of all \lsm{s} of $(G,\Pot)$ and the partial order~$\wpref_U$ forms a distributive lattice, with $\join$ and $\meet$ representing the join and meet of any two \lsm{s}.
\end{proposition}

Using the above fundamental property and by \cref{obs:ostable->cstable}\eqref{ordinal->linear}, we show that \osm{s} form a lattice substructure of \lsm{s} regarding the partial order~\myemph{$\wpref_U$} on the set of fractional matchings.
\else \refstepcounter{theorem} 
Using the lattice structure properties for \lsm{s} \cite{RothRothblumVate1993lsm-lattice} and by \cref{obs:ostable->cstable}\eqref{ordinal->linear}, we show that \osm{s} form a lattice substructure of \lsm{s} regarding the partial order~\myemph{$\wpref_U$} on the set of fractional matchings.
\fi

\begin{proposition}[\appsymb]\label{prop:osm-lattice}
  \iflong
  For each bipartite graph with strict preferences, the set of \osm{s} forms a distributive lattice under the partial order~$\wpref_U$.
  \else
  For bipartite graphs with strict preferences, the set of \osm{s} forms a distributive lattice under~$\wpref_U$.
  \fi
\end{proposition}
\iflong
\begin{proof}
  Let $G=(U\cup W, E)$ be a bipartite graph with strict preferences~$\sat$, let $\Pot=(\succ_z)_{z\in U\cup W}$ denote the induced strict preference lists,
  and let $\mathcal{M}^{\osms}$ denote the set of all \osm{s} of $(G,\sat)$.
  Since $\mathcal{M}^{\osms}$ is a subset of the set of all \lsm{s} of $(G,\sat)$,
  to show that $(\mathcal{M}^{\osms}, \wpref_U)$ is a distributive lattice, it suffices to show that
  for each two \osm{s}~$M_1, M_2\in \mathcal{M}^{\osms}$ both $M_1\join M_2$ and $M_1 \meet M_2$ are \ostable. 

  Consider an arbitrary ordered pair~$(x,y)\in U\times W$.
  We need to show that
  \begin{inparaenum}[(a)]
    \item \label{lattice-osm-join}
    $M_1 \join M_2 (x,\wpref y) \ge 1$ or  $M_1 \join M_2 (y,\wpref x) \ge 1$, and
    \item\label{lattice-osm-meet}
    $M_1 \meet M_2 (x,\wpref y) \ge 1$ or  $M_1 \meet M_2 (y,\wpref x) \ge 1$.
  \end{inparaenum}

  To show \eqref{lattice-osm-join} 
  suppose, for the sake of contradiction, that $M_1 \join M_2 (x,\wpref y) < 1$ and $M_1 \join M_2 (y,\wpref x) < 1$. %
  By \cref{prop:lsm-lattice}\eqref{lsm-lattice:joinU}, %
  it follows that
  $M_1(x,\wpref y)<1$ and $M_2(x,\wpref y))<1$.
  Since $M_1$ and $M_2$ are \ostable, by definition, it must hold that
  $M_1(y,\wpref x)\ge 1$ and $M_2(y,\wpref x)) \ge 1$. %
  By \cref{prop:lsm-lattice}\eqref{lsm-lattice:joinW}, this means that $M_1\join M_2(y, \wpref x)=\min(M_1(y,\wpref x), M_2(y,\wpref x)) \ge 1$, a contradiction to our assumption.

  The reasoning for \eqref{lattice-osm-meet} is omitted because it is analogous using \cref{prop:lsm-lattice}\eqref{lsm-lattice:meetW} and \cref{prop:lsm-lattice}\eqref{lsm-lattice:meetU} instead of \cref{prop:lsm-lattice}\eqref{lsm-lattice:joinW} and \cref{prop:lsm-lattice}\eqref{lsm-lattice:joinU}.
\end{proof}
\fi
\looseness=-1
\noindent Similar to \linearstability, the lattice structure of \ordinalstability does not hold in the roommates setting, but \iflong the set of \ostable\ matchings \fi is closed under a median operation\ifshort: \else. \fi
\iflong For this, given three
\else Given $3$
\fi real values~$x,y,z$, let \myemph{$\med(x,y,z)$} denote the second largest \iflong (or smallest)\fi number among \iflong $\{x,y,z\}$. \else them. \fi
To show the median property for \linearstability,
\citet{abeledo_stable_1994} extended the median notion to fractional matchings.
Given a graph~$G$ with preferences~$(\wpref_x)_{x\in V(G)}$ and
\ifshort $3$
\else three
\fi fractional matchings~$M_1,M_2,M_3$\iflong\ of $G$\fi, 
\iflong
let the \myemph{median} of~$M_1,M_2,M_3$, denoted as~\myemph{$\med(M_1,M_2,M_3)$},
be defined as follows:
 \begin{align*}
  \forall x,y\in V &\colon \med(M_1,M_2,M_3)(x,y) \coloneqq\\  
                     &     \med(M_1(x,\wpref y), M_2(x,\wpref y), M_3(x,\wpref y)) -\\
  &\med(M_1(x,\pref y), M_2(x,\pref y), M_3(x,\pref y)).
\end{align*}
\else
define \myemph{$\med(M_1,M_2,M_2)$}~as:~$\forall$$x,y$$\in$$V\colon$ $\med(M_1,M_2,M_3)(\{x,y\})$$\coloneqq$$\med(M_1(x,\wpref y),M_2(x,\wpref y)$, $M_3(x,\wpref y))-\med(M_1(x,\pref y), M_2(x,\pref y)$,$M_3(x,\pref y))$.
\fi

\iflong \begin{proposition}[\cite{abeledo_stable_1994}]\label{lsm:median}
  Let $\mathcal{M}^{\lsms}$ denote the set of all \lsm{s} of a graph with strict preferences.
  Then, for each~$M_1,M_2,M_3\in \mathcal{M}^{\lsms}$ it holds that
  \begin{compactenum}[(1)]
    \item\label{lsm:median:eqs} $\med(M_1,M_2,M_3)(x,\wpref y)$$=$$\med(M_1(x,\wpref y), M_2(x, \wpref y)$, $M_3(x,\wpref))$ for all $x,y\in V$,
    \item $\med(M_1,M_2,M_3)(x,\pref y)$$=$$\med(M_1(x,\pref y), M_2(x,\pref y)$, $M_3(x,\pref))$ for all $x,y\in V$,  and 
    \item\label{lsm:median:closedness} $\med(M_1,M_2,M_3)\in \mathcal{M}^{\lsms}$.
  \end{compactenum}
\end{proposition}

\cref{lsm:median} immediately implies an analogous median property for \osm{s}.
\else \refstepcounter{theorem} 
\noindent\looseness=-1
With the median property for \lsm{s} \cite{abeledo_stable_1994} we show an analogous median property for \osm{s}.
\fi

\begin{proposition}[\appsymb]\label{prop:osm:median}
  Let $\mathcal{M}^{\osms}$ be the set of all \osm{s} of a graph with strict preferences.
  Then, for each~$M_1,M_2,M_3\in \mathcal{M}^{\osms}$ we have $\med(M_1,M_2,M_3)\in \mathcal{M}^{\osms}$.
\end{proposition}
\iflong
\begin{proof}
  Let $G=(V,E)$ be a graph with strict preferences~$\sat$, and let~$\Pot=(\succ_x)_{x\in V}$ denote the induced preference lists.
  Let $\mathcal{M}^{\osms}$ be the set of \osm{s} of $(G,\sat)$, and let $M_1,M_2,M_3\in \mathcal{M}^{\osms}$ be as defined in the statement.
  By \cref{obs:ostable->cstable}\eqref{ordinal->linear},
  $M_1,M_2,M_3$ are \lsm{s}.
  Hence, by \cref{lsm:median}\eqref{lsm:median:closedness}, $\med(M_1,M_2,M_3)$ is a fractional matching of $G$, since it is an \lsm of $(G,\sat)$.
  To show the membership in $\mathcal{M}^{\osms}$, we need to show that for each pair~$\{x,y\}\subseteq V$ of agents, 
  $\med(M_1,M_2,M_3)(x,\wpref y) \ge 1$ or $\med(M_1,M_2,M_3)(y,\wpref x) \ge 1$.
  Consider an arbitrary pair~$\{x,y\}\subseteq V$, if  $\med(M_1,M_2,M_3)(x,\wpref y) \ge 1$, then we are done.
  Hence, let us assume that $\med(M_1,M_2,M_3)(x,\wpref y) < 1$.
  By \cref{lsm:median}\eqref{lsm:median:eqs},
  we have that $\med(M_1(x,\wpref y), M_2(x,\wpref y), M_3(x,\wpref y)) < 1$.
  This means that at least two of the three real values~$M_1(x,\wpref y)$, $M_2(x,\wpref y)$, $M_3(x,\wpref y)$ are strictly smaller than one. Without loss of generality by symmetry, assume that $M_1(x,\wpref y) < 1$, $M_2(x,\wpref y) < 1$.
  Then, since $M_1,M_2\in \mathcal{M}^{\osms}$, it follows that
  $M_1(y,\wpref x)\ge 1$, $M_2(y,\wpref x)\ge 1$.
  We distinguish between two cases for~$M_3(y,\wpref x)$.

  If $M_3(y,\wpref x) \ge 1$, then $\med(M_1,M_2,M_3)(y,\wpref x) \ge 1$.
  Otherwise, $\med(M_1(y,\wpref x),M_2(y,\wpref x), M_3(y,\wpref x)) = \min\big(M_1(y,\wpref x),M_2(y,\wpref x)\big) \ge 1$.
  In both cases, we obtain that $\med(M_1,M_2,M_3)(x,\wpref y) \ge 1$ or  $\med(M_1,M_2,M_3)(y,\wpref x) \ge 1$.
  Hence, $\med(M_1,M_2,M_3)\in \mathcal{M}^{\osms}$.
\end{proof}
\fi
\section{Algorithmic results}
The structural properties from \cref{sec:structural} give rise to efficient algorithms for finding optimal stable matchings.
\iflong
\paragraph{Bipartite graphs with strict preferences.}
To describe efficient algorithms for this case, we first observe that for the case of bipartite graphs, each support of each \osm consists of integral \emph{stable} matchings.
We remark that this has also been proved by \citeauthor{aziz_random_2019}~(\citeyear{aziz_random_2019}, Theorem~3) in their study of random matchings.\footnote{%
  The relevant statement is that an \emph{ex-ante weakly stable} random matching~$p$ is also \emph{robust ex-post weakly stable}.}
For self-containedness and since our proof is short, instructive, and more direct than theirs, %
we include it here:
\else
For the case of bipartite graphs, we utilize the fact that each support of each \osm consists of integral \emph{stable} matchings.
\fi

\begin{lemma}[{\cite[Theorem]{aziz_random_2019}}%
  ]\label{lem:support-stable}
  Let $G$ be a bipartite graph with satisfaction~$\sat$, %
  $M$ be an \osm for $(G, \sat)$, and $(M_j)_{j \in [k]}$ be a support for~$M$.
  Then, each $M_j$, $j\in [k]$, is \mbox{(integrally) stable}. %
\end{lemma}
\iflong
\begin{proof}
  Let $x_1, x_2, \ldots, x_k \in \R_{>0}$ such that $\sum_{j \in [k]}x_j = 1$ and for each edge $e \in E(G)$ we have $M(e) = \sum_{j \in [k]}x_j M_j(e)$.
  Fix an arbitrary edge $\{u, v\} \in E(G)$.
  We show that for each $j \in [k]$ edge $\{u, v\}$ does not (integrally) block~$M_j$.
  It suffices to show that $M_j(u,\wpref v) = 1$ or $M_j(v,\wpref u) = 1$ since $M_j$ is an integral matching.

  By the \ordinalstability\ of $M$ we have
    $M(u,\wpref v) = 1, %
    $ or $M(v,\wpref u) = 1.$ %
  
    Say the latter holds:
    \begin{align}
      M(v,\wpref u) = 1;\label{eq:support2}
    \end{align}
    \noindent the proof if the former holds is analogous.
  If we can show that $M_j(v, \wpref u)=1$ for each $j\in [k]$, then we achieve what we wanted to show, namely that each $M_j$, $j\in [k]$, is stable.

  Thus, it remains to show that $M_j(v,\wpref u)=1$ for each $j\in [k]$.
  By the definition of supports and by \eqref{eq:support2}, we have
  \begin{align}
    1=\sum_{\mathclap{{w \in \bettere{v}{u}}}}~~~~\sum_{j \in [k]} x_j \cdot M_j(v,w) & = \sum_{\mathclap{j \in [k]}} x_j\!\cdot\! \left(\sum_{\mathclap{~~~~~~~w \in \bettere{v}{u}}} M_j(v,w)\right)\nonumber \\
     & = \sum_{\mathclap{j \in [k]}} x_j\!\cdot\! M_j(v,\wpref u).  \label{eq:support-stable1} 
  \end{align}
  Since each~$M_j$ is an integral matching, it holds that $M_j(v,\wpref u)\in \{0,1\}$.
  Since each $x_j$ is a positive real value with $\sum_{j\in [k]}x_j=1$, it must hold that
   $M_j(v,\wpref u)=1$ as otherwise inequality~\eqref{eq:support-stable1} could not hold.
   Thus, for each $j \in [k]$ we have $M_j(v,\wpref u) = 1$, which implies that $\{u, v\}$ does not (integrally) block~$M_j$, as required.
 \end{proof}
 \fi

\noindent\looseness=-1
The reverse of \cref{lem:support-stable} does not hold, i.e.,
a convex combination of stable integral matchings is not necessarily \ostable or \cstable.
This is shown in \cref{ex:cstable-not-always-ostable}, where $M_1$ is a convex combination of stable integral matchings, but it is neither cardinally nor \ostable.
This has also been observed for \ordinalstability\ by \citet{aziz_random_2019}; see their Theorem~3.

Now a replacement argument shows that the maximum achievable welfare and the maximum number of fully matched agents for integral stable matchings are also the maximum for \osm{s}, due to the following.

\begin{lemma}[\appsymb]\label{lem:wel-linear}
  Let $G$ be a bipartite graph with satisfactions~$\sat$,
  $M$ be a matching for $(G, \sat)$, and $(M_j)_{j \in [k]}$ be a support of~$M$ with the coefficients~$x_1, x_2, \ldots, x_k \in \R_{> 0}$.
  Then, the following hold:
    $\wel(M) = \sum_{j\in [k]} \big(x_j\cdot \wel(M_j)\big)$ and 
    $\fullymatched(M) \le \max_{j\in [k]} \fullymatched(M_j)$.
  \end{lemma}
  \iflong
\begin{proof}
  Let $G,\sat,M, (M_j)_{j\in [k]}, x_1,\ldots,x_k$ be as defined in the statement.
  By the definition of supports, we have
  \allowdisplaybreaks
  \begin{align*}
     \wel(M) & = \sum_{v \in V}\util_{M}(v) \nonumber  \\ 
    & = \sum_{v \in V} \sum_{u \in N_G(v)} \sat(v, u)\cdot M(u, v) \nonumber \\
    & = \sum_{v \in V} \sum_{u \in N_G(v)} \sat(v, u)\cdot \left( \sum_{j \in [k]} x_j\cdot M_j(u, v) \right) \nonumber  \\ 
    & = \sum_{v \in V} \sum_{u \in N_G(v)} \sum_{j \in [k]} x_j\cdot \sat(v, u) \cdot M_j(u, v) \nonumber \\
    & = \sum_{j \in [k]} x_j \cdot \sum_{v \in V} \sum_{u \in N_G(v)} \sat(v, u)\cdot M_j(u, v) \nonumber \\
    & = \sum_{j \in [k]} x_j \cdot \wel(M_j) \text{.}
  \end{align*}
  This shows the first statement.

  As for the the number of fully matched agents, again, by the definition of supports, we have
  \begin{align*}
        \fullymatched(M) & \le  \sum_{v \in V}\sum_{u\in N_G(v)} M(u,v)\\
      & =  \sum_{v \in V} \sum_{u \in N_G(v)} \sum_{j \in [k]} x_j\cdot M_j(u, v)  \\ 
      & = \sum_{j \in [k]} x_j \cdot \sum_{v \in V} \sum_{u \in N_G(v)} M(u,v)\\
      & = \sum_{j\in [k]} x_j \cdot \fullymatched(M_j)\text{.}
  \end{align*}
  The last equality holds since every~$M_j$ in the support is an integral matching.
  Now, let $M_r$ with $r\in [k]$ denote the matching with maximum number of fully matched agents among all matchings in the support~$(M_j)_{j\in [k]}$.
  Then, we can continue as follows:
  \begin{align*}
    \fullymatched(M) & \le \sum_{j\in [k]} x_j \cdot \fullymatched(M_j)\\
                     & =   x_r\cdot \fullymatched(M_r) +\left(\sum_{j\in [k]\setminus \{r\}} x_j\cdot \fullymatched(M_j)\right)\\
    & \le  x_r\cdot \fullymatched(M_r) +\left(\sum_{j\in [k]\setminus \{r\}} x_j\cdot \fullymatched(M_r)\right)\\
    & = \fullymatched(M_r).\qedhere
  \end{align*}
\end{proof}
\fi
Using \cref{lem:wel-linear} we may assume that any optimal \ostable matching is an optimal
\emph{integrally} stable matching since using a simple exchange argument, we may swap out matchings in the support of a fractional matching for integral matchings with maximum welfare or with maximum number of fully matched agents in order to decrease the number of matchings in the support until only one remains.
Since finding an optimal stable integrally matching for bipartite graphs with strict preferences is polynomial-time solvable~\cite{irving_efficient_1987}, we immediately obtain the same for \ordinalstability.
\begin{theorem}[\appsymb]\label{thm:egal-bipartite-poly}
  For bipartite graphs with strict preferences, \EOSFM\ and \POSFM are polynomial-time solvable.
\end{theorem}

\iflong \begin{proof}
  \newcommand{\optfrac}{\ensuremath{\textsf{optfrac}}}
  \newcommand{\Mfrac}{\ensuremath{M_\textsf{frac}}}
  \newcommand{\optint}{\ensuremath{\textsf{optint}}}
  To show the statement, we show that for bipartite graphs with strict preferences, the maximization variants of \POSFM and \EOSFM can be solved in polynomial time.
  To this end, let $G=(U \cup W, E)$ be a bipartite graph with strict preferences~$\sat$ and let $\Pot=(\succ_x)_{x\in U\cup W}$.
 
  We first consider the maximization variant of \POSFM. %
  By \cref{lem:wel-linear}, it suffices to find an integral stable matching of $(G,\Pot)$ which has the maximum number of matched agents among all stable integral matchings.
  Now, observe that for strict preferences, every stable integral matching matches the same set of agents~\cite{GusfieldIrving1989}.
  This means that every stable integral matching of $(G,\Pot)$ fulfills our requirements.
  Hence, we can simply use Gale and Shapley's extended algorithm to find a stable integral matching.

  Next, we consider the maximization variant of \EOSFM. %
  Let \optfrac\ be the maximum welfare of an \osm for $(G, \sat)$ and \optint\ be the maximum welfare of a stable integral matching for $(G, \Pot)$.
  We first claim that $\optfrac = \optint$.
  It is clear that $\optfrac \geq \optint$.
  For the other direction, let \Mfrac\ be a \osm for $(G, \sat)$ that has welfare~$\optfrac$.
  By \cref{prop:marriage:support}, let $(M_j)_{j \in [k]}$, together with $x_1, x_2, \ldots, x_k \in \R_{> 0}$ such that $\sum_{j \in [k]}x_j = 1$, be a support of~$\Mfrac$.
  By \cref{lem:support-stable} for each $j \in [k]$ we have that $M_j$ is stable.
  Let $r \in [k]$ such that $M_r$ has maximum $\wel$ among $\{M_j \mid j \in [k]\}$.
  By \cref{lem:wel-linear} we have $\wel(\Mfrac) = \sum_{j \in [k]}x_j\wel(M_j)$.
  Thus,
  \begin{align*}
    \wel(M_r) &= x_r \cdot \wel(M_r) + (1 - x_r)\cdot \wel(M_r) \\
              &\geq \sum_{j \in [k]}x_j\cdot \wel(M_j) = \wel(\Mfrac).
  \end{align*}
  Since $M_r$ is stable, we have $\optfrac = \optint$.

  Now, to find an \osm with maximum welfare it suffices to find a stable integral matching with maximum welfare.
  Since a stable integral matching with maximum welfare has also achieved the minimum \myemph{egalitarian cost}, %
  and since a stable integral matching with minimum egalitarian cost can be found in~$O((|U|+|W|)^4)$~time~\cite{irving_efficient_1987},
  we can find an \osm with maximum welfare in $O((|U|+|W|)^4)$~time.
\end{proof}

\noindent
\textbf{Remark.} The proof for \cref{thm:egal-bipartite-poly} also shows that for bipartite graphs without ties, finding an \osm with maximum sum of matching values can be done in polynomial time.
\medskip

\fi

\iflong\paragraph{Non-bipartite graphs with strict preferences.} \fi 
The tractability result of \POSFM for bipartite graphs with strict preferences heavily utilizes the fact that each fractional matching of a bipartite graph is a convex combination of integral matchings.
This fact, however, does not hold for non-bipartite graphs.
Nevertheless, we can extend the polynomial-time result to the non-bipartite case\iflong, using the first phase of Irving's polynomial-time algorithm for finding a stable integral matching~(\citeyear{Irving1985})~(see \cref{alg:phase1})\fi.
The correctness is based on the following.

\begin{lemma}[\appsymb]\label{ordinal:phase1correct}
  Let $G=(V,E)$ be a graph with cardinal and strict preferences~$\sat$, and let $M$ be an \osm of~$(G,\sat)$.
  The following hold for each agent~$y\in V$:
  \begin{compactenum}[(1)]
    \item\label{lem:osm-mostpreferred}  %
    $M(x,\wpref y)=1$, where $x$ is the most preferred agent of~$y$.
    \item\label{lem:osm-equivalence} \ifshort The following are equivalent:\else
    The following three statements are equivalent: \fi
\ifshort    \begin{inparaenum}[(i)]
\else   \begin{compactenum}[(i)] \fi
      \item\label{lem:osm:matched} $y$ is matched in $M$;
      \item\label{lem:osm:nonsingleton} $y$ is \emph{not} a singleton in a stable partition of $(G,\sat)$;
      \item\label{lem:osm:fully} $y$ is fully matched in $M$.
\ifshort    \end{inparaenum}\else
    \end{compactenum}\fi
  \end{compactenum}
\end{lemma}

\iflong\begin{proof}\else\begin{proof}[Proof sketch]
    \fi
    \looseness=-1
  Let $G=(V,E)$, $\sat$, $M$, and $y$ be as in the statement. 
  To show \eqref{lem:osm-mostpreferred}, suppose, for the sake of contradiction, that $M(x,\wpref y)<1$, where $x$ is the most-preferred agent of~$y$.
  This implies $M(y,\wpref x)$${}={}$$M(y,x)<1$,  %
   a contradiction to $M$ being \ostable\ regarding edge~$\{x,y\}$.
  \ifshort
  The equivalence of the statements in \eqref{lem:osm-equivalence} is based on Statement~\eqref{lem:osm-mostpreferred}. %
  We can perform the first phase of Irving's algorithm~\cite{Irving1985}~(see \cref{alg:phase1}) to guarantee after that
  \begin{inparaenum}[(a)]
    \item an agent is a singleton (in any stable partition) if and only if her preference list is empty, and 
    \item every (fully) matched agent in~$M$ has non-empty preference list.
  \end{inparaenum}
  Note that in \cref{alg:phase1}, we use $\ffirst_{\Pot}(x)$ and $\LLAST_{\Pot}(x)$ to refer to the most-preferred agent and the least-preferred agent in the preference list~$\succ_x\in \Pot$ of $x$.
  \else
  In order to show the equivalence of the statements in \eqref{lem:osm-equivalence},
  we need to know which agents are matched by~$M$. %
  To obtain this, let us first utilize Statement~\eqref{lem:osm-mostpreferred} to repeatedly ``delete'' pairs from $\Pot$ that will not be part of and will not block any \osm.  
  We need some more notations.
  Since we will modify the preference lists in $\Pot$, we use $\ffirst_{\Pot}(x)$ and $\LLAST_{\Pot}(x)$ to refer to the most-preferred agent and the least-preferred agent in the preference list~$\succ_x\in \Pot$ of $x$.
  Now, if $(G,\Pot)$ admits an \osm, say~$M$,
  then $M$ cannot match some pair~$\{x,z\}$ for which there exists an agent~$y$ with $x=\ffirst_{\Pot}(y)$
  such that $x$ prefers $y$ to~$z$ as this violates Statement~\eqref{lem:osm-mostpreferred} regarding~$\{x,y\}$.
  This means that we can repeatedly delete such pairs.
  A pseudocode description of the above approach is given in the while loop of \cref{alg:phase1}.
  We aim to show that no matched pair of $M$ is deleted in \cref{alg:phase1}.
  Clearly, after the first iteration in the while loop of \cref{alg:phase1}~(see lines~\ref{alg:phase1-while-cond}--\ref{alg:phase1-while-end}), no pair~$e=\{i,j\}$ with $M(e)>0$ are deleted; we call such pairs \myemph{matched pairs}. %
  Using the above reasoning successively, we know that in every iteration, no matched pairs of $M$ are deleted.
  In other words, after the while loop, no matched pairs of $M$ are deleted.
  This also means that after the while loop, every agent matched under~$M$ must contain at least one agent in her preference list.
  Hence, if an agent's preference list becomes empty after the while loop, then no \osm will match her.
  Since an agent is a singleton if and only if its preference list becomes empty after the while loop~\citet{Tan1991},
  no \osm will match a singleton agent.

  This shows ``Statement~\eqref{lem:osm:matched} $\Rightarrow$ \eqref{lem:osm:nonsingleton}''.
  
  It remains to show that ``Statement~\eqref{lem:osm:nonsingleton} $\Rightarrow$ \eqref{lem:osm:fully}'' since ``Statement~\eqref{lem:osm:fully} $\Rightarrow$ \eqref{lem:osm:matched}'' clearly holds.
  To this end, let $\mathcal{S}$ denote the set of agents returned from \cref{alg:phase1} on input~$(G,\sat)$,
  and let $\Pot$ denote the modified preference lists after the execution of \cref{alg:phase1}.
  Clearly, for each $x \in V\setminus \mathcal{S}$ there exists an agent~$y$ with $x=\ffirst_{\Pot}(y)$ and
  $y=\LLAST_{\Pot}(x)$.
  By Statement \eqref{lem:osm-mostpreferred} it must hold that $M(x,\wpref y) = 1$.
  In other words, each agent~$x\in V\setminus\mathcal{S}$ must be fully matched under~$M$.
  Since by \citet{Tan1991}, each agent~$x \in V$ is a non-singleton if and only if $x\in V\setminus \mathcal{S}$, it follows that
  each non-singleton is fully matched under~$M$.
  This completes the proof of the equivalence of the statements in \eqref{lem:osm-equivalence}.
  \fi  
\end{proof}

\begin{algorithm}[t!]
  \DontPrintSemicolon
  \caption{Irving's phase 1 algorithm on input~$(G=(V,E)), \sat)$} %
  \label[algorithm]{alg:phase1}
  \small
  Compute the preference lists~$\Pot$ from $\sat$\label{alg:phase1-osm-noties}

  \While{$\exists y\!\in\!V$ \iflong with \else w.\ \fi non-empty pref.\ s.t.~$\LLAST_{\Pot}(\ffirst_{\Pot}(y)) \!\neq\!y$\label{alg:phase1-while-cond}}{%
    $D \leftarrow \{\{\ffirst_{\Pot}(y), z\} \mid \ffirst_{\Pot}(y) \text{ prefers } y \text{ to }z  \text{ in } \Pot\}$
    
    $\Pot \leftarrow \Pot-D$ \label{alg:phase1-while-end}
  }

  \Return $\{v \in V\mid v \text{ has empty pref.\ list in } \Pot\}$ \label{return:OSM}
\end{algorithm}

\noindent\looseness=-1
Using \cref{ordinal:phase1correct} %
we can show that the \osm returned from \cref{alg:arbitrary-osm}
achieves the maximum number of fully matched agents whenever no ties are present.
\begin{lemma}[\appsymb]\label{osm:noties+perfect:p}
  For graphs with~$n$~vertices and with strict preferences, \POSFM can be solved in $O(n^2)$~time. %
\end{lemma}
\iflong\begin{proof}
  Let $(G,\sat, \fullnum)$ be an instance of \POSFM with $G=(V,E)$ and $\sat$ having strict preferences, %
  and let $\Pot=(\succ_{v})_{v\in V}$ denote the strict preference lists derived from~$\sat$.

  We aim to show that on input~$(G,\sat)$, \cref{alg:arbitrary-osm} returns an optimal \osm~$M^{\pi}$ which maximizes the number of fully matched agents.
  To this end, let $\mathcal{S}$ the set of singleton agents according to \cref{prop:singletons} and let $M$ denote an arbitrary \osm of $(G,\sat)$.
  By \cref{ordinal:phase1correct}\eqref{lem:osm-equivalence}, we know that no agent in $\mathcal{S}$ is matched under $M$.
  In other words, every agent matched under~$\mathcal{S}$ comes from $V\setminus\mathcal{S}$.
  Since by \cref{osroommates:always}, every agent from $V\setminus \mathcal{S}$ is fully matched under $M^{\pi}$, we obtain that $M^{\pi}$ maximizes the number of fully matched agents.

  Hence, to solve \POSFM,  we only need to compare whether $\fullnum \le |M^{\pi}|$, and return yes if and only if this is the case.
  The running time of solving \POSFM comes from using \cref{alg:arbitrary-osm} which can be done in $O(|V|^2)$~time due to \cref{osroommates:always}.
\end{proof}

By \cref{ordinal:phase1correct}, we also obtain the same structural property for \ordinalstability as for integral stability.

\begin{observation}\label{cor:osm-matched-unmatched}
  For graphs with strict cardinal preferences, the set of agents is partitioned into two subsets, those that are (fully) matched in every \osm and those that are matched in none.
\end{observation}
\fi

We close this section by observing that the \osm that we obtain from \cref{alg:arbitrary-osm} is also a $2$-approximate solution \iflong for maximizing fully matched agents\fi for both \cardinalstability and \ordinalstability, even\iflong\ for preferences\fi\ with ties.
\begin{proposition}[\appsymb]\label{prop:2-approx}
  \cref{alg:arbitrary-osm} is a $2$-approximation algorithm for the maximization variants of \POSFM and \PCSFM.
\end{proposition}

\iflong\begin{proof}\else\begin{proof}[Proof sketch] \fi
  Let $G=(V,E)$ denote a graph with cardinal preferences~$\sat$ with ties.
  We only show the case of \cardinalstability since the case of \ordinalstability is analogous.
  Let $M^{\csms}$ denote an optimal \csm of $(G,\sat)$
  \iflong with maximum number of fully matched agents.
  \else with~$\fullymatched(M^{\csms})$ being maximum.
  \fi
  By \cref{osroommates:always}, let $M^{\pi}$ denote the \osm returned by \cref{alg:arbitrary-osm} on input~$(G,\sat)$, and let $\Pot$ be the strict preference lists that are used to compute the corresponding stable partition~$\pi$~(see lines \eqref{alg:osm-noties}--\eqref{alg:osm-partition}).
  Recall that in~$M^{\pi}$ every matched agent is fully matched.
  Let $A^{\pi}$ denote the set of (fully) matched agents in $M^{\pi}$; note that $|A^{\pi}|=\fullymatched(M^{\pi})$. 
  \ifshort
  To show that \cref{alg:arbitrary-osm} is a $2$-approximation algorithm for \cardinalstability,
  we observe that at least one agent of each pair of~$E$ must come from $A^{\pi}$ because otherwise this pair is \oblocking~$M^{\pi}$, i.e.,\\  
  {
  $~$\qquad~~$\forall \{x,y\}\in E\colon x\in A^{\pi} \text{ or } y\in A^{\pi}.\hfill~\refstepcounter{equation}(\theequation){\label{eq:2-approx-Api}}$ %
  }
  \else
  
  Since $M^{\pi}$ is a \csm, to show that \cref{alg:arbitrary-osm} is a $2$-approximation algorithm for \cardinalstability, we only need to show that $|\fullymatched(M^{\csms})|\le 2|A^{\pi}|$. %
  We first show that
  each acceptable pair must include some agent from $A^{\pi}$, i.e.,\\
  {
    $~~$~\qquad ~~~$\forall \{x,y\}\in E\colon x\in A^{\pi} \text{ or } y\in A^{\pi}.\hfill~\refstepcounter{equation}(\theequation){\label{eq:2-approx-Api}}$ 
  }
  \fi

  \ifshort
 \noindent Altogether, we derive 
  $\fullymatched(M^{\csms})$$\le$$\sum_{(x,y)\in V\times V}M(x,y) \stackrel{\eqref{eq:2-approx-Api}}{=} \sum_{\substack{x\in A^{\pi},\\ y \in N_G(x)}}M(x,y) + \sum_{\substack{x \in V\setminus A^{\pi},\\ {\color{red}y\in N_G(x)\cap A^{\pi}}}}M(x,y)
     \le    2\cdot |A^{\pi}|. %
  $
  \else Consider an arbitrary pair~$\{x,y\}\in E$. %
  Since $M^{\pi}$ is an \osm of $(G,\sat)$, it follows that $M^{\pi}(x,\wpref) = 1$ or $M^{\pi}(y,\wpref x) =1$.
  In other words, $x\in A^{\pi}$ or $y\in A^{\pi}$ since an agent is fully matched under~$M^{\pi}$ if and only if she is in~$A^{\pi}$(see \cref{osroommates:always}).
  
  \noindent Altogether, we derive that 
  \allowdisplaybreaks
  \begin{align*}
    \fullymatched(M^{\csms})&\le \sum_{x\in V\times V}M(x,y)\\
                        &    = \sum_{\substack{x\in A^{\pi},\\ y \in N_G(x)}}M(x,y) + \sum_{\substack{x \in V\setminus A^{\pi},\\ y\in N_G(x)}}M(x,y)\\
                        &    \stackrel{\eqref{eq:2-approx-Api}}{=} \sum_{\substack{x\in A^{\pi},\\ y \in N_G(x)}}M(x,y) + \sum_{\substack{x \in V\setminus A^{\pi},\\ {\color{red}y\in N_G(x)\cap A^{\pi}}}}M(x,y)\\
    & \le    2\cdot |A^{\pi}|. %
  \end{align*}
  
  Since \cref{alg:arbitrary-osm} runs in $O(|V|^2)$~(see \cref{osroommates:always}), it is a $2$-approximation for the maximization variant of \PCSFM.
  \fi
\end{proof}
\section{Hardness results}
We first give our results for \ostable\ matchings and then turn to \cstable\ matchings.

\myparagraph{Hardness for optimal \ostable matchings} %
The structural property that for \iflong the marriage case (i.e., bipartite graphs), \else
bipartite graphs, \fi
\ostable matchings can be decomposed into a convex combination of stable \emph{integral} matchings (\cref{lem:support-stable,lem:wel-linear}) implies that
\POSFM is equivalent to finding a stable integral matching with maximum cardinality~(\textsc{Max-Card SMTI}),
and \EOSFM is equivalent to finding a stable integral matching with minimum egalitarian cost~(\textsc{Min-Egal SMT}); the egalitarian cost and the welfare of an integral matching are dual to each other.
Since both problems are known to be NP-hard~\citet{MaIrIwMiMo2002}, we obtain the following.
\begin{theorem}[\appsymb]\label{th:Perfect-Welfare-ties-OSM_hard}
  When the preferences have ties, \POSFM and \EOSFM become \NPC, even for bipartite graphs.
\end{theorem}
\iflong
\begin{proof}
  We first tackle \POSFM by reducing from the \NPH problem~\textsc{Max-Card SMTI}~\cite{MaIrIwMiMo2002}.
  \textsc{Max-Card SMTI} is the problem of deciding whether, given
  a bipartite graph~$G=(U\cup W,E)$ with preference lists~$\Pot=(\wpref_x)_{x\in U\cup W}$ 
  such that
  each agent~$u\in U$ (resp.\ $w\in W$) has a preference list over~$N_G(u)$ (resp.\ $N_G(w)$), 
  and a non-negative integer~$\fullnum$,
  there exists a stable \emph{integral} matching of cardinality at least~$\fullnum$.
  Recall that an integral matching~$M$ is \myemph{stable} if no two agents form a blocking pair,
  and two agents~$u,w$ form a \myemph{blocking pair} of~$M$ if
  \begin{compactenum}[(i)]
    \item $\{u,w\}\in E$ and 
    \item $M(u, \wpref w) + M(w, \wpref u) = 0$.
  \end{compactenum}

  Let $I=(G,\Pot, \fullnum)$
  be an instance of \textsc{Max-Card SMTI} and let $\sat$ be the cardinal preferences derived from $\Pot$, i.e.,
  for each $x\in U\cup W$ and each $y\in N_G(x)$ define
  \begin{alignat}{3}
    \nonumber    \forall (x,y)\in & (U\cup W)\times (U\cup W)\colon\\ 
    \label{maxsmti-rank}
    \rank_{\Pot}(x,y) \coloneqq &
    \left\{
    \begin{array}{ll}
      |\{y'\mid y'\pref_x y\}| ,&\text{if } \{x,y\}\in E \\
     |N_G(x)|, & \text{otherwise,}\\
    \end{array}
    \right. \\
    \label{maxsmti-sat}
    \sat_{\Pot}(x,y)  \coloneqq & |N_G(x)| - \rank_{\Pot}(x,y).
  \end{alignat}
  Note that this definition satisfies Condition~\eqref{eq:Pot} so that the preference lists derived from $\sat_{\Pot}$ are equivalent to~$\Pot$.
  We aim to show that $(G,\sat_{\Pot})$ admits a stable integral matching with cardinality~$\fullnum$ if and only if $(G,\sat_{\Pot})$ admits an \osm where at least $\fullnum$ agents are fully matched.
  The ``only if'' direction is clear since every stable integral matching of $(G, \Pot)$ is an \osm of $(G,\sat_{\Pot})$. %
  For the ``if'' direction, let $M$ be an \osm of $(G, \sat_{\Pot})$ with $\fullymatched(M) \ge \fullnum$.
  By \cref{prop:marriage:support}, let $(M_j)_{j\in [k]}$ be a support of~$M$.
  By \cref{lem:wel-linear}, $\max_{j\in [k]}\fullymatched(M_j) \ge \fullymatched(M) \ge \fullnum$.
  Since every~$M_j$, $j\in [k]$, is stable by \cref{lem:support-stable},
  there exists an integral stable matching of $(G, \sat_{\Pot})$ with cardinality at least~$\fullnum$.
  This completes the proof for showing that \POSFM is NP-hard.
  
  The proof for \EOSFM works similarly.
  Instead of reducing from \textsc{Max-Card SMTI}, we reduce from the \NPH \textsc{Min-Egal SMT} problem~\cite{MaIrIwMiMo2002}.
  \textsc{Min-Egal SMT} is the problem of deciding whether, given
  a complete bipartite graph~$G=(U\cup W,E)$ with complete preference lists~$\Pot=(\wpref_x)_{x\in U\cup W}$
  such that
  each agent~$u\in U$ (resp.\ $w\in W$) has a preference list over all agents from~$W$ (resp.\ $U$), 
  and a non-negative integer~$\gamma'$,
  there exists a stable \emph{integral} matching of \myemph{egalitarian cost} at most~$\gamma'$.
  Here, the \myemph{egalitarian cost} of an integral matching~$M$ is defined as
  \begin{align*}
    \egal_{\Pot}(M) \coloneqq \sum_{x\in U\cup W}\rank_M(x). %
  \end{align*}
  Note that for complete preference lists, each stable integral matching must be perfect.

  Now, we observe that if we use the satisfaction function~$\sat_{\Pot}$ given in \eqref{maxsmti-sat},
  then each each integral matching~$N$ of $(G, \sat_{\Pot})$ satisfies the following duality:
  \begin{align}
   \egal_{\Pot}(N) + \wel_{\sat_{\Pot}}(N) = 2|E| \label{eq:egal-welfare}
  \end{align}
  Hence, a stable integral matching has maximum social welfare if and only if it has minimum egalitarian cost (among all stable integral matchings).
  This means that finding a stable integral matching with maximum social welfare is NP-hard.
  Hence, to show \NPH{ness} for \EOSFM, it suffices to show that
  $(G,\Pot)$ admits a stable integral matching with welfare~$\welfare$ if and only if
  $(G,\sat_{\Pot})$ admits an \osm with welfare~$\welfare$.
  The ``only if'' direction is clear since every stable integral matching of~$(G, \Pot)$ is an \osm of~$(G,\sat_{\Pot})$.

  For the ``if'' direction, let $M$ be an \osm of~$(G, \sat_{\Pot})$ with $\wel(M)\ge \welfare$.
  By \cref{prop:marriage:support}, let $(M_j,x_j)_{j \in [k]}$ be a convex combination of $M$. %
  To show the ``if'' direction, we only need to show that $\max_{j}(\wel(M_j))\ge \wel(M)$.
  Similarly to the proof of \cref{thm:egal-bipartite-poly}, by \cref{lem:wel-linear} and by the property of convex combinations
  $\wel(M) = \sum_{j\in }x_j\cdot \wel(M_j) \le \max_{j\in [k]}\wel(M_j)$.
  This completes the proof for \EOSFM.
\end{proof}

\noindent \textbf{Remark.} The proof for \cref{th:Perfect-Welfare-ties-OSM_hard} also implies NP-hardness when we instead aim to find an \ostable matching with maximum sum of matching values.

\smallskip

\fi

\citeauthor{Feder1994}~(\citeyear{Feder1994}) showed that finding a maximum-welfare stable integral matching in non-bipartite graphs is \APXH, even if no ties are present.
In the next theorem we show that the idea behind the NP-hardness reduction of \citeauthor{Feder1994} with some additional analysis yields the same inapproximability result for \ordinalstability.
\begin{theorem}[\appsymb]\label{th:Welfare-OSR_hard}
  The maximization variant of \EOSFM is \APXH\ and \EOSFM is \NPC, even if no ties are present.
\end{theorem}%
\iflong
\begin{proof}
  To show the inapproximability, we reduce from the maximization variant of \IS, the \MaxIS problem, which is defined as follows.
  \probdefopt{\MaxIS~(\MaxISs)}
  {A graph $G=(V,E)$.}
  {Find a maximum-cardinality independent set of~$G$;
  here, an \myemph{independent set} of $G$ is a vertex subset~$X\subseteq$ such that $G[X]$ is edgeless.}

  \MaxIS is \APXH\ and hence \NPH, even for cubic graphs~\cite{AliKann2000apx-cubic}.
  Let $G$ be a cubic graph with vertex set~$V=\{v_1,v_2,\ldots,v_n\}$. %
  The basic idea is to construct an instance of \EOSFM where for each vertex~$v_i$ of $G$, there are essentially \emph{two} possible matchings, say $M^{\lis}_i$ and $M^{\lvc}_i$, which each may lead to some \ostable matching.
  However, $M^{\lis}_i$ has a higher welfare than $M^{\lvc}_i$, but it is not possible to include~$M^{\lis}_i$ for two adjacent vertices as they will induce an \oblocking pair.

  We create our instance~$I'$ of \EOSFM as follows.
  For each vertex~$v_i \in V$, we create four agents~$u_i,w_i,x_i,y_i$.
  For ease of notation, define \myemph{$L\coloneqq 2n+2m$}.
  We describe the cardinal preferences of the agents in~\cref{fig:Welfare-OSR-card}, where for each vertex~$v_i\in V$, let \myemph{$d_i$} and \myemph{$N(u_i)$} denote the degree of $v_i$ in $G$ and the set of agents~$u_{i'}$ which correspond to the neighbors $v_{i'}$ of $v_i$ in $G$, respectively,
  and let $\pi_i\colon N(u_i) \to [d_i]$ denote an arbitrary but fixed enumeration of the agents in $N(u_i)$; the intent of~$\pi_i$ is to give each ``adjacent'' agent a unique rank so that no ties are present.
  Observe that the constructed preferences indeed contain no ties.

  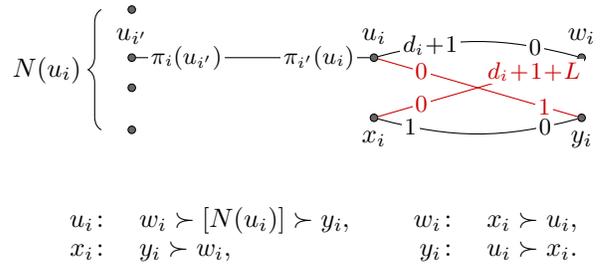
\begin{figure}
      \begin{tikzpicture}
      \def \xscale {2.3}
      \def \yscale {.8}
      \def \ss {.8}
      \def \ee {-1.2}
      \foreach \x / \y / \n / \c / \p in {0/0/{u_i}/ui/above, 1.2/0/{w_i}/wi/above, 0/-1/{x_i}/xi/below, 1.2/-1/{y_i}/yi/below,-1.4/0/{u_{i'}}/uip/above} {
        \node[agent] at (\x*\xscale, \y*\yscale) (\c) {};
        \node[\p = 0pt of \c] {$\n$};
      }

      \foreach \sstep in {\ee, -0.5 0, \ss}{
        \node[agent] at (-1.4*\xscale, \sstep*\yscale) {};
      }

       \draw[decoration={brace,mirror, raise=5pt,amplitude=5pt},decorate] (-1.5*\xscale, \ss*\yscale) -- node[left=2ex] {$N(u_i)$} (-1.5*\xscale, \ee*\yscale);

       \foreach \c / \a / \s /  \t / \n / \p / \m / \q in {black/15/ui/wi/{d_i\!+\!1}/0.25/0/{0.8},
         red!80!black/0/ui/yi/{0}/0.2/1/0.85,
         black/{-15}/xi/yi/1/0.15/0/0.85,
         red!80!black/0/xi/wi/0/0.2/{d_i\!+\!1\!+\!L}/0.8,
         black/0/ui/uip/{\pi_{i'}(u_i)}/0.2/{\pi_{i}(u_{i'})}/0.8}
      {
        \draw[] (\s) edge[bend left=\a,color=\c] node[pos=\p,fill=white, inner sep=1pt] {\footnotesize $\n$} node[pos=\q,fill=white,inner sep=1pt] {\footnotesize $\m$} (\t);
      }
    \end{tikzpicture}
    \[
      \begin{array}{rlrl}
        \\
        ~u_i\colon & w_i \succ [N(u_i)] \succ y_i, ~~~~~ & w_i\colon & x_i \succ u_i,\\
        ~x_i\colon & y_i \succ w_i, ~~~~~ & y_i\colon & u_i \succ x_i.
      \end{array}
    \]
    \caption{Top: The cardinal preferences of the four agents~$u_i,x_i,w_i,y_i$ which correspond to vertex~$v_i$, constructed in the proof of \cref{th:Welfare-OSR_hard}.
      Here, agent~$u_{i'}$ corresponds to a neighbor $v_{i'}$ of~$v_i$, i.e., $u_{i'}\in N(u_i)$. Recall that $\pi_i$ and $\pi_{i'}$ are two fixed enumerations of the ``neighboring'' agents from $N(u_i)$ and $N(u_{i'})$, respectively. Bottom: The induced preference lists, where $[N(u_i)]$ denotes an order of the agents in $N(u_i)$ in decreasing order according to~$\pi_i$.}
    \label{fig:Welfare-OSR-card}
  \end{figure}

  This completes the construction of our instance~$I'$.
  Clearly,~$I'$ can be constructed in linear time and every agent has strict preferences.

  We first show the ``equivalence'' in terms of solutions.
  \begin{claim}\label{approx:equivalence}
    $G$ has an independent set~$V'$ with $|V'| = k$ if and only if the constructed instance~$I'$ admits an \osm~$M$ with $\wel(M) = L\cdot (k+1)$.
  \end{claim}

  \begin{proof}[Proof of \cref{approx:equivalence}]
    \renewcommand{\qedsymbol}{$\diamond$} For the ``only if'' part,
    assume that $G$ has an independent set of cardinality~$k$ and let $V'\subseteq V$ denote such an independent set.
  We form an integral matching~$M$ as follows.
  \begin{enumerate}[(i)]
    \item\label{it:IS} For each $v_i\in V'$, let $M(u_i,y_i)=M(x_i,w_i)=1$.
    \item\label{it:notIS} For each $v_i\in V\setminus V'$, let $M(u_i,w_i)=M(x_i,y_i)=1$.
    \item For each remaining acceptable edge~$e=\{a,b\}$ not mentioned in~\eqref{it:IS}--\eqref{it:notIS}, let $M(a,b) = 0$.
  \end{enumerate}
  By construction, matching~$M$ has the following social welfare:
  \allowdisplaybreaks
  \begin{align*}
    \wel(M)  &= \sum_{\mathclap{v_i\in V'}} \big(\sat(u_i,y_i)\!+\!\sat(x_i,w_i)\big) \\
    & \qquad + \sum_{\mathclap{v_i \in V\setminus V'}} \big(\sat(u_i,w_i)\!+\!\sat(x_i,y_i)\big)\\
     &=  \sum_{\mathclap{v_i\in V'}} (d_i+2+L) + \sum_{\mathclap{v_i \in V\setminus V'}} (d_i+2)\\
    &= 2n+2m+|V'|\cdot L \stackrel{{\scriptsize L=2n+2m}}{=} L\cdot (k+1).
  \end{align*}
  It remains to show that $M$ is \ostable.

  Clearly, for each $v_i\in V \setminus V'$, neither~$u_i$ nor $x_i$ is involved in an \oblocking pair since they already receive their most preferred agents, respectively.
  Analogously, for each~$v_i\in V'$, neither $w_i$ nor $y_i$ is involved in an \oblocking pair. %
  Hence, for each~$i\in [n]$, no two agents from $\{u_i,w_i,x_i,y_i\}$ form an \oblocking pair.
  
  Finally, suppose that $\{u_i,u_{i'}\}$ is an \oblocking pair of $M$.
  By the definition of~$M$,~$u_i$ and $u_{i'}$ appear in each other's preference lists and that $M(u_i,y_i) = M(u_{i'},y_{i'})=1$.
  By the cardinal preferences and by the definition of~$M$,~$v_i$ and $v_i'$ are neighbors in $G$, and thus $v_i,v_{i'}\in V'$, a contradiction to $V'$ being an independent set.

  For the ``if'' part, assume that $I'$ admits an \osm with $\wel(M)=L\cdot (k+1)$, and let $M$ be such a matching.
  We claim that the following vertex set~$V'\coloneqq \{v_i \in V \mid M(u_i,y_i) > 0\}$ is an independent set of size at least~$k$.

  To show that $V'$ is an independent set, let us consider an arbitrary vertex~$v_i\in V'$.
  We aim to show that no neighbor of~$v_i$ belongs to~$V'$.
  By applying \cref{ordinal:phase1correct}~\eqref{lem:osm-mostpreferred} three times (setting $y=u_i$, $y=x_i$, $y=y_i$, respectively), we obtain that
  \begin{align}
   \nonumber M(w_i, x_i) \!+\! M(w_i,u_i)  &= M(y_i,u_i)+M(y_i,x_i)\\
     &= M(x_i, y_i) + M(x_i,w_i) =1.\label{wiyixi}
  \end{align}
  This means that
  \begin{align}
    M(u_i,w_i)+M(u_i,y_i) = 1. \label{ui}
  \end{align}
  In other words, $M(u_{i}, u_{i'})=0$ for each $u_{i'}$ with $v_{i'}$ being the neighbor of $v_i$, and
  $V'$ is indeed an independent set.
  
  It remains to consider the size~$|V'|$. %
  To ease notation, for each~$v_i\in V$, let $\alpha_i = M(u_i,y_i)$ with $0 \le \alpha_i \le 1$.
  Then, by \eqref{wiyixi}--\eqref{ui},
  we have that $M(u_i,y_i) = M(x_i,w_i) = \alpha_i$ while $M(u_i,w_i) = M(x_i,y_i) = 1-\alpha_i$.
  Hence,
  \begin{align*}
    L\cdot (h+1) & =  \sum_{v_i \in V} \big( (1-\alpha_i) \cdot (d_i+1+0+1+0)+\\
    &~~~~~~~~~~~\alpha_i\cdot (0+1+0+d_i+1+L)\big)\\
    & =  \sum_{v_i\in V} \big(d_i+2 \big) + L\cdot \sum_{v_i\in V} \alpha_i\\
    & \stackrel{{\scriptsize L=2n+2m}}{=}  L+ L\cdot \sum_{v_i\in V}\alpha_i.
  \end{align*}
  Since $\alpha_i \le 1$, the above implies that there must be at least~\(k\) vertices~$v_i\in V$ such that $\alpha_i > 0$.
  In other words, $|V'|\ge k$.
  Since the independent set property is closed under subsets, we delete some vertices from $V'$ to let it contain exactly $k$ vertices if $|V'|>k$.
  \end{proof}

  Now, suppose that we have a polynomial-time algorithm which approximates the maximization variant of \EOSFM within factor~$\varepsilon$.
  That is, on input~$(G,\sat)$ and a positive approximation error value~$\varepsilon$ with $0 < \varepsilon < 1$,
  the algorithm returns an \osm~$M$ such that $\wel(M)\ge \varepsilon \cdot \welfare^*$ where $\welfare^*$ denotes the maximum welfare of all \osm{s} of $(G,\sat)$.
  If the maximum-cardinality independent set of $G$ has at least $k^*$ vertices, then 
  by \cref{approx:equivalence}, our instance~$I'$ has an \osm~$M^*$ with
    $\wel(M^*)=L\cdot (k^*+1)$.
  Hence, the approximation algorithm finds an \osm~$M$ with %
  \begin{align}\label{eq:wel-approx}
    \wel(M)\ge \varepsilon \cdot L\cdot (k^*+1).
  \end{align}
  Again, by \cref{approx:equivalence}, the approximation algorithm also finds an independent set~$V'$ for $G$ with $|V'|=k$ and %
  \begin{align*}
    k = \frac{\wel(M)}{L} - 1 \stackrel{\eqref{eq:wel-approx}}{\ge} \varepsilon k^* + \varepsilon -1 > \varepsilon k^*-1.
  \end{align*}
  Since $k$ is an integer, it follows that $k\ge \varepsilon \cdot k^*$.
  In other words, the approximation algorithm can also be used to approximate \MaxIS to an arbitrary factor, a contradiction.

  Since \MaxISs is also NP-hard, the same reduction shows that \EOSFM is NP-hard, even for strict preferences.
\end{proof}

\noindent \textbf{Remark.} Note that in the instance created in the reduction for~\cref{th:Welfare-OSR_hard} the number of agents acceptable to each agent is bounded by five.
This implies that the both APX-hardness and \NPH{ness} remain even for this restricted case.
\smallskip

\fi

\myparagraph{Hardness for optimal \cstable matchings.}
We now prove that \ECSFM and \PCSFM are \NPC even when the input graph $G$ is bipartite and \iflong the values of $\sat(v,\cdot)$ are distinct for each vertex $v$ in $G$, i.e., \fi\sat\ has no ties.  For each of the problems we give a \iflong many-to-one \fi
polynomial-time reduction from the \iflong well-known \fi
\iflong
\NPC problem \IS~\cite{GJ79}.%
\probdef{\textsc{Independent Set}~(\ISs)}
{A graph $G$ and a non-negative integer $k$.}
{Is there a size-at-least-\(k\)-vertex \myemph{independent set} $X$ in $G$, that is, a vertex subset~\(X\) such that $G[X]$ is edgeless?}
\else
\NPC problem \IS~(\ISs)~\cite{GJ79}, wherein we are given a graph $G$ and a non-negative integer $k$ and ask to find a \myemph{independent set} $X$ in $G$ with at least \(k\) vertices, i.e., a vertex subset~\(X\) such that $G[X]$ is edgeless.
\fi

The gadgets used in and the correctness proof of the two reductions have some similar parts.
In fact, we use the same \myemph{edge gadgets}, which we now describe.

\looseness=-1
\ifshort
Let $G=(V,E)$ be a graph with vertex set~$V=\{v_1,\ldots,v_n\}$ and edge set~$E=\{e_1,\ldots,e_m\}$.
The \myemph{edge gadgets} are contained in a bipartite graph $G_E$ with preference function~$\sat$, constructed as follows.
The vertex set $V(G_E)$ is the union of the two disjoint sets $U_E$ and $W_E$ which are defined as follows.
For each $j \in [m]$:
\begin{compactitem}[--]
\item Add vertices $e_j$, $f_j$, and $g_j$ to $U_E$ and a vertex~$h_j$ to $W_E$.
\item Add $\{ e_j^{i}, e_j^{i'}$$\mid$$e_j$$\in$$E, e_j$$=$$\{v_i, v_{i'}\} \text{ and } i,i' \in [n] \}$ to~$W_E$.
\item Let $e_j = \{v_i,v_{i'}\} \in E$. Then add two extra vertices
  $u_i$ and $u_{i'}$ to $U_E$ if not already present.
  They represent the connections of the edge gadget to the \myemph{vertex gadgets}.%
\end{compactitem}
\looseness=-1
The edges of $G_E$ along with the values of \sat\ are for each $j \in [m]$ as described in Figure~\ref{fig:csm_edge_gadget}.
The \sat-values for $u_i$ and $u_{i'}$ will be defined in the hardness construction. 
For each $j \in [m]$ by the \myemph{edge gadget for edge $e_j$} (of $G$) we refer to the subgraph of $G_E$ induced by $\{e_j, f_j, g_j, h_j, e^i_j, e^{i'}_j\}$.
\fi

\iflong
\begin{construction}[Edge gadget for the \cardinalstability]\label{cons:edge-gadget}
  Let $(G=(V,E), k)$ be an instance of \ISs, where $V=$ $\{v_1, v_2, \ldots, v_n\}$ and $E$$=$$\{e_1, e_2, \ldots, e_m\}$ denote the vertex set and the edge set, respectively.
  The \myemph{edge gadgets} are contained in a bipartite graph $G_E$ with preference function~$\sat$.
The \myemph{vertex set $V(G_E)$} is the union of two disjoint sets~\myemph{$U_E$} and \myemph{$W_E$}; we will call the vertices in $G_E$ \emph{agents} to distinguish them from the vertices in~$G$.
For each $e_j\in E$ with $e_j=\{v_i, v_{i'}\}$ do the following:
\begin{compactitem}[--]
  \item Add to~$U_E$ three {agents}~$e_j$, $f_j$, and $g_j$;
  add to~$W_E$ an agent~$h_j$.
\item Add to~ $W_E$ two agents $e_j^{i}$ and  $e_j^{i'}$. %
\item %
Add to~$U_E$ two extra agents~$u_i$ and~$u_{i'}$ to $U_E$ if not already present.
They represent the connections of the edge gadget to the \myemph{vertex gadgets}.%
\end{compactitem}
\looseness=-1
For each $j \in [m]$ by the \myemph{edge gadget for edge~$e_j$} (of $G$) we refer to the subgraph of $G_E$ induced by the agents $\{{e_j}, f_j, g_j\} \cup \{h_j, e^i_j, e^{i'}_j\}$. 
\cref{fig:csm_edge_gadget} illustrates the edge gadget for edge~$e_j\in E$.
More precisely, it depicts the acceptability (sub)graph corresponding to edge~$e_j$, $e_j\in E$, labeled with the cardinal preferences~$\sat$,
and the derived preferences lists.
The cardinal preferences~$\sat$ of the agents~$u_i$ and $u_{i'}$ from the vertex gadgets towards the agents~$e_j^{i}$ and $e_j^{i'}$ are bounded by~$5m$.
The cardinal preferences not mentioned until now are set to zero.
The precise values are irrelevant for now and will be defined in the hardness constructions when we use the edge gadgets.\hfill~$\diamond$
\end{construction}
Note that although we use the symbol~$e_j$ for both an edge in~$G$ from the \ISs instance and an agent in the edge gadget,
the precise meaning will be clear from the context.
\fi
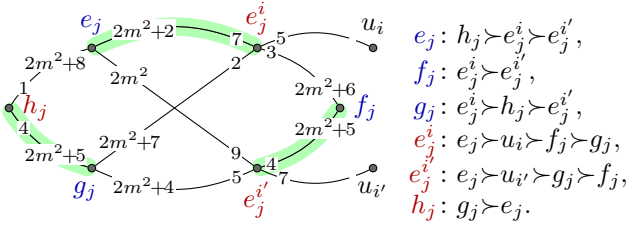
\begin{figure}[t!]
\begin{minipage}{.63\linewidth}
  \begin{tikzpicture}
    \def \xscale {2.2}
    \def \yscale {1.6}
    \def \ss {.8}
    \def \ee {-1.2}
    \foreach \col / \x / \y / \n / \c / \p in {blue/0/0/{e_j}/ej/above, red/1/0/{e_j^{i}}/eji/above, blue/0/-1/{g_j~~}/gj/below, red/1/-1/{e_j^{i'}}/ejip/below,
      red/-.5/-0.5/{h_j}/hj/right,blue/1.5/-0.5/{f_j}/fj/right, black/1.7/0/{u_i}/ui/above, black/1.7/-1/{u_{i'}}/uip/below} {
      \node[agent] at (\x*\xscale, \y*\yscale) (\c) {};
      \node[\p = 0pt of \c,\col!70!black] {$\n$};
    }
    \foreach \c / \a / \s /  \t / \n / \p / \m / \q in
    {black/-20/ej/hj/{2m^2\!\!+\!\!8}/0.35/1/{0.8},
      black/20/gj/hj/{2m^2\!\!+\!\!5}/0.35/4/{0.8},
    black/25/ej/eji/{2m^2\!\!+\!\!2}/0.3/{7}/{.9}, 
    black/-25/gj/ejip/{2m^2\!\!+\!\!4}/0.3/{5}/{.9},
    black/0/ej/ejip/{2m^2}/0.2/{9}/{.9},
    black/0/gj/eji/{2m^2\!\!+\!\!7}/0.2/{2}/{.9},
    black/-20/fj/eji/{2m^2\!\!+\!\!6}/0.2/{3}/{.9},
    black/20/fj/ejip/{2m^2\!\!+\!\!5}/0.2/{4}/{.9},%
    black/25/uip/ejip/{~}/0.15/{7}/{.8},
    black/-25/ui/eji/{~}/0.15/{5}/{.82}}%
    {
        \draw[] (\s) edge[bend left=\a,color=\c] node[pos=\p,fill=white, inner sep=.3pt] {\scriptsize $\n$} node[pos=\q,fill=white,inner sep=.3pt] {\scriptsize $\m$} (\t);
      }
      \begin{pgfonlayer}{background}
      \foreach \i / \j / \ang in {ej/eji/25,hj/gj/-20,fj/ejip/20} {
        \draw[linemarkg] (\i) edge[bend left=\ang] (\j);
        \draw[gray] (\i) edge[bend left=\ang] (\j);
      }
    \end{pgfonlayer}
    \end{tikzpicture}
    \end{minipage} 
    \begin{minipage}{.35\linewidth}
  \[\begin{array}{@{}r@{}l@{}}
        \textcolor{blue!70!black}{e_j}\colon & h_j \pref e_j^i\pref e_j^{i'},\\
        \textcolor{blue!70!black}{f_j}\colon & e_j^i \pref e_j^{i'},  \\
        \textcolor{blue!70!black}{g_j}\colon & e_j^i \pref h_j \pref e_j^{i'}, \\
        \textcolor{red!70!black}{e_j^{i}}\colon & e_j \pref u_i \pref f_j \pref g_j,\\
        \textcolor{red!70!black}{e_j^{i'}}\colon & e_j \pref u_{i'} \pref g_j \pref f_j,\\
        \textcolor{red!70!black}{h_j}\colon & g_j \pref e_j.
      \end{array}
    \]
    \end{minipage}
    \caption{The edge gadget. %
      Left: The cardinal preferences of the six agents~$e_j,f_j,g_j,h_j,e_j^i$, and $e_j^{i'}$ which correspond to edge~$e_j=\{v_i,v_{i'}\}\in E$ with $i<i'$. %
      When constructing a matching from an independent set, we
      integrally match the pairs corresponding to the green edges if $v_i \in V(G)$ is in the independent set.
      Right: The induced preference lists.}\label{fig:csm_edge_gadget}
  \end{figure}

  \iflong
  The edge gadget by \cref{cons:edge-gadget} has two crucial properties for a \csm, one regarding social welfare and another regarding fully matched agents. 
 We will exploit both properties for the hardness proofs. %
 The desired behavior is essentially that, in every \myemph{feasible matching}~$M$ for $G_E$, that is,
 in every \csm that has sufficiently large social welfare or sufficiently many fully matched agents,
 at least one of $\{e^i_j, e^{i'}_j\}$ is \myemph{unsatisfied} with respect to~$M$, i.e.,
 $e^i_j$ or $e^{i'}_j$ would rather like to be integrally matched with $u_i$ or $u_{i'}$, respectively.
 We then use this property to force the matching to assign~$u_i$ or $u_{i'}$ to partners so that at most one of the two assignments signifies that $v_i$ or $v_{i'}$ is supposed to be in the independent set.

 We first summarize this property regarding social welfare. %
 \else
 We now prove the essential property of the edge gadgets: Let $e_j \in E$ with $e_j = \{v_i,v_{i'}\}$. We show that $e^i_j$ or $e^{i'}_j$ are unsatisfied with every feasible matching.
 This will force that not both $v_i$ and $v_{i'}$ are in independent set. %
 \fi

\begin{lemma}[\appsymb]\label{cl:edge-gadget-max-sat}
  Let \(M\) be a \csm\ for \(G_E\).
  Then, the total welfare, \(\omega\), received from \(M\) by the vertices of the edge gadgets is at most $3m(2m^2+9)$.
  Moreover, if there \iflong is an edge~\else exists~\fi \(e_j \in E\) with \(e_j = \{v_i, v_{i'}\}\) such that for both \(\nu \in \{i, i'\}\) we have \(\util_M(u_\nu) < \sat(u_{\nu}, e^\nu_j)\), then \(\omega < 3m(2m^2 + 9) - n\).
\end{lemma}
\iflong
\begin{proof}
  Fix $j \in [m]$ and assume that $e_j=\{v_i,v_{i'}\}$ is an edge from $E$.
  For the sake of readability,
  define $\colU{U_j}\coloneqq \{\colU{e_j},\colU{g_j},\colU{f_j}\}$
  and $\colW{W_j}\coloneqq \{\colW{e^{i}_j},\colW{e^{i'}_j},\colW{h_j}\}$.
  We use \myemph{$G_j$} to refer to the edge gadget corresponding to~$e_j$, where the edges are incident to some agent from~$U_j\cup W_j$.
  Note that the total welfare~$\omega_j$,%
   received from $M$ by the agents of the edge gadget~$G_j$ is $\omega_j = \sum_{x\in \colU{U_j}}\util_{M}(x) + \sum_{y\in \colW{W_j}}\util_M(y)$.

  We first show that the maximum \wel\ given by \(M\) to the agents in~$G_j$ is at most $3(2m^2 +9)$, i.e.,
  \begin{align}
    \omega_j  \le 3(2m^2+9). \label{eq:edgebound}
  \end{align}
  Observe that for each edge~$\{x,y\}\in E(G_j)$ from gadget~$G_j$ we have $\sat(x, y) + \sat(y, x) = 2m^2 + 9$, i.e.,
  \begin{align}
    \nonumber \forall (x,y)\in \colU{U_j}\times \colW{W_j}\text{ with }\{x,y\}\in E(G_j)\colon\\
    \sat(x, y) + \sat(y, x) = 2m^2 + 9.    \label{eq:edge-welfare}
  \end{align}
  Moreover, each edge in $G_j$ is incident with $e_j$, $g_j$, or $f_j$,
  and no edge is incident with two of them since $G_j$ is bipartite.
  Thus, the \wel~$\omega$ achieved by the agents in $G_j$ is bounded as follows:
  \allowdisplaybreaks
  \begin{align}
    \nonumber  \omega_j  = &~~~ \Big( \sum_{\substack{(y,x)\in \colW{W_j}\times \colU{U_j}\colon\\\{y,x\}\in E(G_j)}} M(\{y,x\}) \cdot (\sat(x,y)+\sat(y,x))\Big)  \\
 \nonumber   &  + \sat(e_j^{i}, u_i)\cdot M(e_j^{i},u_i)+\sat(e_j^{i'}, u_{i'})\cdot M(e_j^{i'},u_{i'})  \\
    \nonumber    \stackrel{\eqref{eq:edge-welfare}}{=} &~~       (2m^2 + 9) \cdot \sum_{\mathclap{\substack{(y,x)\in \colW{W_j}\times \colU{U_j}\colon\\\{y,x\}\in E(G_j)}}
    } M(\{y,x\}) + \sat(e_j^{i}, u_i)\cdot M(e_j^{i},u_i)  \\
 \nonumber &  + \sat(e_j^{i'}, u_{i'})\cdot M(e_j^{i'},u_{i'})  \\
    \nonumber
    = &  ~~~ (2m^2\!+\!9) \cdot \Big(~~
        \sum_{\mathclap{x\in \{e_j,g_j\}}}M(\{h_j,x\})+
     \sum_{\mathclap{x\in U_j\cup \{u_i\}}} M(\{e_j^i,x\})\\
  \nonumber  & \phantom{2(m^2\!+\!9) \cdot \Big(~~}  +   \sum_{\mathclap{x\in U_j\cup \{u_{i'}\}}} M(\{e_j^{i'},x\})\Big)    \\
   \nonumber \le & 3(2m^2\!+\!9).%
  \end{align}
  The last inequality holds since $M$ is a fractional matching, implying that the sum of values assigned by~$M$ to
  each agent in~$W_j$ %
  is at most one,
  and $\sat(e_j^i,u_i)$ and
  $\sat(e_j^{i'},u_{i'})$ are bounded by $7$ which are strictly smaller than than $2m^2+9$.
  This shows Inequality~\eqref{eq:edgebound}.
  Indeed, the equation hold only when the matching values assigned to the pairs~$\{e_j^{i}, u_i\}$ and $\{e_j^{i'}, u_{i'}\}$ are both zero.

  To complete the proof,
  we show that if \(\util_M(u_i) < \sat(u_i, e^i_j)\) and \(\util_M(u_{i'}) < \sat(u_{i'}, e^{i'}_j)\), then the total welfare received by~$M$ for the agents in $G_j$ has $\omega\le 3(2m^2+9)-n$.
  Since \(\util_M(u_i) < \sat(u_i, e^i_j)\) and \(M\) is a \csm, we have
  \begin{equation}
    \label{eq:3}
    \util_M(e^i_j) \geq \sat(e^i_j, u_i) = 5.
  \end{equation}
  Similarly, since \(\util_M(u_{i'}) < \sat(u_{i'}, e^{i'}_j)\) and \(M\) is a \csm, we have
  \begin{equation}
    \label{eq:4}
    \util_M(e^{i'}_j) \geq \sat(e^{i'}_j, u_{i'}) = 7.
  \end{equation}
  Consider the two edges incident with~\(f_j\) in $G_j$.
  In order for Equation~\eqref{eq:3} to hold,
  we must have \(7(1 - M(e^i_j, f_j)) + 3M(e^i_j, f_j) \geq 5\),
  i.e., \(M(e^i_j, f_j) \leq 0.5\).
  In order for Equation~\eqref{eq:4} to hold,
  we must have \(9(1 - M(e^{i'}_j, f_j)) + 4M(e^{i'}_j, f_j) \geq 7\),
  i.e., \(M(e^{i'}_j, f_j) \leq 0.4\).
  Combined, this implies the desired upper bound for the welfare~$\omega_j$, as follows:
  \begin{align*}
    \omega_j = & ~~~ \Big( \sum_{\substack{(y,x)\in \colW{W_j}\times \colU{U_j}\colon\\\{y,x\}\in E(G_j)}} M(\{y,x\}) \cdot (\sat(x,y)+\sat(y,x))\Big)  \\
 \nonumber   &  + \sat(e_j^{i}, u_i)\cdot M(e_j^{i},u_i)+\sat(e_j^{i'}, u_{i'})\cdot M(e_j^{i'},u_{i'})  \\
    \nonumber    \stackrel{\eqref{eq:edge-welfare}}{\le} & ~~       (2m^2 + 9) \cdot \sum_{\mathclap{\substack{(y,x)\in \colW{W_j}\times \colU{U_j}\colon\\\{y,x\}\in E(G_j)}}
    } M(\{y,x\}) + 12 \\
    \nonumber
    = & ~~ (2m^2 + 9) \cdot \big(\sum_{\mathclap{(x,y)\in \{e_j,g_j\}\times \colW{W_j}}}M(\{x,y\})\big)\\
             & + (2m^2+9)\cdot (M(\{f_j,e_j^i\}) + M(\{f_j,e^{i'}_j\}) + 12\\
    \le & ~~ 2.9(2m^2+9)+12 \\
    < & ~~ 3(2m^2+9)-n,
  \end{align*}
  where the first inequality holds since  $\sat(e_j^i,u_i)\le 5$ and  $\sat(e_j^{i'},u_{i'})\le 7$, 
  the second last inequality holds since the sum of values of the matching for each agent is at most one, and
  the last inequality holds since we assume without loss of generality that $0.2\cdot m^2-11.1 > n$.
\end{proof}
\fi

\ifshort
\stepcounter{theorem}
\stepcounter{theorem}
\else

The following lemma summarizes the desired property regarding the fully matched agents.
\begin{lemma}\label{cl:edge-gadget-toVC}
  Let $M$ be a fractional matching of $G_E$ (see \cref{cons:edge-gadget}). %
  For each edge $e_j \in E(G)$ with $e_j = \{v_i, v_{i'}\}$ it holds that
  if agent~$f_j$ is fully matched, then %
  we have $\util_{M}(e_j^i) < \sat(e_j^i,u_i)$ or $\util_{M}(e_j^{i'}) <\sat(e_j^{i'},u_{i'})$.
\end{lemma}
\begin{proof}
  Let fractional matching~$M$, edge~$e_j=\{v_i,v_{i'}\}\in E(G)$, and agent~$f_j$ be as defined.
  Since $f_j$ is fully matched, $M(f_j,e^{i})+M(f_j,e^{i'})=1$.
  We distinguish between two cases~(see \cref{fig:csm_edge_gadget} for the cardinal preferences):
  \begin{compactitem}[--]
    \item If $M(f_j,e^{i'}_j) > 2/5$, then
    $\util_M(e_j^{i'}) \le \sat(e_j^{i'},e_j)\cdot (1-M(f_j,e^{i'}_j)) + \sat(e_j^{i'},f_j) \cdot M(e^{i'}_j, f_j) = 9-5\cdot M(f_j, e^{i'}_j) < 7$, implying that  $\util_{M}(e_j^{i'}) <\sat(e_j^{i'},u_{i'})$.
    \item If $M(f_j,e^{i'}_j) \le 2/5$, meaning that $M(f_j,e^i_j) \ge 3/5 > 1/2$, then
    $\util_M(e_j^{i}) \le \sat(e_j^{i},e_j)\cdot (1-M(f_j,e^{i}_j)) + \sat(e_j^{i},f_j) \cdot M(e^{i}_j, f_j) = 7-4\cdot M(f_j, e^{i}_j) < 5$, implying that $\util_{M}(e_j^{i}) <\sat(e_j^{i},u_{i})$.\qedhere
  \end{compactitem}
\end{proof}

Finally, to ease the hardness proofs for both optimality criteria, we show a technical result regarding some specific fractional matchings, which is straight-forward to verify.
\begin{lemma}\label{lem:specific-matching}
  Let $M$ be a fractional matching for edge gadget~$G_E$.
  Consider an edge~$e_j=\{v_i,v_{i'}\}$ of $G$ with $i< i'$.
  \begin{compactenum}[(1)]
    \item\label{vi-not-is} If matching~$M$ fulfills the following:
  \begin{compactitem}[--]
    \item $M(e_j,h_j)\!=\!0.3$, $M(e_j,e_j^{i})=0.1$, $M(e_j,e_j^{i'})=0.6$,
    \item $M(f_j,e_j^i)=0.8$, $M(f_j,e_j^{i'})=0.2$, and
    \item $M(g_j,e_j^{i})\!=\!0.1$, $M(g_j,h_j)=0.7$, $M(g_j,e^{i'}_j)=0.2$,
  \end{compactitem}
  then no \cblocking pair of $M$ involves an agent from $\{e_j,f_j,g_j,e_j^{{\color{red}i'}},h_j\}$.
  \item\label{vi-is} If matching~$M$ satisfies $M(e_j,e^i_j)=M(f_j,e_j^{i'})=M(g_j,h_j)=1$~(see the green edges in \cref{fig:csm_edge_gadget}), 
  then no \cblocking pair of $M$ involves an agent from $\{e_j,f_j,g_j,e_j^{{\color{blue}i}},h_j\}$.
\end{compactenum}
\end{lemma}

\begin{proof}
  Let $M$ and $e_j=\{v_i,v_{i'}\}$ be as defined.
  For Statement~\eqref{vi-not-is}, assume that $M$ have the values stated in the if-condition.
  We first compute the utilities of the agents from~$\{e_j,f_j,g_j,e_j^{{\color{red}i'}},h_j\}$:
  \[
    \begin{array}{@{}r@{}l@{}r@{}l@{}r@{}l@{}}
    \util_M(e_j) &= 2m^2 + 2.6,~ & \util_M(f_j) & = 2m^2+5.8, ~& \\
     \util_M(g_j) & =2m^2+5, & \util_M(e^{i'}_j) & = 7.2, ~& \util_M(h_j) & = 3.1.
    \end{array}
  \]
  It is straight-forward to verify that no \cblocking pair of $M$ involves agents from $\{e_j,g_j,f_j,e_j^{i'}\}$.
  Hence, no \cblocking pair of $M$ involves agent~$h_j$.
  This completes the proof for Statement~\eqref{vi-not-is}.

  For Statement~\eqref{vi-is}, assume that $M$ have the values stated in the if-condition.
  We compute the utilities of the agents from $\{e_j,f_j,g_j,e_j^{{\color{blue}i}},h_j\}$:
     \[
    \begin{array}{@{}r@{}l@{}r@{}l@{}r@{}l@{}}
    \util_M(e_j) &= 2m^2 + 2,~ & \util_M(f_j) & = 2m^2+5, ~& \\
     \util_M(g_j) & =2m^2+5, & \util_M(e^{i}_j) & = 7, ~& \util_M(h_j) & = 4.
    \end{array}
  \]
  It is also straight-forward to verify that no \cblocking pair of $M$ involves agent~$e_j^{i}$ or $h_j$ as they both are integrally matched with their most preferred agents, respectively.
  Hence, no \cblocking pair of~$M$ involves an agent from~$\{e_j,f_j,g_j\}$.
  This completes the proof for Statement~\eqref{vi-is}.
\end{proof}
\fi

We are now ready to prove hardness for finding stable matchings with maximum social welfare.

\begin{theorem}[\appsymb]\label{th:ECSFM_hard}
  \ECSFM is \NPC, even for bipartite graphs with strict preferences.
\end{theorem}
\begin{proof}
  \iflong
  Let $G = (V, E)$ be a graph with $V =\{v_1, v_2, \ldots, v_n\}$ and $E=\{e_1, e_2, \ldots, e_m\}$, and let $(G, k)$ denote an instance of \IS.
  \else
  Let \((G, k)\) be an instance of \ISs, where $G=(V,E)$, $V=\{v_1,\ldots, v_n\}$ and $E=\{e_1,\ldots, e_m\}$. 
  \fi
  We will indeed reduce from \ISs in cubic graphs~\cite{AliKann2000apx-cubic}, i.e., each vertex in $G$ has degree three.
\ifshort  We construct an instance $(G',\sat, \gamma)$ of \ECSFM where $\gamma = (3n+7)n  + k +3m(2m^2+9)$ and $G'= (U \cup W, E')$ is a bipartite graph with partite sets $U$ and $W$.

  \begin{compactitem}
  \item We first introduce the vertices in the \myemph{vertex gadgets}, that is, for each $i \in [n]$, we add two vertices $u_i$ and $x_i$ to $U$ and we add two vertices $w_i$ and $y_i$ to $W$.
  \item We add $U_E$ to $U$ (if not already present) and $W_E$ to $W$, and we add the edges of $G_E$ to $E'$.
  \item The remaining edges in $E'$ and the $\sat$ function are for each $i \in [n]$ as defined in Figure~\ref{fig:Welfare-csm-card}.
\end{compactitem}
\looseness=-1 For $i \in [n]$ we also call the subgraph of~$G'$ induced by $\{u_i, w_i, x_i, y_i\}$ the \emph{vertex gadget for vertex~$v_i$}.
Note that $G'$ is a bipartite graph because all the introduced edges are between $U$ and~$W$.
This completes the construction of $(G',\sat, \gamma)$ which clearly takes polynomial time.
\else
We construct an instance~$(G',\sat, \gamma)$ of \ECSFM where $\gamma = (3n+7)n  + k +3m(2m^2+9)$ and $G'= (U \cup W, E')$ is a bipartite graph with partite sets~$U$ and~$W$ such that
$U=\colU{U_E}\cup \{u_i,x_i\mid v_i\in V\}$ and $W=\colW{W_E}\cup \{w_i,y_i\mid v_i\in V\}$.
Note that $\colU{U_E} = \{e_j,f_j,g_j\mid e_j\in E\}$ and $\colW{W_E} = \{e_j^{i}, e_j^{i'}, h_j \mid e_j=\{v_i,v_{i'}\}\in E\}$.
  In other words, $U$ (resp.\ $W$) include the agents from~$U_E$ (resp.\ $W_E$) of the edge gadgets (see \cref{cons:edge-gadget}) and four \myemph{vertex agents}~$u_i,x_i$ (resp.\ $w_i,y_i$) for each vertex~$v_i\in V$.
  In total, we have $|U|=|W|=6m+4n$.
  
  Similarly to the proof of \cref{th:Welfare-OSR_hard}, we will construct the cardinal preferences for the vertex agents to ensure the following.
  When combined with \cref{cl:edge-gadget-max-sat},
  there are essentially two possible ways of fractionally matching the vertex agents~$u_i,x_i,w_i,y_i$ corresponding to the same vertex~$v_i\in V$. %
  The first one will have a higher welfare than the second one,
  but it is not possible to use the first one for two adjacent agents as this will induce a \cblocking pair. 
  
  \cref{fig:Welfare-csm-card} illustrates the cardinal preferences of the vertex agents for each vertex~$v_i\in V$.
  The cardinal preferences that are not depicted in the figure are set to zero.
  We call the subgraph induced by the vertex agents~$u_i,x_i,w_i,y_i$ a \myemph{vertex gadget for~$v_i$}, and use $H_i$ to denote it.
Note that $G'$ is a bipartite graph since all the introduced edges are between $U$ and~$W$.
This completes the construction of $(G',\sat, \gamma)$ which clearly takes polynomial time.
\fi 
 \begin{figure}[t!]
   \begin{minipage}{.63\linewidth}
      \begin{tikzpicture}
      \def \xscale {2.3}
      \def \yscale {1.2}
      \def \ss {0.2}
      \def \ee {-.8}
      \foreach \col / \x / \y / \n / \c / \p in {
        {blue!70!black}/0/0/{~~u_i}/ui/above,
        {red!70!black}/1/0/{w_i}/wi/above,
        {blue!70!black}/0/-.6/{x_i}/xi/below,
        {red!70!black}/1/-.6/{y_i}/yi/below,
        {black}/-.5/\ss/{}/Eui1/above,
        black/-.5/\ee/{}/Eui2/above} {
        \node[agent] at (\x*\xscale, \y*\yscale) (\c) {};
        \node[\p = 0pt of \c,color=\col] {$\n$};
      }

      \node[agent] at ($(Eui1)!.5!(Eui2)$) (Eui3) {};

       \draw[decoration={brace,mirror, raise=5pt,amplitude=5pt},decorate] (-.5*\xscale, \ss*\yscale) -- node[left=2ex] {$[E(v_i)]$} (-.5*\xscale, \ee*\yscale);

       \foreach \c / \a / \s /  \t / \n / \p / \m / \q in {%
         black/15/ui/wi/{3n\!\!+\!\!3}/0.25/1/{0.8},
           red!80!black/2/ui/yi/1/0.2/{3n\!\!+\!\!4}/0.8,
           black/{-15}/xi/yi/2/0.15/1/0.85,
           red!80!black/0/xi/wi/1/0.2/2/0.8,
           black/-14.5/ui/Eui1/{3n}/0.4/{}/.9,
           black/0.5/ui/Eui3/{3n\!\!+\!\!1}/0.4/{}/.9,
           black/15/ui/Eui2/{3n\!\!+\!\!2}/0.35/{}/.9,
           black/2/ui/yi/1/0.2/{3n\!\!+\!\!4}/0.8,
           black/0/xi/wi/1/0.2/2/0.8
         }
         {
           \draw[] (\s) edge[bend left=\a,color=\c] node[pos=\p,fill=white, inner sep=.4pt] {\scriptsize $\n$} node[pos=\q,fill=white,inner sep=.4pt] {\scriptsize $\m$} (\t);
         }

         \begin{pgfonlayer}{background}

         \foreach \c / \a / \s /  \t / \n / \p / \m / \q in
         {red!60!black/2/ui/yi/1/0.2/{3n\!+\!4\!}/0.8,
           black/0/xi/wi/1/0.2/2/0.8%
         }
         {
           \draw[linemarkr] (\s) edge[bend left=\a]  (\t);
         }

       \end{pgfonlayer}
     \end{tikzpicture}
     \end{minipage}
     \begin{minipage}{.36\linewidth}
       \[%
      \begin{array}{@{}r@{}l@{}}
        \textcolor{blue!70!black}{u_i}\colon & w_i \pref [E(v_i)] \pref y_i,\\
        \textcolor{blue!70!black}{x_i}\colon & y_i \pref w_i, \\
         \textcolor{red!70!black}{w_i}\colon & x_i \pref u_i,\\
         \textcolor{red!70!black}{y_i}\colon & u_i \pref x_i.
      \end{array}
    \]
    \end{minipage}
    \caption{The vertex gadget for a vertex $v_i \in V$.
      Left: The cardinal preferences of agents~$u_i,x_i,w_i,y_i \in V(G')$ which correspond to vertex~$v_i \in V(G)$, constructed in the proof of Theorem~\ref{th:ECSFM_hard}.
      Here, $E(v_i)= \{e_j^i \mid e_j \in E$ and $v_i$ is incident to~$e_j\}$ and $[E(v_i)]$ denotes the sequence resulting from ordering $E(v_i)$ in increasing order of the indices of the edge agents in~$E(v_i)$.
      The red edges indicate the matching that signifies that $v_i$ is in the independent set.
      Right: The induced preference lists of the vertex gadget.}
    \label{fig:Welfare-csm-card}
  \end{figure}
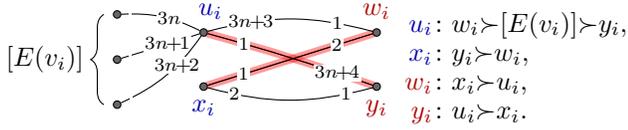

  Next, we prove
  \iflong the correctness, that is,
  \fi $G$ has an independent set of size at least~$k$ if and only if $(G', \sat)$ has a \cstable\ matching with \wel\ at least $\welfare$.
\ifshort
The ``if'' direction amounts to a straightforward check of a somewhat peculiar matching and is available in the appendix.
\else

For the ``if'' direction, given a \(k\)\nobreakdash-vertex independent set $V'$ of $G$, construct a fractional matching $M$ as follows:
\begin{itemize}[--]
\item For each $z\in [n]$, if $v_{z} \in V'$, then set $M(u_z,y_z) = M(x_z,w_z) = 1$; else set $M(u_z,w_z) = M(x_z,y_z)=1$.
\item For each edge $e_j = \{v_i,v_{i'}\}\in E$ with $i<i'$, do:
  \begin{compactenum}[(1)]
  \item If $v_i \notin V'$, then set
  \[  \begin{array}{@{}l@{\,}l@{\,}l@{}}
     M(e_j,h_j)=0.3,& M(e_j,e_j^{i})=0.1, & M(e_j,e_j^{i'})=0.6,\\
     M(f_j,e_j^i)=0.8,& M(f_j,e_j^{i'})=0.2 & \\
     M(g_j,e_j^{i})=0.1,& M(g_j,h_j)=0.7, & M(g_j,e^{i'}_j)=0.2,
  \end{array}\]
  Note that, $M$ satisfies the if-condition for $e_j$ stated in \cref{lem:specific-matching}\eqref{vi-not-is}.
  \item If $v_{i} \in V'$, then set \[M(e_j, e_j^i) = M(f_j, e_j^{i'}) = M(g_j, h_j) = 1;\]
  note that this corresponds to the if-condition stated in \cref{lem:specific-matching}\eqref{vi-is}.
  \end{compactenum}%
\item For every pair $e$ of agents in $G'$ not assigned above, set $M(e) = 0$.
\end{itemize}
First, observe that $M$ is a fractional matching, that is, for each vertex, the values assigned by $M$ to edges incident with that vertex sum to at most one.
Indeed, this matching is a perfect matching, i.e., every agent is fully matched.
A simple calculation shows that
\begin{align*}
  \wel(M)=\sum_{i=1}^n \sum_{\beta \in \{u_i, w_i, x_i, y_i\}} \util_{M}(\beta) = (3n+7)n + k
\end{align*}
and that the total \wel\ received from \(M\) by the agents in the edge gadgets is~$3m(2m^2+9)$, giving the overall required welfare of~\(\gamma\).

It remains to prove that $M$ is \cstable.

\smallskip
\noindent \textit{The edge gadget.} We begin by showing that there is no agent of any edge gadget is involved in a \cblocking pair.
Consider the edge gadget of an edge $e_j=\{v_i, v_{i'}\} \in E(G)$ with $i< i'$.
We distinguish between two cases.
\begin{compactenum}[(1)]
  \item If $v_i\notin V'$, then $M$ fulfills the if-condition given in \cref{lem:specific-matching}\eqref{vi-not-is}.
  Hence, no \cblocking pair involves an agent from $\{e_j,f_j,g_j,e_j^{{\color{red}i'}}, h_j\}$.
  To show that neither does a \cblocking pair involve agent~$e_j^{{\color{blue}i}}$, it remains to consider pair~$\{e_j^{{\color{blue}i}}, u_i\}$.
  Since $v_i\notin V'$, by definition, we have that $M(u_i,w_i)=1$, meaning that $u_i$ is integrally matched with her most preferred agent.
  Therefore, $\{e_j^{{\color{blue}i}}, u_i\}$ is \emph{not} \cblocking~$M$.

  \item If $v_i\in V'$, then $M$ fulfills the if-condition given in \cref{lem:specific-matching}\eqref{vi-is}.
  Hence, no \cblocking pair involves an agent from $\{e_j,f_j,g_j,e_j^{{\color{blue}i}}, h_j\}$.
  To show that neither does a \cblocking pair involve agent~$e_j^{{\color{red}i'}}$, it remains to consider pair~$\{e_j^{{\color{red}i'}}, u_{i'}\}$.
  Since $v_{i}\in V'$ and $V'$ is an independent set, we have that $v_{i'}\notin V'$.
  By definition, it must hold that $M(u_{i'},w_{i'})=1$, meaning that $u_{i'}$ is integrally matched with her most preferred agent.
  Therefore, $\{e_j^{{\color{red}i'}}, u_{i'}\}$ is \emph{not} \cblocking~$M$.
\end{compactenum}
In both cases, we have shown that no \cblocking pair involves an agent from the edge gadgets.

\smallskip
\noindent \textit{The vertex gadget.}
We now show that \cblocking pair of $M$ involves an agent from the vertex gadgets.
This will show that $M$ is \cstable.
Let $i \in [n]$ and consider the vertex gadget for vertex \(v_i \in V\).
Again, we distinguish between two cases:
\begin{compactenum}[(1)]
  \item If $v_i\in V'$, then no \cblocking pair of $M$ involves agents~$w_i$ or $y_i$ since they are integrally matched with their most preferred agents, respectively.
  Hence, no \cblocking pair involves~$x_i$.
  Since no \cblocking pair involves an agent from the edge gadget, neither is $y_i$ involved in a \cblocking pair.
  \item If $v_i \in V'$, then no \cblocking pair of $M$ involves agents~$u_i$ or $x_i$ since they are integrally matched with their most preferred agents, respectively.
  Hence, no \cblocking pair involves~$w_i$ or $x_i$.
\end{compactenum}
Hence, no \cblocking pair of~\(M\) involves an agent from the vertex gadget.
This concludes the proof of the ``if'' direction of the correctness.
\medskip

\fi
For the ``only if'' direction, let $M$ be a \cstable\ matching for $(G',\sat)$ with $\wel(M) \geq \gamma$.
We define a vertex subset \myemph{$V' = \{v_i \in V \mid M(u_i,y_i) \geq 1/n\}$}.
We show that $V'$ is an independent set in $G$ and $|V'| \geq k$.

We first show that \(|V'| \geq k\).
\iflong For brevity we \else We \fi introduce the following notation.
For each $i \in [n]$,
\ifshort
$ m_{i}^{uw}  \coloneqq M(u_i, w_i)$,  $m_i^{ue}   \coloneqq \sum_{e \in E(v_i)} M(u_i, e)$, 
  $m_{i}^{uy}  \coloneqq M(u_i y_i)$,  $m_{i}^{wx}  \coloneqq M(w_i,x_i)\text{, and } 
  m_{i}^{xy} \coloneqq M(x_i,y_i)$.
\else
define \begin{alignat*}{3}
  m_{i}^{uw} & \coloneqq M(u_i, w_i),~~ & m_i^{ue}  & \coloneqq \sum_{\mathclap{e \in E(v_i)}} M(u_i, e),~~ & m_{i}^{uy} & \coloneqq M(u_i y_i), \\
   m_{i}^{wx} & \coloneqq M(w_i,x_i), &
  m_{i}^{xy} & \coloneqq M(x_i,y_i),
       \end{alignat*}
       where $E(v_i)=\{e_j^i\mid v_i \in e_j \text{ for some } e_j \in E(G)\}$.
\fi
\iflong

Then we have,
\begin{align}
  & \sum_{i=1}^n \sum_{z \in \{u_i, w_i, x_i, y_i\}} \util_{M}(z) \nonumber \\
  & = \sum_{i=1}^n (3n+4)m_{i}^{uw} + (3n+2)m_{i}^{ue} \nonumber \\
  & \qquad  + (3n+5)m_{i}^{uy} + 3m_{i}^{wx} + 3m_{i}^{xy} \nonumber \\
  & \leq \sum_{i=1}^n \Big((3n+4)(m_{i}^{uw} + m_{i}^{ue}+ m_{i}^{uy}) \nonumber \\
  & \qquad + 3(m_{i}^{wx} + m_{i}^{xy})\Big) + \sum_{i=1}^nm_{i}^{uy}.\label{eq:5}
\end{align}
\else%
Then,

  $\sum_{i=1}^n \sum_{\substack{z \in \{u_i, w_i, x_i, y_i\}}} \util_{M}(z) \nonumber = \sum_{i=1}^n (3n+4)m_{i}^{uw} \nonumber \\
 \quad~~~~~~~~~~ + (3n+2)m_{i}^{ue} 
   + (3n+5)m_{i}^{uy} + 3m_{i}^{wx} + 3m_{i}^{xy} \nonumber \\
   \leq \sum_{i=1}^n \Big((3n+4)(m_{i}^{uw} + m_{i}^{ue}+ m_{i}^{uy}) + 3(m_{i}^{wx} + m_{i}^{xy})\Big)  \nonumber \\
   \quad~~~~~~~ + \sum_{i=1}^nm_{i}^{uy}.\hfill~\refstepcounter{equation}(\theequation){\label{eq:5}}$
\fi

\ifshort
Since the total matching values of the edges incident to $u_i$ and $w_i$ is at most~$1$, we get that the right-hand side of \cref{eq:5} is at most
$(3n+4)n + 3n + \sum_{i=1}^nm_{i}^{yu}. \hfill~\refstepcounter{equation}(\theequation){\label{eq:6}}$
\else
Since $M$ is a fractional matching, meaning that the sum of the matching values for $u_i$ (resp.\ $w_i$) is at most one, the right-hand side of Inequality~\eqref{eq:5} is upper-bounded by 
\begin{equation}
(3n+4)n + 3n + \sum_{i=1}^nm_{i}^{yu}=(3n+7)n + \sum_{i=1}^nm_{i}^{yu}.\label{eq:6}
\end{equation}\fi

\ifshort
\noindent We show that $\lfloor \sum_{i=1}^n m_{i}^{uy} \rfloor \leq |V'|$.
Define
$k_1 \coloneqq |\{u_i \mid i \in [n]\text{ and } m_{i}^{uy} <
  1/n\}|$.
\else 
\noindent To show that $|V'|\ge k$, we first show that $|V'|\ge \lfloor \sum_{i=1}^n m_{i}^{uy} \rfloor$.
To this end, define \[k_1 \coloneqq |\{u_i \mid i \in [n] \text{ and } m_{i}^{uy} <
  1/n\}|.\]\fi
Then we have \(\sum_{i=1}^nm_{i}^{yu} < k_1/n + |V'|\).
Since \(k_1/n \leq 1\) thus \(|V'| \geq \lfloor \sum_{i=1}^nm_{i}^{yu} \rfloor  \).
To see that \(|V'| \geq k\) it thus suffices to show \(\sum_{i=1}^nm_{i}^{yu} \ge k\).
For this, out of \iflong the~\(\wel(M)\) \else \(\wel(M)\) \fi at most \(3m(2m^2 + 9)\) stems from the agents of the edge gadgets (see \cref{cl:edge-gadget-max-sat}).
Hence, at least \((3n +4)n + 3n + k\) must stem from the agents of the vertex gadgets. 
By the upper bound on the welfare of these agents derived in \cref{eq:6}, thus indeed \(\sum_{i=1}^nm_{i}^{yu} \geq k\), as required.

To conclude the proof, we show that \(V'\) is an independent set in~\(G\).
Towards a contradiction, suppose that there is an edge \(e_j \in E\) with \(e_j = \{v_i, v_{i'}\}\) and \(i' > i\) such that \(e_j \subseteq V'\).
\iflong By the definition of \(V'\), for each \(\nu \in \{i, i'\}\) 
the utility of $u_\nu$ has \(\util_M(u_\nu) \leq 1/n + (n - 1) \cdot (3n + 3) / n < 3n\).
That is, \(\util_M(u_\nu) < \sat(u_\nu, e^\nu_j)\).
By \cref{cl:edge-gadget-max-sat}, the total welfare received from \(M\) by the agents of the edge gadgets is at most \(3m(2m^2 + 9) - n\).
By Inequality~\eqref{eq:6}, the total welfare received from the agents in the vertex gadgets is at most \((3n+4)n + 3n + n\), meaning that \(\wel(M) < \wel\), a contradiction.
\else By definition of \(V'\) we have for both \(\nu \in \{i, i'\}\) that \(\util_M(u_\nu) \leq 1/n + (n - 1) \cdot (3n + 3) / n < 3n\).
That is, \(\util_M(u_\nu) < \sat(u_\nu, e^\nu_j)\).
By \cref{cl:edge-gadget-max-sat}, the total welfare received from \(M\) by vertices of the edge gadgets is at most \(3m(2m^2 + 9) - n\).
By \cref{eq:6}, the total welfare received from the vertices %
in the vertex gadgets is at most \((3n+4)n + 3n + n\), meaning that \(\wel(M) < \gamma\), a contradiction.
\fi
Thus, \(V'\) is indeed an independent set in~\(G\).
\end{proof}

\noindent\textbf{Remark.}
Note that even if we require the preference lists to be complete, Theorem~\ref{th:ECSFM_hard} still holds:
We can add a sufficiently large value to $\sat(u,v)$ and $\sat(v,u)$ for each acceptable pair~$\{u,v\}$ in the constructed instance, and assign small but distinct positive values to each un-mentioned pair in the instance.
Summarizing, Theorem~\ref{th:ECSFM_hard} does not rely on edges with zero values and holds even if we have complete preferences without ties.

\iflong

\par

Next, we prove that \PCSFM is \NPC even when the graph $G$ is bipartite and the values of $\sat(v,\cdot)$ are distinct, for each vertex $v$ in $G$.
\else
Next, we turn to \PCSFM.
As mentioned, the hardness construction uses the same edge gadgets as before, albeit with a different analysis.
\fi
\begin{theorem}[\appsymb]\label{th:PCSFM_hard}
  \PCSFM is \NPC,  even for bipartite graphs with strict preferences.%
\end{theorem}
\iflong
\begin{proof}
  Let \(G = (V, E)\) be a graph with \(V=\{v_1, v_2, \ldots, v_n\}\) and \(E=\{e_1, e_2, \ldots, e_m\}\) and let \((G, k)\) be an instance of \ISs.
  We construct an instance $(G',\sat)$ of \PCSFM where $G'= (U \cup W, E')$ is a bipartite graph with partite sets \(U\) and \(W\).
  The basic idea is similar to \cref{th:ECSFM_hard}: We have a vertex gadget for each vertex~\(v\) of~\(G\) that has two possibilities of being matched, signifying whether \(v\) is supposed to be in the independent set.
  If \(v\) is in the independent set, then an agent in \(v\)'s vertex gadget will be unsatisfied with respect to all agents in the edge gadgets that correspond to the edges incident to~\(v\).
  In order to ensure that at least \(k\) vertices are selected in the independent set, we use a selector gadget.
  In this gadget there are \(k\) ``selector'' agents which, in order to be fully matched, need to be matched to agents in the vertex gadget.
  If an agent~\(u_i\) in a vertex gadget is matched with large-enough value to a selector agent, this makes him unsatisfied with respect to the agents in his edge gadgets, thus signifying that the vertex~\(v_i \in V\) corresponding to \(u_i\) is selected into the independent set.
  The properties of the edge gadget then ensure that the selected vertices indeed form an independent set.

  We construct $G' = (U \cup W, E')$ and \(\sat\) as follows.
  For each $i \in [n]$ we let $d_i$ denote the degree of the vertex $v_i$ in~\(G\).
  \begin{itemize}
  \item For each \(i \in [n]\) we introduce a \myemph{vertex gadget} for vertex \(v_i \in V\) which contains an agent $u_i$ in $U$ and an agent $w_i$ in $W$.
  \item We introduce a \myemph{selector gadget} which contains the following agents:
    \begin{itemize}
    \item We add $2k$ agents, called~$t_1, t_2, \ldots, t_k$ and $c_1, c_2, \ldots, c_k$, to~$U$.
    \item We add $2k$ agents, called $s_1, s_2, \ldots, s_k$ and $a_1, a_2, \ldots, a_k$,
    to~$W$.
    \end{itemize}
  \item Let \(G_E\) be the graph with vertex set \((\colU{U_E} \cup \colW{W_E})\) that
  contains the edge gadgets as defined in \cref{cons:edge-gadget}.
  That is, $\colU{U_E} = \{e_j,f_j,g_j\mid e_j \in E\}$ and $\colW{W_E} = \{e_j^i, e_j^{i'}, h_j \mid e_j\in E\}$.
  We add each vertex in $\colU{U_E}$ to $U$ (if not already present) and $\colW{W_E}$ to $W$, and we add the edges of $G_E$ to $E'$.
  The \(\sat\) function on these agents is defined in \cref{cons:edge-gadget} and is as shown in \cref{fig:csm_edge_gadget}.
  \item For each $i\in [n]$ and $j\in [k]$ we define the acceptability relations of the agents from the vertex gadget and the selector and their satisfaction~$\sat$ as in \cref{fig:Perefect-csm-card}.
\end{itemize}
  In total, we have introduced $6m+2n+4k$~agents.
  To completes the construction, we define~$\fullnum=|U|+|W|=6m+2n+4k$,
  and obtain the instance~$(G',\sat, \fullnum)$ of \PCSFM.
  Clearly, it can be carried out in polynomial time.
  Moreover, \(G'\) is bipartite since all the introduced edges are between \(U\) and \(W\).
  Since $\fullnum = |U|+|W|$, searching for an \csm with $\fullymatched$ fully matched agents means searching for a perfect \csm.
  Next, we prove the correctness.
  We will show that $G$ has an independent set of size at least~$k$ if and only if $(G', \sat)$ has a perfect \csm. %

 \begin{figure}
    \tikzstyle{ellipsis} = [circle, inner sep=0.5pt, draw, fill=black!60]
   \begin{tikzpicture}
      \def \xscale {2.3}
      \def \yscale {.8}
      \def \ss {2}
      \def \ee {0}
      \foreach \x / \y / \n / \c / \p in {0/1.5/{u_i}/ui/above,1.2/1.5/{s_j}/sj/above, 0/0/{c_j}/cj/left, 1.2/0/{w_i}/wi/right, 0/-1/{t_j}/tj/below, 1.2/-1/{a_j}/aj/below,  -1.4/2/{}/Eui1/above,-1.4/0/{}/Eui2/above} {
        \node[agent] at (\x*\xscale, \y*\yscale) (\c) {};
        \node[\p = 0pt of \c] {$\n$};
      }

      \foreach \sstep in {\ee, 1.0 0,0.8 0, 1.2 0, \ss}{
        \node[ellipsis] at (-1.4*\xscale, \sstep*\yscale) {};
      }

       \draw[decoration={brace,mirror, raise=5pt,amplitude=5pt},decorate] (-1.5*\xscale, \ss*\yscale) -- node[left=2ex] {$[E(v_i)]$} (-1.5*\xscale, \ee*\yscale);

       \foreach \c / \a / \s /  \t / \n / \p / \m / \q in {%
         black/{10}/ui/wi/{2n^2\!+\!d_i}/0.35/k/{0.85},
         red!80!black/0/ui/sj/j/0.2/{i\!-\!1\!}/0.8,
         red!80!black/{-15}/tj/aj/2n/0.15/0/0.85,
         red!80!black!80!black/{-5}/tj/wi/{i\!-\!1\!}/0.2/2k\!+\!{j\!-\!1\!}/0.8,
          red!80!black/{0}/cj/aj/0/0.15/1/0.85,
         red!80!black!80!black/{10}/cj/wi/{i}/0.2/{j\!-\!1\!}/0.7,
         black/0.5/ui/Eui1/{2n^2\!+\!d_i\!-1\!}/0.3/{}/0,
         black/0/ui/Eui2/{2n^2}/0.2/{}/0
         }
      {
        \draw[] (\s) edge[bend left=\a] node[pos=\p,fill=white, inner sep=.5pt] {\scriptsize $\n$} node[pos=\q,fill=white,inner sep=.5pt] {\scriptsize $\m$} (\t);
      }

       \begin{pgfonlayer}{background}

         \foreach \c / \a / \s /  \t / \n / \p / \m / \q in
         {  red!80!black/0/ui/sj/j/0.2/{i\!-\!1\!}/0.8,  red!80!black!80!black/{10}/cj/wi/{i\!-\!1\!}/0.2/{j\!-\!1\!}/0.7,  red!80!black/{0}/cj/aj/0/0.15/1/0.85,  red!80!black/{-15}/tj/aj/2n/0.15/0/0.85,
         red!80!black!80!black/{-5}/tj/wi/{i\!-\!1\!}/0.2/2k\!+\!{j\!-\!1\!}/0.8%
         }
         {
           \draw[linemarkr] (\s) edge[bend left=\a]  (\t);
         }

       \end{pgfonlayer}
    \end{tikzpicture}
    \[
      \begin{array}{@{}r@{}l@{\;}r@{}l@{}}
        \\
        u_i\colon\!&w_i \pref [E(v_i)] \pref s_k\pref \cdots \pref s_1,& a_j\colon\!&c_j \pref t_j,\\
        t_j\colon\!&a_j \pref w_n \pref \cdots \pref w_1,& s_j\colon&u_n \pref \cdots \pref u_1,\\
        c_j\colon\!&w_n \pref \cdots \pref w_1 \pref a_j,&w_i\colon\!&t_k\pref\cdots \pref t_1 \pref u_i \pref c_k \pref \cdots\pref c_1.\\
      \end{array}
    \]
    \caption{The vertex and selector gadget used in the proof of \cref{th:PCSFM_hard}. Top: The cardinal preferences of the two agents~$u_i$ and $w_i$ which correspond to vertex~$v_i$ and four agents $s_j,t_j,a_j,c_j$, $j\in [k]$. %
      Here, $E(v_i)= \{e_j^i \mid e_j \in E$ and $v_i$ is incident to~$e_j\}$ and $[E(v_i)]$ denotes the sequence resulting from ordering $E(v_i)$ in increasing order of
      the indices of the edge agents in~$E(v_i)$.
      The red edges indicated in the matching signifies that $v_i$ is in the independent set.
      Bottom: The induced preference lists.}
    \label{fig:Perefect-csm-card}
  \end{figure}

  For the forward direction, given an independent set of $G$ of size at least~\(k\), we construct a fractional matching~$M$ as follows.
  Let \(V'\) be a subset of the independent set of size exactly~\(k\).
  Denote the vertices of $V'$ as $v_{\ell_1}, \ldots, v_{\ell_k}$, where $\ell_1 < \ell_2 < \cdots < \ell_k$ %
  (this ordering will be crucial to make sure that the agents \(s_j\) will not be involved in a \cblocking pair).
  \begin{compactitem}
    \item For each $j \in [k]$, we set $M(u_{\ell_j},s_j) =1$ and set
    \begin{equation*}
      M(w_{\ell_j},c_j) = M(c_j,a_j) = M(a_j,t_j) = M(t_j,w_{\ell_j}) = 1/2.
    \end{equation*}
    This matching is indicated by the red lines in \cref{fig:Welfare-csm-card}.
  \item For each $v_i \in V \setminus V'$, we set $M(u_i,w_i) = 1$.
  \item For each pair of agents in the edge gadget the values set by $M$ are the same as in Theorem~\ref{th:ECSFM_hard}:  
    \item For each edge $e_j = \{v_i,v_{i'}\} \in E$ with $i<i'$, do:
    \begin{compactenum}[(1)]
  \item If $v_i \notin V'$, then set
  \[  \begin{array}{@{}l@{\,}l@{\,}l@{}}
     M(e_j,h_j)=0.3,& M(e_j,e_j^{i})=0.1, & M(e_j,e_j^{i'})=0.6,\\
     M(f_j,e_j^i)=0.8,& M(f_j,e_j^{i'})=0.2 & \\
     M(g_j,e_j^{i})=0.1,& M(g_j,h_j)=0.7, & M(g_j,e^{i'}_j)=0.2,
  \end{array}\]
  Note that, $M$ satisfies the if-condition for $e_j$ stated in \cref{lem:specific-matching}\eqref{vi-not-is}.
  \item If $v_{i} \in V'$, then set \[M(e_j, e_j^i) = M(f_j, e_j^{i'}) = M(g_j, h_j) = 1;\]
  note that this corresponds to the if-condition stated in \cref{lem:specific-matching}\eqref{vi-is}.
  \end{compactenum}%
  \item For every pair $e$ of agents in $G'$ not assigned above, set $M(e) = 0$. %
\end{compactitem}
  First, observe that $M$ is a matching.
  Indeed, \(M\) is perfect: The agents in the edge gadgets are fully matched by a direct calculation.
  The agents in the vertex and selector gadgets are also fully matched:
  For each $i \in [n]$, agent $u_i$ is matched integrally to either some \(s_j\) or  some \(w_j\).
  As to the agents \(w_i\), \(i \in [n]\), if $v_i \in V \setminus V'$, then agent $w_i$ is also matched integrally.
  Otherwise, if $v_i \in V'$, then $w_i$ is matched half-integrally with both $c_j$ and $t_j$, for some $j \in [k]$.
  Similarly, for each $j \in [k]$, the vertices $c_j$, $a_j$, and $t_j$ are matched half-integrally with two vertices.
  Hence, every vertex is fully matched.
  Therefore, $M$ is a perfect matching.
  
  It remains to prove that $M$ is \cstable.
  We consider the agents from the edge gadgets and the vertex and selector gadgets separately.

  \smallskip
  \noindent \textit{The edge gadget.}
  Consider an arbitrary edge \(e_j \in E\) with \(e_j = \{v_i, v_{i'}\}\) for \(i < i'\) and consider the edge gadget for~\(e_j\).
  Assume that \(v_i \notin V'\); the case \(v_i \in V'\) is analogous.
  By the definition of \(M\) in the edge gadgets in combination with \cref{lem:specific-matching}\eqref{vi-not-is}, we have that every \cblocking pair that involves an agent from the edge gadget for~\(e_j\) must involve both \(e^i_j\) and \(u_i\).
  However, since \(v_i \notin V'\) and by the definition of \(M\), we have \(\util_M(u_i) = \sat(u_i, w_i) > \sat(u_i, e^i_j)\).
  Thus, there is no \cblocking pair involving the agents of the edge gadgets.

  \smallskip
  \noindent \textit{The vertex and selector gadget.} It remains to show that
  no \cblocking pair involves an agent from a vertex gadget or a selector gadget in~$U$.
  Note that each pair of agents from the vertex or selector gadget involves an agent from the following set~\(V^*=\{u_i \mid i \in [n]\} \cup \{c_j, t_j \mid j \in [k]\}\). %
  It thus suffices to show that the agents in this set~$V^*$ are not involved in any \cblocking pairs.
  
  Consider an arbitrary agent \(u_i\) with $i\in [n]$.
  Since there is no \cblocking pair involving \(u_i\) and an agent of the edge gadgets,
  each \cblocking pair involving \(u_i\) must involve either \(w_i\) or \(s_j\) for some \(j \in [k]\).
  Note that if $M(u_i,w_i)=1$, then there is no \cblocking pair involving~$u_i$ because \(w_i\) is \(u_i\)'s most preferred agent.
  Otherwise, we have $i = \ell_p$ for some $p \in [k]$ and $M(u_i,s_p)=1$.
  Then $\{u_i,w_i\}$ is not \cblocking:
  This is because $M(w_i,t_p)=1/2$ and hence $\util_M(w_i)\ge 1/2\cdot \sat(w_i,t_p) = k = \sat(w_i, u_i)$. %
  Moreover, none of the agents \(s_j\) with $j < p$ can form a \cblocking pair with $u_i$ since $M(u_i,s_p)=1$ and hence $\util_M(u_i) = \sat(u_i, s_p) = p > j = \sat(u_i,s_j)$.
  Neither does any agent~$s_j$ with $j > p$ form a \cblocking pair with $u_i$ since $M(s_j,u_{\ell_j}) = 1$ and hence, $\util_{M}(s_j)=\sat(s_j, u_{\ell_j}) = \ell_j > \ell_p = \sat(s_j, u_{i})$; %
  Recall the ordering of \(V'\) and that \(i = \ell_p\).
  Hence, no blocking pair involves~$u_i$.

  Next, let an arbitrary agent~\(c_j\) with $j\in [k]$.
  Since for each $i' \in [n]$, we have $\util_M(w_{i'}) \geq k > \sat(w_{i'}, c_j)$, pair $\{c_j,w_{i'}\}$ is not \cblocking.
  Since $\util_M(c_j) \geq 1/2 > \sat(c_j, a_j)$, neither is pair~$\{c_j, a_j\}$ \cblocking.
  Hence, no blocking pair involves~$c_j$.
  
  Finally, consider an arbitrary agent \(t_j\) with $j\in [k]$.
  Since $\sat(a_j, t_j) = 0$, pair~$\{t_j,a_j\}$ is never a \cblocking pair. %
  Moreover, since for each $i' \in [n]$, $\sat(t_j,w_{i'}) \leq n-1$, $\{t_j,w_{i'}\}$ is not \cblocking.
  Hence, no blocking pair involves~$t_j$.
  Therefore, $M$ is \cstable.
  This proves the forward direction.

  For the other direction, let $M$ be a \cstable matching in the instance $(G',\sat)$ where every vertex is fully matched in $M$.
  We define the vertex set \[V' = \{v_i \in V \mid \sum_{j \in [k]} M(u_i,s_j) \geq 1/n\}.\]
  We show that $V'$ is an independent set in $G$ and $|V'| \geq k$.

  We first show that \(|V'| \geq k\).
  Call \(i \in [n]\) \myemph{good} if \(\sum_{j \in [k]} M(s_j, u_i) \geq 1/n\) and \myemph{bad} otherwise.
  Observe that \(|V'|\) is equal to the number of good indices in \([n]\).
  Since for each \(j \in [k]\) agent \(s_j\) is fully matched by \(M\), we have \(k = \sum_{i \in [n]}\sum_{j \in [k]} M(u_i, s_j)\).
  Thus,
  \begin{align*}
    k & = \sum_{\mathclap{{i \in [n]\colon i \text{ bad }}}}~~~~~~ \sum_{j \in [k]} M(s_j, u_i) + \sum_{\mathclap{i \in [n]\colon i \text{ good }}}~~~~~~ \sum_{j \in [k]} M(s_j, u_i)\label{eq:9}\\
      & < \frac{n - |V'|}{n} + |V'|.
  \end{align*}
  Thus, since \((n - |V'|)/n \leq 1\) and \(k\) is an integer, we have \(|V'| \geq k\), as required.

  It remains to show that \(V'\) is an independent set.
  Consider an arbitrary edge \(e_j \in E\) with \(e_j = \{v_i, v_{i'}\}\) where \(i < i'\).
  Towards a contradiction assume that for both \(\nu \in \{i, i'\}\) we have \(v_\nu \in V'\).
  Then, since \(\sum_{j \in [k]} M(u_\nu, s_j) \geq 1/n\) we have
  \begin{align*}
    \util_M(u_\nu) & \leq (2n^2 + d_\nu)\cdot\frac{n - 1}{n} + \frac{k}{n} \\
                 & = \frac{1}{n} (2n^3 + d_\nu \cdot n + k - 2n^2 - d_\nu) < 2n^2.
  \end{align*}
  Note that the last inequality holds since $k\le n$ and $0\le d_{\nu} < n$.
  Thus, \(\util_M(u_\nu) < \sat(u_\nu, e^\nu_j)\).
  However, by \cref{cl:edge-gadget-toVC} we have \(\util_M(e^i_j) < \sat(e^i_j, u_i)\) or \(\util_M(e^{i'}_j) < \sat(e^{i'}_j, u_{i'})\).
  Thus, \(\{e^i_j, u_i\}\) or \(\{e^{i'}_j, u_{i'}\}\) forms a blocking pair, a contradiction.
  Thus indeed, \(V'\) is an independent set.
  As shown above, \(|V'| \geq k\), as required.
\end{proof}
\fi

\ifshort
The specific construction above allows us also to derive a W[1]-hardness result with respect to a parameterization above a tight lower bound:
\fi
Since \ISs is \WOH wrt.\ the solution size~\cite{clique}, the reduction for Theorem~\ref{th:PCSFM_hard} implies that \PCSFM is \WOH wrt.\ parameter~``$\fullymatched(M)$$-$$\fullymatched(M^{\pi})$''%
\iflong\  
where $M$ is a perfect \csm and $M^{\pi}$ is as defined in \cref{def:stable-partitions}%
\fi.

\iflong
\begin{corollary}
  \PCSFM for bipartite graphs~$G'=(U\cup W, E')$ with strict preferences~$\sat$ is \WOH wrt.\ to the parameter~``$\fullnum-\fullymatched(M^{\pi})$'', where $M^{\pi}$ denotes a \csm of $(G',\sat)$ returned by \cref{alg:arbitrary-osm}. %
\end{corollary}
\begin{proof}
  We use the same reduction as the one given in the proof of \cref{th:PCSFM_hard}.
  Let $I=(G, k)$ denote an instance of \ISs with $G$ being a graph with vertex set~$V=\{v_1,v_2,\ldots, v_n\}$ and
  edge set~$E=\{e_1,e_2,\ldots, e_m\}$, 
  and let $I'=(G', \sat, \fullnum)$ be the constructed instance in the proof with $G'=(U\cup W, E')$ and $U=\{e_j, f_j, g_j\mid e_j \in E\}\cup \{u_i\mid v_i \in V\}\cup \{c_j,t_j\mid j\in [k]\}$,
  $W=\{e_j^{i}, e_j^{i'}, h_j\mid e_j=\{v_i, v_{i'}\} \text{ for some }e_j \in E\} \cup \{w_i\mid v_i \in V\}\cup \{s_j,a_j\mid j\in [k]\}$,
  and $\fullnum=|U|+|W|=6m+2n+4k$.
  
  Since \ISs parameterized by the independent set size~$k$ is \WOH~\cite{CyFoKoLoMaPiPiSa2015} and since we have shown in the proof of \cref{th:PCSFM_hard} that the reduction from \ISs runs in polynomial time and is correct,
  to show \WOH{ness} for \PCSFM,
  it suffices to show that $\fullnum-\fullymatched(M^\pi)\le 2k$,
  where $M^{\pi}$ is a matching computed by \cref{alg:arbitrary-osm} on input~$(G',\sat)$.
  For this, let \(N\) be an arbitrary stable integral matching for \(G', \sat\).
  Observe that \(N\) exists, because \(G'\) is bipartite.
  Moreover, since $\sat$ is strict, \citet{Tan1991} proved in his Proposition~2.1 that the stable matching~\(N\) induces a stable partition in which the transpositions (i.e., the edges) one-to-one correspond to the assignments in~\(N\).
  Moreover, by \cref{prop:singletons} the singletons in any two stable partitions are the same,
  and by \cref{ordinal:phase1correct}\eqref{lem:osm-equivalence} an agent is fully matched in $M^{\pi}$ if and only if she is a non-singleton.
  Thus, \(\fullymatched(M^\pi) = \fullymatched(N)\).
  Again, since $(G, \sat)$ has strict preferences, every stable integral matchings matches the same set of agents. 
  It thus suffices to prove that there exists a stable integral matching~\(M\) for \(G'\) such that \(\fullnum - \fullymatched(M) \leq 2k\).
  Since \(\fullnum = |U|+|W|\) we hence only need to show that that an arbitrary stable integral matching~\(M\) of $(G',\sat)$ (fully) matches all but at most \(2k\) agents.
  One way to show this is to check which agents must be matched under~$M$. 
  We show this by carrying out the propose-and-reject algorithm by \citet{GaleShapley1962}. %

  We let the agents in \(U\) propose to the agents in \(W\) with the following specific order of propositions:
  \begin{compactenum}[(1)]
  \item For each \(j \in [m]\):
  \begin{compactenum}[(a)]
    \item Agent~$f_j$  proposes to her most-preferred partner \(e^i_j\).
    \item Agent~$g_j$ proposes to her next most-preferred partner~$h_j$.
    \item Agent~$e_j$ proposes to her next most-preferred partner \(e^i_j\).
    Observe that \(e^i_j\) accepts, leaving \(f_j\) without partner.
    \item Agent \(f_j\) proposes to her next most-preferred partner \(e^{i'}_j\).
    \end{compactenum}
  Afterwards, each agent from the edge gadgets which is also from $W$ receives a proposal. 
    \item For each $i\in [n]$ agent~$u_i$  proposes to her most-preferred partner~$w_i$.
    \item For each $j\in [k]$ agent~$c_j$  proposes to her next most-preferred partner~$a_j$ (and gets accepted).
    \item For $j = k, k-1, \ldots, 1$, agent~$t_j$ proposes to her next most-preferred partner~$w_{n+k-j}$.
    Observe that \(w_{n+k-j}\) accepts, leaving \(u_{n+k-j}\) without partner.
  \end{compactenum}
  Observe that all agents in \(W\), except those from $\{s_j\mid j\in [k]\}$, receive at least one proposal from some agent from~$U$.
  Since in bipartite graph with strict preferences, an agent from $W$ never becomes unmatched once she got a proposal from some agent from~$U$,
  we know that at most $k$ agents from $W$ will remain unmatched under~$M$.
  Since $M$ is integral and $|U|=|W|$, at most $k$ agents remain unmatched under~$M$.
  Hence, at most \(2k\) agents will be unmatched under~$M$, as required.
\end{proof}
\fi

\section{Conclusion and outlook}
\looseness=-1%
\iflong
Motivated by the benefits of fractional matchings under preferences,
we studied three natural stability concepts (\linearstability, \ordinalstability, and \cardinalstability) and two optimization criteria, from structural and algorithmic perspectives.
We obtained a comprehensive picture of the algorithmic complexity of computing a stable fractional matching which maximizes the either number of fully matched agents or the social welfare, taking into account whether the preferences may contain ties and whether the underlying market is a marriage market or a roommates market.
\fi

\looseness=-1
We conclude with some challenges for future work.
First, it would be interesting to know whether the set of \cstable matchings has some form of lattice structure.
Second, studying optimal stable and fractional matchings using the framework of parameterized algorithmics~\cite{Nie06,CyFoKoLoMaPiPiSa2015} may provide more insights into the fine-grained complexity of the problem.
Promising parameters are the number of fully matched agents and the social welfare of the fractional matchings in the solution.
Finally, regarding preference restrictions~\cite{BreCheFinNie2020-spscSM-jaamas}, it would be interesting to know whether assuming a special preference structure can help in finding tractable cases for optimal fractional stable matchings.

\clearpage

\bibliography{sfm}

\end{document}

